\documentclass[final]{IEEEtran}

\usepackage[utf8]{inputenc} 
\usepackage[T1]{fontenc}    

\usepackage{hyperref}       
\usepackage{url}            
\usepackage{booktabs}       
\usepackage{amsfonts}       
\usepackage{nicefrac}       
\usepackage{microtype}      
\usepackage{xcolor}         

\usepackage{cite}
\usepackage{bm}
\usepackage{amsmath,amsthm}
\usepackage{extarrows}
\usepackage{amssymb}
\usepackage{graphicx}
\usepackage{subfig}
\usepackage{algorithmic,algorithm}
\usepackage{mathrsfs}
\usepackage{cmll}

\newcommand{\mb}{\bm}
\newcommand{\mr}{\mathrm}

\newcommand{\BE}{\begin{equation}}
\newcommand{\EE}{\end{equation}}
\newcommand{\BS}{\begin{subequations}}
\newcommand{\ES}{\end{subequations}}
\newcommand{\UT}{\mathsf{T}}   

\newcommand{\Mydef}{\overset{  \scriptscriptstyle \Delta  }{=}}

\newtheorem{theorem}{Theorem}

\newtheorem{definition}{Definition}
\newtheorem{remark}{Remark}
\newtheorem{lemma}{Lemma}

\newcommand{\mmse}{\mathsf{mmse}}
\newcommand{\snr}{\mathsf{snr}}
\newcommand{\var}{\mathrm{var}}

\newcommand{\Tr}{\mathsf{Tr}}

\newcommand{\Alg}{{GLM-EP}}

\newcommand{\perfect}{\delta^{\textsf{p}}}

\graphicspath{{figures/}}

\usepackage{mathtools}

\DeclarePairedDelimiter\floor{\lfloor}{\rfloor}

\title{Towards Designing Optimal Sensing Matrices for Generalized Linear Inverse Problems}

\author{Junjie Ma,~Ji Xu,~Arian Maleki

\thanks{This paper was presented at the Thirty-fifth Conference on Neural Information Processing Systems, NeurIPS 2021.}

\thanks{J.~Ma was with the Department of Statistics of Columbia University, New York, USA. He is now with the Institute of Computational Mathematics and Scientific/Engineering Computing, Academy of
Mathematics and Systems Science, Chinese Academy of Sciences, Beijing, China. (e-mail:majunjie@lsec.cc.ac.cn). J.~Xu was with the Department of Computer Science, Columbia University, New York, USA. (e-mail: jixu@cs.columbia.edu). A. Maleki is with the Department of Statistics, Columbia University, New York, USA. (e-mail: arian@stat.columbia.edu).}

\thanks{J. Ma was partially supported by National Natural Science Foundation of China (Grant NO. 12101592) and the Key Research Program of the
Chinese Academy of Sciences (Grant NO. XDPB15, NO. XDA27010102). A. Maleki was partially supported by a Google Faculty Research Award and the office of naval research (Grant N00014-23-1-2371).}
}

\begin{document}

\maketitle

\begin{abstract}

We consider an inverse problem $\bm{y}= f(\bm{Ax})$, where $\bm{x}\in\mathbb{R}^n$ is the signal of interest, $\bm{A}$ is the sensing matrix, $f$ is a nonlinear function and $\bm{y} \in \mathbb{R}^m$ is the measurement vector. In many applications, we have some level of freedom to design the sensing matrix $\bm{A}$, and in such circumstances we could optimize $\bm{A}$ to achieve better reconstruction performance. As a first step towards optimal design, it is important to understand the impact of the sensing matrix on the difficulty of recovering $\bm{x}$ from $\bm{y}$.

In this paper, we study the performance of one of the most successful recovery methods, i.e., the expectation propagation (EP) algorithm. We define a notion of spikiness for the spectrum of $\bm{A}$ and show the importance of this measure for the performance of EP. We show that whether a spikier spectrum can hurt or help the recovery performance depends on $f$. Based on our framework, we are able to show that, in phase-retrieval problems, matrices with spikier spectrums are better for EP, while in 1-bit compressed sensing problems, less spiky spectrums lead to better performance. Our results unify and substantially generalize existing results that compare Gaussian and orthogonal matrices, and provide a platform towards designing optimal sensing systems.
\end{abstract}

\section{Introduction}
\subsection{Problem statement and contributions}
Consider the problem of estimating a signal $\bm{x}\in\mathbb{R}^n$ from the nonlinear measurements:
\BE\label{Eqn:model}
\bm{y} = f\big(\bm{Ax}\big),
\EE
where $\bm{A}\in\mathbb{R}^{m\times n}$ is a sensing matrix and $f:\mathbb{R}\mapsto\mathcal{Y}$ is a function accounting for possible nonlinear effect of the measuring process. Here, the function $f(\cdot)$ is applied to $\bm{Ax}$ in a component-wise manner. The above model arises in many applications of signal processing \cite{Candes2013,boufounos20081,Rangan11}, communications \cite{wang2010optimized,bereyhi2019rls,Genzel19}, and machine learning \cite{shinzato2008perceptron,plan2016generalized}. For instance, the phase retrieval problem, which is a special case of \eqref{Eqn:model} with $f(z)=|z|$, has received significant interest in recent years \cite{Candes2013, CaLiSo15, ChenCandes17, Wang2016,Zhang2016reshaped,gao2020perturbed, Goldstein2016phasemax, Aahmani2016, Schniter2015,MXM19_IT,MiladIT}. In this paper, we assume that the signal is generic and prior information such as sparsity is not explored.

This work is motivated by the problem of optimizing the sensing matrix for the nonlinear inverse problem. Towards this goal, here we seek to understand the impact of the sensing matrix, or more specifically the spectrum of the sensing matrix, on the difficulty of recovering the signal $\bm{x}$ from its measurements $\bm{y}$. In many applications, one has certain level of freedom in designing the sensing matrix (e.g., transmitter design in communications or the masks used in phase retrieval application) and hence understanding the impact of the sensing matrix on the recovery algorithms is the first step toward the optimal design of such systems. Rather than studying the information theoretic limits, where the computational complexity of the recovery algorithm is ignored, we would like to study the impact of the spectrum of the sensing matrix on efficient algorithms that are used in applications. For this reason, we consider one of the most successful recovery algorithms that has received substantial attention in the last few years, i.e. expectation propagation (EP) \cite{Minka2001,opper2005} (referred to as {\Alg} in this paper\footnote{The name {\Alg} is chosen because the model \eqref{Eqn:model} is an instance of generalized linear models (GLM).}), and study the impact of the spectrum of the sensing matrix on the performance of this algorithm. The EP algorithm studied here is an instance of the algorithm introduced in \cite{fletcher2016,He2017} and is closely related to the orthogonal AMP (OAMP) \cite{Ma2016} and vector AMP (VAMP) \cite{Rangan17} algorithms (in that all these algorithms use divergence-free denoising functions \cite{Ma2016}).

Similar to the approximate message passing (AMP) algorithm \cite{DoMaMo09}, {{\Alg}} has two distinguishing features: (i) Its asymptotic performance could be characterized exactly by a simple dynamical system (with very few states) called the state evolution (SE). (ii) It is conjectured that AMP or {{\Alg}} achieve the optimal performance among polynomial time algorithms \cite{aubin2019committee,celentano2020estimation}. Based on the SE framework, we investigate the impact of the spectrum of the sensing matrix $\bm{A}$ on the performance of {{\Alg}}. It turns out that the ``spikiness'' (or conversely ``flatness'') of the spectrum  of the sensing matrix spectrum has a major impact on the performance of {{\Alg}}. To formalize this statement, we first define a measure of ``spikiness'' of the spectrum based on Lorenz partial order \cite{Arnold2018}. We show that whether the spikiness of the spectrum benefits or hurts GLM-EP depends on the choice of the nonlinear mapping $f$ (as well as the sampling ratio). For instance, spikier spectrums help the performance of phase retrieval problem (where $f(x) = |x|$) but hurt the performance of 1-bit compressed sensing (where $f(x)=\text{sign}(x)$). We will characterize the classes of functions on which spikiness hurts or helps {\Alg} based on the monotonicity of a function (which is related to the scalar minimum mean square error) that will be defined in this paper. As a byproduct of our studies, we will also show that when the spectrum is spiky enough, the number of measurements required by {{\Alg}} to achieve perfect recovery approaches the information theoretical lower bound.

\subsection{Related Work.}
Message passing algorithms \cite{DoMaMo09,Bayati&Montanari11,Bayati&Montanari12,Rangan11,Barbier2019,rush2018finite,sur2019modern,bu2020algorithmic,ZhouFanAOS,feng2021unifying,Minka2001,opper2005,ma2015performance,Ma2016,fletcher2016,schniter2016vector,Rangan17,KeigoV2,He2017,takeuchi2019unified,fletcher2017inference,takahashi2020macroscopic} have been used extensively for solving the estimation problems similar to the one we have in (\ref{Eqn:model}). As a result of such studies, it is known that partial orthogonal matrix is better than iid Gaussian matrix for noisy compressed sensing \cite{ma2015performance}, and the spectral methods for phase retrieval perform better with iid Gaussian sensing matrices than coded diffraction pattern matrices \cite{Lu17,Mondelli2017,ma2021spectral,dudeja2020analysis}. However, studying the impact of spectrum of the sensing matrix in the generality of our paper has not been done to the best of our knowledge. Recently, \cite{maillard2020phase} considered the phase retrieval problem and a sensing matrix which can be written as the product of Gaussian and another matrix. They reached the conclusion that the weak recovery threshold with this type of matrices can be made arbitrarily close to zero. As a special case of our results, we will also show that if we make the spectrum of the sensing matrix spiky, {\Alg} can reach the information theoretic lower bounds in the phase retrieval problem. \cite{aubin2020exact} considered the phase retrieval problem with generative priors in the form of deep neural networks with random weight matrices, and showed that it yields smaller statistical-to-algorithmic gap than sparse priors.

Another venue of research that is also related to our work is the derivation of the information theoretic limits for analog compression schemes. Analog compression framework was first introduced in \cite{Wu10,Wu12} for compressed sensing. It was shown in \cite{Wu10,Wu12} that the minimum number of measurements required for successful signal reconstruction in an information theoretic framework is related to the R\'{e}nyi information dimension of the signal distribution. \cite{Riegler15} studied the phase retrieval problem using the analog compression framework and proved that (real-valued) phase retrieval has the same fundamental limit as that of compressed sensing. In order to compare the performance of {{\Alg}}  on matrices with different spectral, we generalize the work of \cite{Wu10,Wu12} and \cite{Riegler15} and obtain information theoretic limit for our sensing model. Note that while we are using such information theoretic tools, the problem we are studying in this paper is fundamentally different from the one studied in \cite{Wu10,Wu12,Riegler15}. Here we are interested in the impact of the spectrum of the sensing matrix on the performance of {\Alg}, and information theoretic limits are mainly derived for comparison purposes (and evaluating the optimality of {\Alg}).

\subsection{Definitions}\label{Sec:Renyi}

In this section, we mention some definitions that will be frequently used throughout this paper. We first start with the R\'{e}nyi information dimension of a random variable.

\begin{definition}[Information dimension \cite{Renyi,Wu10}]\label{Def:Renyi_ID}
Let $X$ be a real-valued random variable, and $\langle X\rangle_M=\lfloor MX\rfloor/M$ be a quantization operator.\footnote{The notation $\lfloor z \rfloor$ denotes the largest integer that is smaller than $z$. } Suppose the following limit exists
\[
d(X)=\lim_{M\to\infty}\frac{H\left(\langle X\rangle_M\right)}{\log M},
\]
where $H(\cdot)$ is the entropy of a discrete random variable. The limit $d(X)$ is called the information dimension of $X$. Further, if $H(\lfloor X\rfloor)<\infty$, then $0\le d(X) \le 1$.
\end{definition}

As will be discussed later, $d(X)$ plays a critical role in the information theoretic lower bounds we derive for the recovery algorithms. The next lemma shows how $d(X)$ can be calculated for the simple distributions we observe in our applications.

\begin{lemma}[Information dimension of mixed distribution \cite{Renyi,Wu10}]\label{Lem:ID_property}
Let $X$ be a random variable such that $H(\lfloor X\rfloor)$ is finite. Suppose the distribution of $X$ can be represented as
\[
P_X = (1-\rho)P_{d}+\rho P_c,
\]
where $P_d$ is a discrete measure and $P_c$ is an absolutely continuous measure with respect to Lebesgue, and $0\le\rho\le1$. Then,
\[
d(X)=\rho.
\]
\end{lemma}

The minimum mean squared error (MMSE) defined below is an important notion in our analysis of \Alg.

\begin{definition}[MMSE for AWGN channel \cite{guo2005mutual}] Let $(Z,U)$ be a pair of random variables. The MMSE $\mmse(Z,\snr)$ and the conditional MMSE $\mmse(Z,\snr|U)$ given $U$ are defined as
\BE\label{Eqn:cond_MMSE}
\begin{split}
\mmse(Z,\snr)&=\mathbb{E}\left[ \left(Z-\mathbb{E}[Z|\sqrt{\snr}Z+N]\right)^2 \right],\\
\mmse(Z,\snr|U)&=\mathbb{E}\left[ \left(Z-\mathbb{E}[Z|\sqrt{\snr}Z+N,U]\right)^2 \right],
\end{split}
\EE
where $N\sim\mathcal{N}(0,1)$ is independent of $(Z,U)$, and the outer expectations are taken over all random variables involved.
\end{definition}

More properties of the MMSE function and the MMSE dimension are detailed in Appendix \ref{App:SE_maps}.

\section{Information-theoretic limit for signal recovery}

As we discussed earlier, our main objective is to evaluate the impact of the spectrum of the sensing matrix on the performance of {\Alg}. However, it is still useful to compare what {\Alg} achieves (for different spectral) with the information theoretic lower bounds, which we derive in this section.

\subsection{Assumptions}
Before we proceed to the technical part of the paper, let us review the assumptions we make throughout this paper.

\begin{itemize}
\item [(A.1)] The elements of $\bm{x}$ are independently drawn from $P_X$, which is an absolutely continuous distribution with respect to the Lebesgue measure. Further, $\mathbb{E}[X^2]=1$.
\item [(A.2)] Let the SVD of $\bm{A}\in\mathbb{R}^{m\times n}$ ($m\ge n$) be $\bm{A}=\bm{U\Sigma V}^\UT$, where $\bm{U}\in\mathbb{R}^{m\times m}$ and $\bm{V}\in\mathbb{R}^{n\times n}$ are independent Haar matrices, which are further independent of $\bm{\Sigma}$. Let $\{\sigma_i\}_{i=1}^n$ be the diagonal entries of $\bm{\Sigma}$ and $\Lambda_i\Mydef \sigma_i^2$. We assume that the empirical distribution of $\{\Lambda_i\}_{i=1}^n$ converges almost surely to a deterministic limit $P_\Lambda$ with a compact support bounded away from zero,  as $m,n\to\infty$ with $m/n\to\delta\in(1,\infty)$. Further, $\frac{1}{n}\sum_{i=1}^n\Lambda_i^2\overset{a.s.}{\longrightarrow}\mathbb{E}[\Lambda^2]<\infty$, where $\Lambda\sim P_\Lambda$. Without loss of generality, we assume $\mathbb{E}[\Lambda]=\delta$.
\item [(A.3)] $f:\mathbb{R}\mapsto\mathcal{Y}$ is a piecewise smooth function. Specifically, the domain $\mathbb{R}$ can be decomposed into $K\in\mathbb{N}_+$ non-overlapping intervals, and $f$ is continuously differentiable and monotonic on each sub-interval. Furthermore, we assume $|f^{-1}(y)|<\infty$ for all $y$, where $f^{-1}(y):=\{x: f(x)=y\}$, and $H(\lfloor f(Z)\rfloor)<\infty$ where $Z\sim\mathcal{N}(0,1)$.
\end{itemize}

Note that Assumption (A.2) is a standard assumption in theoretical analysis of {\Alg} \cite{Rangan17,KeigoV2,ZhouFanAOS}.
Furthermore, all the nonlinearities that we observe in applications satisfy Assumption (A.3). We consider generic signal and do not impose any structural assumption (e.g., sparsity). Finally, the independence assumption we have made in the prior of $\bm{x}$ is again standard in the literature of approximate message passing and expectation propagation \cite{Rangan11,fletcher2016,Bayati&Montanari11,weng2018overcoming,wang2020bridge}. One may relax this assumption and consider correlated signals at the expense of making more assumptions about the recovery algorithm.

\subsection{Perfect reconstruction in a noiseless setting}

In this section, we derive the information theoretic lower bound on the number of measurements required by a Lipschitz recovery scheme to achieve vanishing error probability. Note that the computational complexity of the recovery algorithm is {\em not} of any concern in these lower bounds. We will later compare our results for {\Alg} with these information theoretic lower bounds.

\begin{theorem}[Perfect reconstruction under Lipschitz decoding]\label{Lem:converse2}
Suppose Assumptions (A.1)-(A.3) hold. Suppose that there exists a limiting eigenvalue distribution $P_\Lambda$ and a Lipschitz continuous decoder $g:\mathcal{Y}^m\mapsto \mathbb{R}^n$ such that $\mathbb{P}\{\bm{x}\neq g(f(\bm{Ax}))\}\to0$ as $m,n\to\infty$ and $m/n\to\delta\in(1,\infty)$,
then necessarily we have
\BE\label{Eqn:converse2_condition}
\delta\ge\frac{1}{d(Y)}  
\EE
where $d(Y)$ is the information dimension of $Y:=f(Z)$, $Z\sim\mathcal{N}(0,1)$. Here, the error probability is taken with respect to both $\bm{x}$ and $\bm{A}$.
\end{theorem}

The proof of this result can be found in Appendix \ref{App:Low_bound}. The Lipschitz regularity condition on the decoder is natural for robustness considerations. It is interesting future work to study whether the converse result still holds with the Lipschitz condition removed or relaxed.

Intuitively speaking, $d(Y)\in[0,1]$ may be interpreted as a measurement discount factor and the total number of effective measurements is $m\cdot d(Y)$\footnote{In the rest of this paper, we will use $Y$ to denote the random variable $f(Z)$, where $Z\sim\mathcal{N}(0,1)$.} .
{
\begin{remark}[1-bit CS]
For the 1-bit compressed sensing (CS) problem, we have $f(z)=\mr{sign}(z)$ and $d(Y)=0$. In this case, the condition $\delta\ge1/d(Y)=+\infty$ implies that perfect recovery is impossible in the regime $m,n\to\infty$ and $m/n\to\delta\in(1,\infty)$. Notice that our result does not contradict with existing 1-bit CS  results \cite{dirksen2017one,plan2016generalized}. For instance, \cite{dirksen2017one} analyzes the number of random measurements required by a convex minimization algorithm to achieve a non-zero target distortion $\rho$, and the bound blows up to infinity as $\rho\to0$.
\end{remark}
}

\subsection{Stable reconstruction in the noisy setting}\label{Sec:noisy_converse_1}
Theorem \ref{Lem:converse2} focuses on signal reconstruction for model \eqref{Eqn:model} without any noise. For practical considerations, it is desirable to make sure that a small amount of measurement noise does not cause major performance degradation. In this paper, we consider the following noisy model\footnote{Other types of noisy models are possible, e.g., $\bm{y}=f(\bm{Ax})+\bm{w}$. For such noisy models, we expect that the fundamental noise sensitivity result in Theorem \ref{The:noise_opt} still holds, but the noise sensitivity result of {\Alg} in Theorem \ref{Lem:noise_sensitivity} may require new analysis. Extending our results to these models is beyond the aim of the current paper.}
\BE\label{Eqn:model_noisy}
\bm{y}=f(\bm{Ax}+\bm{w}),
\EE
where $\bm{w}\sim\mathcal{N}(\mathbf{0},\sigma_w^2\bm{I})$ is independent of $\bm{A}$ and $\bm{x}$. Define the \textit{noise sensitivity} \cite{DMM_11,Wu12} of the minimum mean square error (MMSE) estimator by
\BE
M^\ast(X,f,\Lambda,\delta)\Mydef  \sup_{\sigma_w}\, \limsup_{n\to\infty}\frac{\frac{1}{n}\mmse(\bm{x}|\bm{y},\bm{A})}{\sigma^2_w},
\EE
where $\mmse(\bm{x}|\bm{y},\bm{A})\Mydef \mathbb{E}\big[\left(\bm{x}-\mathbb{E}[\bm{x}|\bm{y},\bm{A}]\right)^2\big]$ is the MMSE of estimating $\bm{x}$ from $\bm{y}$. In the above definition, the limit $n\to\infty$ is understood as $n\to\infty$ and $m/n\to\delta$. Theorem \ref{The:noise_opt} below shows that to achieve bounded noise sensitivity, one needs $\delta\ge1/d(Y)$, the same necessary condition for achieving vanishing error probability in the noiseless setting. Its proof can be found in Appendix \ref{App:noise_opt}.

\begin{theorem}[Noise sensitivity]\label{The:noise_opt}
Suppose Assumptions (A.1)-(A.3) hold. Additionally, assume $\bm{x}\sim\mathcal{N}(\mathbf{0},\bm{I})$. A necessary condition for achieving bounded noise sensitivity, namely $M^\ast(X,f,\Lambda,\delta)<\infty$, is $\delta\ge1/d(Y)$.
\end{theorem}

Note that the same fundamental limit $1/d(Y)$ appears for both noiseless recovery (Theorem \ref{Lem:converse2}) and noise sensitivity (Theorem \ref{The:noise_opt}) converse results. The situation is similar to the pioneering work \cite{Wu10} which established the information theoretical limits for compressed sensing. 

We would also like to mention that the asymptotic MMSE (and so the noise sensitivity) may be calculated using the replica method \cite{maillard2019high}. However, since the correctness of the replica predictions has not been proved for the current setting, we do not pursue it in this paper and leave it as possible future work.

\subsection{Discussion of Theorems \ref{Lem:converse2} and \ref{The:noise_opt} }\label{Sec:MMSE_dim_A}

Theorems \ref{Lem:converse2} and \ref{The:noise_opt} show that the quantity $d(Y)$ determines the fundamental limit for signal recovery from the nonlinear model \eqref{Eqn:model}. Notice that $Y:=f(Z)$ is a mixed discrete-continuous distribution (by Assumption (A.3)), where the discrete component in $Y$ corresponds to ``flat'' sections of $f$; see Figure \ref{Fig:finite_f} for illustration. According to Lemma \ref{Lem:ID_property}, $d(Y)$ is simply the weight in the continuous component of the distribution of $Y$, which is the probability of $Z\sim\mathcal{N}(0,1)$ falling into the non-flat sections of $f$. For illustration, Figure \ref{Fig:finite_f} shows three representative examples of $f$.

\textbf{Type I:} $f$ is a piece-wise smooth function without flat sections; see the left panel of Figure~\ref{Fig:finite_f} for illustration. This type of functions includes the absolute value function $f=|z|$, which appears in phase retrieval problems. For such functions, $f(Z)$ has an absolutely continuous distribution when $Z\sim\mathcal{N}(0,1)$, and hence $d(Y)=1$ according to Lemma \ref{Lem:ID_property}.

\textbf{Type II:}  $f$ consists of purely flat sections. A special case is the quantization function. Clearly, $Y$ has a discrete distribution and $d(Y)=0$.

\textbf{Type III:} $f$ consists of both flat and non-flat sections, e.g., the function shown on the right panel of Figure~\ref{Fig:finite_f}. Such scenarios happen, for instance, when sensors saturate in the phase retrieval application. In this case, $Y$ has a mixed discrete-continuous distribution and $0<d(Y)<1$.

\begin{figure}[htbp]
\begin{center}
\includegraphics[width=\linewidth]{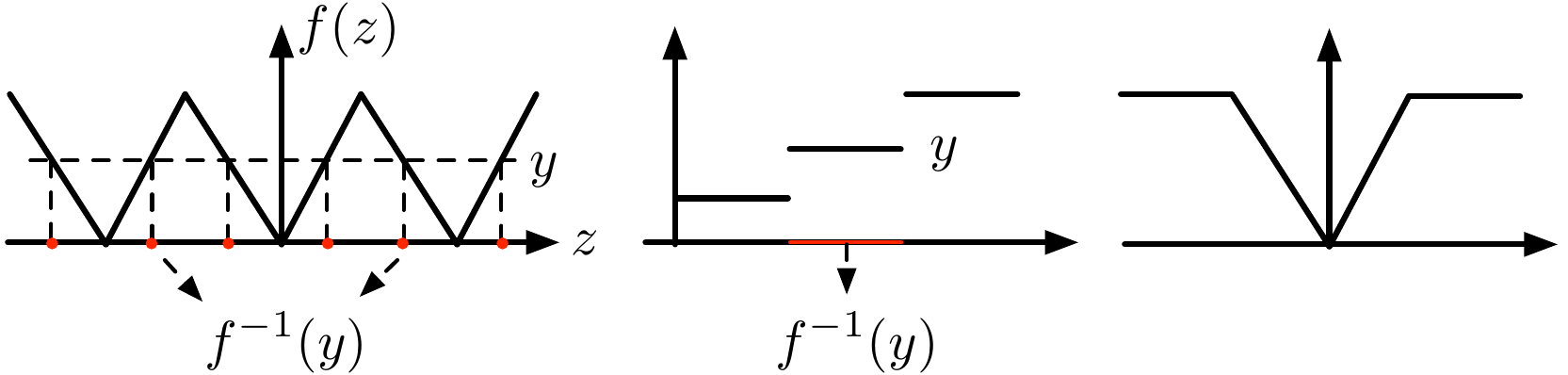}
\caption{Three types of $f$. Left: $d(Y)=1$. Center: $d(Y)=0$. Right: $0<d(Y)<1$.}\label{Fig:finite_f}
\end{center}
\end{figure}

\section{{\Alg} algorithm and performance analysis}

In this section, we introduce an expectation propagation (EP) \cite{Minka2001,opper2005} type algorithm, referred to as {\Alg}, for solving our nonlinear inverse problem and derive its state evolution (SE). We then study the impact of the spectrum of the sensing matrix on the performance of this algorithm.

\subsection{Summary of {\Alg}}

The {\Alg} algorithm is summarized below. We use superscripts to represent iteration indices, and subscripts `$l$' and `$r$' to distinguish different variables.

\noindent\textbf{Initialization:} $\bm{z}_r^{-1}=\mathbf{0}$, $v_r^{-1}=1$. For $t=0,\ldots$, execute the following steps iteratively:\vspace{-2pt}
\BS\label{Eqn:EP_Gaussian}
\begin{align}
\bm{z}_l^{t}&= \frac{1}{1- \left\langle \eta_z' (\bm{z}_r^{t-1},\bm{y},v_r^{t-1})\right\rangle}\cdot \\
&\ \Big(\eta_z(\bm{z}_r^{t-1},\bm{y},v_r^{t-1})-  \left\langle \eta_z' (\bm{z}_r^{t-1},\bm{y},v_r^{t-1})\right\rangle \cdot  \bm{z}_r^{t-1}\Big) ,\\
v_l^{t}&=v^{t-1}_r \cdot \frac{\left\langle \eta_z' (\bm{z}_r^{t-1},\bm{y},v_r^{t-1})\right\rangle}{ 1-\left\langle \eta_z' (\bm{z}_r^{t-1},\bm{y},v_r^{t-1})\right\rangle},\\
{\bm{R}}^{t}&\Mydef \bm{A} \left(v_l^t\bm{I}+\mb{A}^{\UT}\bm{A} \right)^{-1}\bm{A}^\UT,\\
\bm{z}_r^{t} &= \frac{1}{1-\frac{1}{m}\Tr(\bm{R}^{t})}\cdot\Big(\bm{R}^{t}-\frac{1}{m}\Tr(\bm{R}^{t})\cdot\bm{I} \Big)\cdot{\bm{z}}_l^t,\\
v_r^{t}&=v_l^{t}\cdot \frac{\frac{1}{m}\Tr(\bm{R}^{t})}{1-\frac{1}{m}\Tr(\bm{R}^{t})},
\end{align}
where $\eta_z$ is defined by
\BE\label{Eqn:eta_z_def}
\eta_z(z_r,y,v)\Mydef \frac{\int_{f^{-1}(y)}u\cdot \mathcal{N}(u;z_r,v)du}{\int_{f^{-1}(y)}\mathcal{N}(u;z_r,v)du},
\EE
\ES
and $\eta'_z$ denotes the derivative of $\eta_z$ with respect to the first argument. Here, $\mathcal{N}(x;m,v):=\frac{1}{\sqrt{2\pi v}}\exp(-(x-m)^2/(2v))$ denotes the Gaussian pdf function. When $f^{-1}(y):=\{z: f(z)=y\}$ is a discrete set, the integration in the above formula is simply replaced by a summation.
\\[3pt]
\noindent\textbf{Output:} $\hat{\bm{x}}_{\mr{out}}^t=v_l^t(\bm{I}+v_l^t\bm{A}^\UT\bm{A})^{-1}\bm{A}^\UT\bm{z}_l^t$.\vspace{5pt}

In the above descriptions of the algorithm, we adopted the convention commonly used in the AMP literature: $\eta_z(\bm{z}_r,\bm{y},v)$ denotes a vector with elements obtained by applying the scalar function $\eta_z$ to the corresponding elements of $\bm{z}_r$ and $\bm{y}$, and $\langle\cdot\rangle$ denotes the empirical mean of a vector.

\subsection{Asymptotic analysis}\label{Sec:SE_analysis}

The asymptotic performance of {\Alg} could be described by two scalar sequences $\{V_l^t,V_r^t\}_{t\ge0}$, defined recursively by
\BS\label{Eqn:SE_Gaussian}
\begin{align}
V_l^t &=\Bigg({ \frac{1}{\mmse_z\left(V_r^{t-1}\right)}-\frac{1}{V_r^{t-1}}}\Bigg)^{-1}\Mydef \phi(V_r^{t-1}), \label{Eqn:SE_Gaussian_a}\\
V_r^{t}&=\Bigg({\frac{1}{\frac{1}{\delta}\cdot\mathbb{E}\left[ \frac{ V_l^t\Lambda}{V_l^t + \Lambda} \right]}-\frac{1}{V_l^t}}\Bigg)^{-1}\Mydef \Phi(V_l^t),\label{Eqn:SE_Gaussian_b}
\end{align}
where $\left.V_r^{t-1}\right|_{t=0}=1$, $\mmse_z(V_r) \Mydef  \mmse\left(Z,V_r^{-1}-1|Y\right)$, and the expectation in \eqref{Eqn:SE_Gaussian_b} is w.r.t. the limiting eigenvalue distribution of $\bm{A}^\UT\bm{A}$. (Recall that $\mmse(Z,\snr|U)$ denotes a conditional MMSE; see \eqref{Eqn:cond_MMSE}). Equations \eqref{Eqn:SE_Gaussian_a} and \eqref{Eqn:SE_Gaussian_b} are known as the state evolution (SE) for {\Alg}. More properties of the functions $\phi(\cdot)$ and $\Phi(\cdot)$ are given in Appendix \ref{App:SE_maps}.

Roughly speaking, the deterministic sequences $\{V_l^t,V_r^t\}_{t\ge0}$ are expected to be accurate predictions of $\{v_l^t,v_r^t\}_{t\ge0}$ (which are generated by {\Alg}) asymptotically. We will formalize this claim later. Further, we will show that the per coordinate MSE of $\bm{x}_{\mr{out}}^t$ (see Lemma \ref{Eqn:LFD2} below) is characterized by
\BE\label{Eqn:MSE_lambda_def}
\mathsf{MSE}_\Lambda(V_l^t)\Mydef \mathbb{E}\left[ \frac{V_l^t}{V_l^t + \Lambda} \right].
\EE
The subscript emphasizes the fact that the MSE depends on the limiting eigenvalue distribution $P_\Lambda$.
\ES

Lemma \ref{Eqn:LFD2} below gives a formal statement of the accuracy of SE, and its proof is mainly based on that of \cite[Theorem 1]{fletcher2017inference}. Note that \cite{fletcher2017inference} requires both the continuity of $f$ and $\eta_z$. Similar to the analysis of the AMP for rotationally-invariant matrix in \cite{ZhouFanAOS,venkataramanan2021estimation} we expect the state evolution to hold if the composite function $\tilde{\eta}(z_r,z,v):=\eta_z(z_r,f(z),v)$ is Lipschitz-continuous with respect to the first two arguments except for sets of zero measure. Such a result would be general enough to cover many interesting applications, e.g., GLM-EP for 1-bit CS. However, a complete proof requires careful analysis and we leave it as possible future work. 

In this work, we employ a simple smoothing technique to get rid of the Lipschitz-continuity requirement on $\eta_z$. (Note that we still require the acquisition function $f$ to be Lipschitz-continuous.) Specifically, we construct a new algorithm, called {GLM-EP-app} hereafter, which satisfies the requirements of \cite{fletcher2017inference}. This allows us to use SE for predicting the performance of this algorithm. {GLM-EP-app} uses the following iterations:
\BS\label{Eqn:GLM_EP_app}
\begin{align}
\bm{z}_l^{t}&= C_t \Big(\tilde{\eta}_z(\bm{z}_r^{t-1},\bm{y},V_r^{t-1})-   \mathbb{E}\left[\tilde{\eta}_z' (Z_r^{t-1},Y,V_r^{t-1})\right]   \bm{z}_r^{t-1}\Big) , \label{Eqn:GLM_EP_app_a}\\
\bm{z}_r^{t} &= \frac{1}{1-\frac{1}{m}\Tr(\bm{R}^{t})}\cdot\Big(\bm{R}^{t}-\frac{1}{m}\Tr(\bm{R}^{t})\cdot\bm{I} \Big)\cdot{\bm{z}}_l^t, \label{Eqn:GLM_EP_app_b}
\end{align}
\ES
where $\tilde{\eta}$ is a function for which $\mathbb{E}\left[\tilde{\eta}_z' (Z_r^{t-1},Y,V_r^{t-1})\right]$ exists, ${\bm{R}}^{t}\Mydef \bm{A} \left(V_l^t\bm{I}+\mb{A}^{\UT}\bm{A} \right)^{-1}\bm{A}^\UT$, and $\{C_t\}$ is a sequence of fixed numbers. The choices we choose for $\tilde{\eta}$ and $C_t$ (to make them close enough to {\Alg}) is discussed in the proof of Lemma \ref{Eqn:LFD2}. Finally, similar to {\Alg} the output of  {GLM-EP-app} is given by $$\hat{\bm{x}}_{\mr{out}}^t=V_l^t(\bm{I}+V_l^t\bm{A}^\UT\bm{A})^{-1}\bm{A}^\UT\bm{z}_l^t
.$$

Lemma \ref{Eqn:LFD2} shows that the performance of {GLM-EP-app} could be arbitrarily close to the SE prediction. The details of the proof can be found in Appendix \ref{App:Smoothing}.

\begin{lemma}\label{Eqn:LFD2}
Suppose Assumptions (A.1)-(A.3) hold. Additionally, assume $f:\mathbb{R}\mapsto\mathcal{Y}$ to be Lipschitz continuous. Let $\{V_l^t,V_r^t\}_{t\ge0}$ be generated according to \eqref{Eqn:SE_Gaussian}. For any $\epsilon>0$, there exists $\tilde{\eta}_z$ and $\{C_t\}_{t\ge0}$ such that $\hat{\bm{x}}_{\mr{out}}^t$ of \textsf{GLM-EP-app} satisfies
\BE
\mathsf{MSE}_{\Lambda}(V_l^t)-\epsilon \leq \frac{1}{m}\left\| \hat{\bm{x}}_{\mr{out}}^t -\bm{x}\right\|^2  < \mathsf{MSE}_{\Lambda}(V_l^t)+\epsilon,
\EE
almost surely as $m,n\to\infty$ with $m/n\to \delta\in(1,\infty)$, where $\mathsf{MSE}_\Lambda$ is defined in \eqref{Eqn:MSE_lambda_def}.
\end{lemma}

Note that we still require the acquisition function $f$ to be Lipschitz-continuous. Hence, Lemma \ref{Eqn:LFD2} does not apply to 1-bit CS. Nevertheless, we expect the state evolution of GLM-EP holds for 1-bit CS as well.

According to Lemma \ref{Eqn:LFD2}, the asymptotic MSE of {GLM-EP-app} in the large system limit as $t \rightarrow \infty$ can be obtained from the limiting value of $V_r^t$ (or $V_l^t$). Since this quantity is of particular importance to us, we will characterize it in the following lemma.

\begin{lemma}[MSE performance]\label{Lem:MSE_SE}
Suppose $\delta>1$. Define $V_r^{\star}$ by
\BE \label{Eqn:gamma_leftFP}
V_r^{\star}\Mydef\inf\Big\{v\in[0,1]\,:\,P(v_r)>0, \forall v_r\in[v,1]\Big\}.
\EE
where
\BE\label{Eqn:Pv_def}
P(v_r)\Mydef \mathbb{E}\left[ \frac{\phi(v_r)}{\phi(v_r) + \Lambda} \right] - \underbrace{\left[ 1-\delta\left(1-\frac{\mmse_z(v_r)}{v_r}\right)\right]}_{g(v_r)}.
\EE
In case $P(1)=0$, we define $V_r^{\star}=1$.
Let $\{V_l^t,V_r^t\}_{t\ge0}$ be sequences generated according to \eqref{Eqn:SE_Gaussian} with $\left.V_r^{t-1}\right|_{t=0}=1$. We have
\[
\lim_{t\to\infty} V_r^{t}=V_r^{\star}.
\]
Further, the final MSE is given by $\mathsf{MSE}^\star_\Lambda\Mydef \mathsf{MSE}_\Lambda(\phi(V_r^\star))$, where $\phi$ is defined in \eqref{Eqn:SE_Gaussian_a}.
\end{lemma}
The proof of this lemma can be found in Appendix \ref{App:MSE_SE}. A direct consequence of Lemma \ref{Lem:MSE_SE} is the perfect reconstruction condition stated in Lemma \ref{Lem:Gaussian_condition} below.

\begin{lemma}[Perfect reconstruction condition]\label{Lem:Gaussian_condition}
Let $\{V_l^t,V_r^t\}_{t\ge0}$ be a sequence generated through \eqref{Eqn:SE_Gaussian} with $\left.V_r^{t-1}\right|_{t=0}=1$, and let $\mathsf{MSE}^\star_\Lambda$ be the final MSE. Then, the following hold.
\begin{itemize}
\item [(i)] $\mathsf{MSE}^\star_\Lambda=0$ if and only if
 \begin{align}\label{Eqn:Gaussian_perfect2}
P(v_r) > 0, \quad\forall v_r\in(0,1],
\end{align}
where $P(v_r)$ is defined in \eqref{Eqn:Pv_def}. 
\item [(ii)] If there exists a spectrum $P_\Lambda$ such that $\mathsf{MSE}^\star_\Lambda=0$, then $\delta\ge1/d(Y)$. Conversely, if $\delta>1/d(Y)$ and $\mmse_z(1)<1$, then there exists a spectrum $P_\Lambda$ such that $\mathsf{MSE}^\star_\Lambda=0$.
\end{itemize}
\end{lemma}

The proofs of  Lemma \ref{Lem:Gaussian_condition} can be found in Appendix \ref{App:Proof_Condition}. It should be noted that to approach the lower bound using {\Alg}, the function $f$ has to satisfy the requirement $\mmse_z(1)<1$. This is a regularity condition that makes sure the SE equation \eqref{Eqn:SE_Gaussian} does not have a undesirable fixed point at $V_r=1$. Notably, this condition does not hold when $f$ is an even function (e.g., $f(z)=|z|$). For such functions, the achievability result is still valid if there is a small amount of side information about the signal. Alternatively, one might consider using the spectral method to initialize the {\Alg} algorithm \cite{MXM19_IT,Mondelli2021PCAIF,Mondelli2021ApproximateMP}.

\section{Impact of sensing matrix spectrum}\label{Sec:impact}

In this section, we use Lemmas \ref{Lem:MSE_SE} and \ref{Lem:Gaussian_condition} to study the impact of the sensing matrix on the MSE performance of GLM-EP-app. Before presenting our detailed analysis, we first discuss the so-called Lorenz order that compares the ``spikiness'' of different distributions.

\subsection{A measure of spikiness of distributions}
A natural tool to compare the spikiness of the distributions of two non-negative random variables is Lorenz partial order \cite{Arnold2018}. (Since it is a partial order, there exist distributions that are not comparable in the Lorenz sense.) Lorenz order is widely used to characterize wealth inequality, and is closely related to majorization, a tool that has been extensively studied for transceiver design in communication systems \cite{Palomar2003}.

\begin{definition}[Lorenz partial order \cite{Arnold2018}]\label{Def:Lorenz}
Consider a nonnegative random variable with cumulative density function $F(x)$. Let $F^{-1}(y)$ be the quantile function defined by
\BE
F^{-1}(y)=\sup\{x: F(x)\le y\},\quad 0<y<1.
\EE
The Lorenz curve corresponding to $F(x)$ is defined by
\[
L(u)= \frac{\int_0^u F^{-1}(y)dy}{\int_0^1 F^{-1}(y)dy}, \quad 0\le u\le1.
\]
Let $X$ and $Y$ be two nonnegative random variables, and $L_X(u)$ and $L_Y(u)$ be the corresponding Lorenz curves. We say $X$ is less spiky than $Y$ in the Lorenz sense, denoted as $X\preceq_L Y$, if $L_X(u)\ge L_Y(u)$ for every $u\in[0,1]$. Conversely, $X\succeq_L Y$ if $L_X(u)\le L_Y(u)$ for every $u\in[0,1]$.
\end{definition}

\begin{figure}[htbp]
\begin{center}
\includegraphics[width=.4\textwidth]{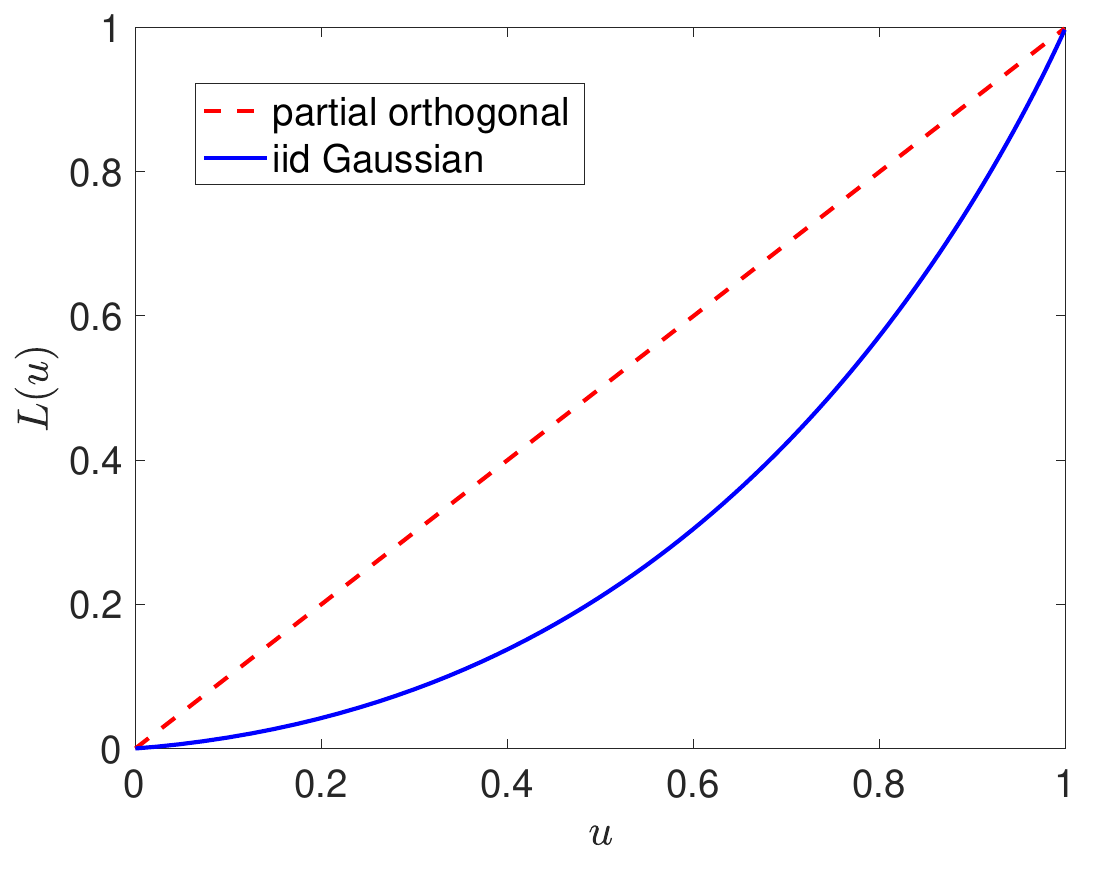}
\caption{Lorenz curves corresponding to the eigenvalue distributions of an i.i.d. Gaussian matrix and a partial orthogonal matrix. $\delta=2$.}\label{Fig:Lorenz}
\end{center}
\end{figure}

The use of Lorenz order to measure spikeness of distribution is very natural. In the context of income inequality, Lorenz curve has the following interpretation -- the poorest $100\times u$ percentage of the population contribute to $100\times L_X(u)$ percentage of the total wealth. Therefore, a larger $L_X(u)$ represents a more equal (or less spiky) income distribution. Fig.~\ref{Fig:Lorenz} demonstrates the Lorenz curves for the uniform distribution (corresponding to the spectrum of a column-orthonormal matrix) and the Marchenko-Pastur distribution (corresponding to the spectrum of an i.i.d. Gaussian matrix).

An important property of Lorenz partial ordering is the following.
\begin{lemma}[\cite{Arnold2018}]\label{Lem:convex_Lorenz}
Suppose $X\ge0$, $Y\ge0$ and $\mathbb{E}[X]=\mathbb{E}[Y]$. We have $X\preceq_L Y$ if and only if $\mathbb{E}[h(X)]\le \mathbb{E}[h(Y)]$ for every continuous convex function $h:\mathbb{R}_+\to\mathbb{R}$.
\end{lemma}

\subsection{Impact on MSE}\label{Sec:impact_MSE}
Let $\Lambda_1\sim P_{\Lambda_1}$ and $\Lambda_2\sim P_{\Lambda_2}$ be two limiting eigenvalue distributions of $\bm{A}^\UT\bm{A}$. Let $V_{\Lambda_1}^\star$ and $V_{\Lambda_2}^\star$ denote the corresponding limiting values of $V_r^t$ (as $t \rightarrow \infty$) in \eqref{Eqn:SE_Gaussian} (proving that the iterations \eqref{Eqn:SE_Gaussian} converge to a fixed point is straightforward). The associated MSEs, denoted as $\mathsf{MSE}_{\Lambda_1}^\star$ and $\mathsf{MSE}_{\Lambda_2}^\star$, can be compared according to the following lemma. See Appendix \ref{App:Proof_Impact_MSE} for its proof.\vspace{3pt}

\begin{lemma}\label{Lem:impact}
Let $\delta>1$. Suppose $P_{\Lambda_1}$ is more spiky than $P_{\Lambda_2}$ in the Lorenz sense, i.e., $\Lambda_1\succeq_L \Lambda_2$. Define
\BE\label{Eqn:g_def}
G(v_r;\delta)\Mydef \max\big(g(v_r),\,0\big),\quad \forall v_r\in[0,1],
\EE
where $g(\cdot)$ is defined in \eqref{Eqn:Pv_def}.
We have
\begin{itemize}
\item If $G(v_r;\delta)$ is non-decreasing on $v_r\in[0,1]$, then $\mathsf{MSE}_{\Lambda_1}^\star\le\mathsf{MSE}_{\Lambda_2}^\star$;
\item If $G(v_r;\delta)$ is non-increasing on $v_r\in[0,1]$, then $\mathsf{MSE}_{\Lambda_1}^\star\ge\mathsf{MSE}_{\Lambda_2}^\star$;
\item If $G(v_r;\delta)$ is not monotonic, then the comparison of $\mathsf{MSE}_{\Lambda_1}^\star$ and $\mathsf{MSE}_{\Lambda_2}^\star$ is not definite.
\end{itemize}
\end{lemma}

\begin{figure}[htbp]
\begin{center}
\includegraphics[width=.5\textwidth]{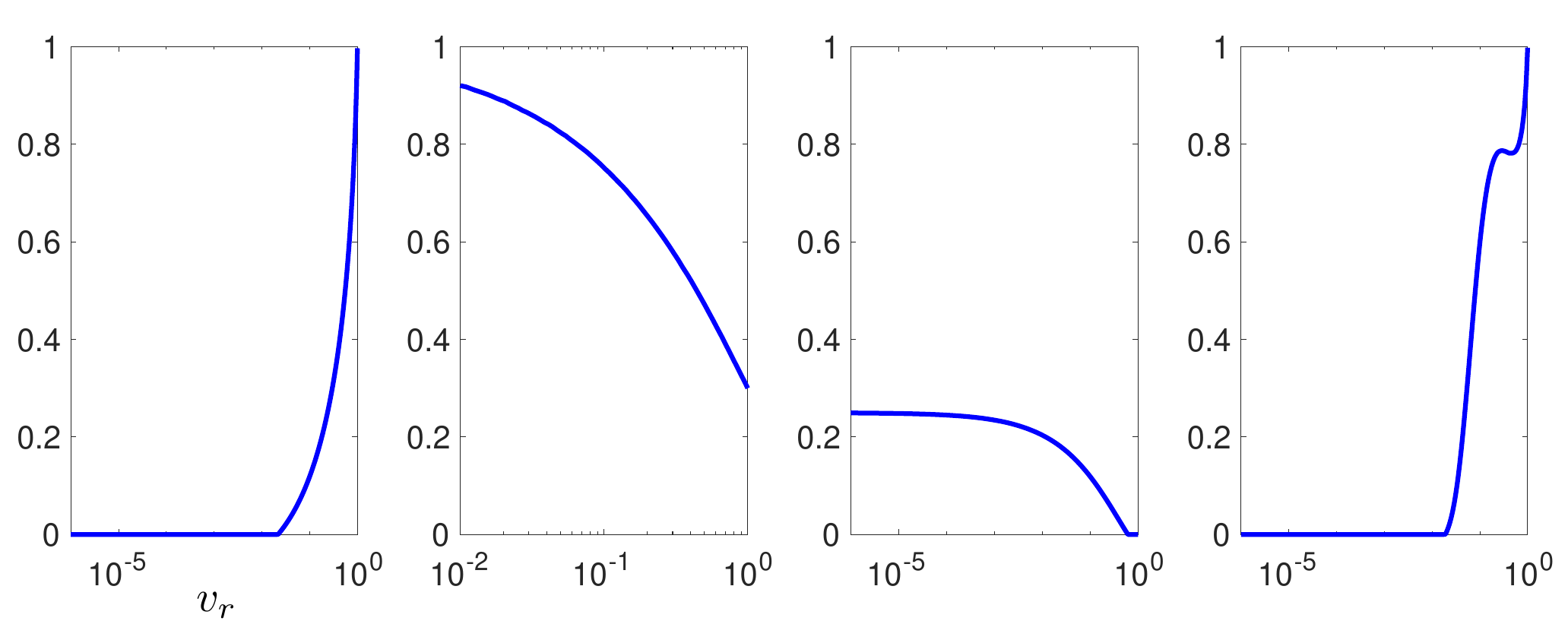}
\caption{Illustration of $G(v_r;\delta)$ for four choices of $f$. From left to right: $f(z)=|z|$, $f(z)=\max(-1,\min(z,1))$, $f(z)=\mr{sign}(z)$, $f(z)=|z|\mathbf{1}_{|z|<1}+(|z|-1)\mathbf{1}_{|z|\ge1}$. $\delta=1.1$.}\label{Fig:recon_functions}
\end{center}
\end{figure}
{
\begin{remark}
Notice that the function $G(v_r;\delta)$ depends on the sampling ratio $\delta$, as can be seen from the definitions in \eqref{Eqn:g_def} and \eqref{Eqn:Pv_def}. (To keep notation light, we do not make such dependency explicit for $g(v_r)$ and $P(v_r)$ though.) Hence, for a given $f$, the monotonicity of $G(v_r;\delta)$ could change as $\delta$ varies.
\end{remark}
}

Lemma \ref{Lem:impact} shows that the impact of the spectrum on the final MSE performance of {GLM-EP-app} depends on the monotonicity of the function $G(v_r;\delta)$ (which further depends on $f$). For a given $f$ and $\delta$, the function $G(v_r;\delta)$ can be numerically computed and its monotonicity can be easily checked. Below are four examples of $f$, corresponding to each of the cases discussed in Lemma \ref{Lem:impact}; see Fig.~ \ref{Fig:recon_functions}.

\textbf{Example 1:} It can be shown that that $G(v_r;\delta)$ of the following $f$ is non-decreasing for all $\delta>1$:
\[
f(z)=|z|.
\]
For such $f$, spiky spectrums are beneficial for MSE performance.

\textbf{Example 2:} The $G(v_r;\delta)$ of the following function is non-increasing for all $\delta>1$:
\[
f(z)=\mr{sign}(z).
\]
For this example, flatter spectrums are better.

\textbf{Example 3:} The $G(v_r;\delta)$ of the following function is non-increasing for all $\delta>1$:
\[
f(z)=\max(-1,\min(z,1)).
\]
For this example, flatter spectrums are better.

\textbf{Example 4:} Consider the following function
\BE\label{Eqn:f_example3}
f(z)=
\begin{cases}
|z|, & \text{if}\ |z|<1\\
|z|-1, & \text{if}\ |z|\ge1.
\end{cases}
\EE
In this case, $G(v_r;\delta)$ is not monotonic and the impact of the spectrum is not solely determined by the Lorenz order.

\subsection{Impact of spectrum on perfect recovery threshold}
We have shown that the impact of the spikiness of the spectrum on the {MSE performance} is related to the monotonicity of the function $G(v_r;\delta)$ which depends on the nonlinear function $f$ and the sampling ratio $\delta$. In this section, we will show that if our goal is to \textit{minimize the number of measurements required for perfect reconstruction}, then more spiky spectrum benefit GLM-EP-app \textit{for all $f$}. Furthermore, the information theoretic lower bound $\perfect_{\mr{opt}}$ can be reached (as close as we wish) if the spectrum of $\bm{A}$ is spiky enough. Theorem \ref{The:1}, whose proof can be found in Appendix \ref{App_Proof_Achie}, summarizes the above discussions.

\begin{theorem}\label{The:1}
 For a given nonlinearity $f$ and eigenvalue distribution $P_\Lambda$, let $\delta_{\Lambda}^{\mathsf{alg}}$ be the minimum $\delta$ required for perfectly recovering the signal, i.e.,
\BE
\delta_{\Lambda}^{\mathsf{alg}}\Mydef \inf\left\{\delta\,:\, \mathsf{MSE}_{\Lambda}^{\star}=0\right\},
\EE
where $\mathsf{MSE}_{\Lambda}^{\star}$ is defined in Lemma \ref{Lem:MSE_SE}. Let $\Lambda_1$ and $\Lambda_2$ denote two limiting eigenvalue distributions and $\delta_{\Lambda_1}^{\mathsf{alg}}$ and $\delta_{\Lambda_2}^{\mathsf{alg}}$ the corresponding thresholds for perfect reconstruction. We have $\delta_{\Lambda_1}^{\mathsf{alg}}\le \delta_{\Lambda_2}^{\mathsf{alg}}$ if $\Lambda_1\succeq_L \Lambda_2$.
\end{theorem}

\subsection{Noise Sensitivity Analysis}\label{Sec:noise}
Up to now, we only studied the performance of GLM-EP-app in the noiseless setting. In practice, it is also important to guarantee that the reconstruction performance does not significantly worsen due to the presence of a small amount of measurement noise. We consider the noisy model in \eqref{Eqn:model_noisy}. GLM-EP-app remains unchanged except that $\eta_z$ is replaced by a posterior mean estimator that takes the noise effect into consideration.

The following lemma analyzes the MSE performance of {GLM-EP-app} in the high SNR regime, and shows that its reconstruction is stable when $\delta$ is larger than the corresponding perfect recovery threshold. The proof of Lemma \ref{Lem:noise_sensitivity} and other details about GLM-EP-app in the noisy setting are provided in Section \ref{App:proof_noise}.

\begin{lemma}\label{Lem:noise_sensitivity}
Assume $d(Y\neq0$. Let $\delta>\delta_{\Lambda}^{\mathsf{alg}}$, where $\delta_{\Lambda}^{\mathsf{alg}}$ is defined in Theorem \ref{The:1}. Let $\mathsf{MSE}_\Lambda^{\star}(\sigma^2_w)\Mydef \lim_{t\to\infty}\mathsf{MSE}_\Lambda(V_l^t)$ be the MSE in the noisy setting. As $\sigma^2_w\to0$, we have
\[
\mathsf{MSE}_\Lambda^\star(\sigma^2_w) = C(\delta,f)\mathbb{E}\left[\Lambda^{-1}\right]\sigma^2_w\cdot(1+o(1)),
\]
where $0<C(\delta,f)<\infty$ is a constant depending only on $\delta$ and $f$.
\end{lemma}
This lemma confirms that as long as $\delta>\delta_{\Lambda}^{\mathsf{alg}}$, {GLM-EP-app} can offer stable recovery. However, the minimum mean square error in this case depends on another feature of the spectrum, namely $\mathbb{E}\left[\Lambda^{-1}\right]$. The optimal sensing mechanism should be designed by considering both features based on the expected noise level in the system.

\section{{Simulation Results}}

We next provide some simulation results for the {\Alg} algorithm for a few instances of $f$. Note that all our simulations are carried out using the original {\Alg} algorithm. Our results will show that the state evolution predictions are very accurate even without the smoothing introduced in Lemma \ref{Eqn:LFD2}.

\subsection{Sensing Matrix Model}
Let $\bm{A}=\bm{U\Sigma V}^\UT$. In our experiments, we approximate the random orthogonal matrix $\bm{U}$ in the following way: 
\[
\bm{U}=\bm{P}_1 \bm{U}_{\mr{d}} \bm{P}_2 \bm{U}_{\mr{d}}^\UT \bm{P}_3
\]
where $\bm{P}_1,\bm{P}_2,\bm{P}_3$ are three diagonal matrices with entries independently chosen from $\pm 1$ with equal probability, and $\bm{U}_{\mr{d}}$ is a discrete cosine transform (DCT) matrix. Note that all matrices are square. The hope is that by injecting enough randomness in these matrices, we can make them look like Haar orthogonal matrices for GLM-EP. In addition, such constructions allow fast implementation of GLM-EP using the DCT. 

Following \cite{Vehkapera2014}, we consider a geometric distribution for the limiting empirical distribution of $\mr{diag}(\bm{\Sigma}^\UT\bm{\Sigma})$:
\BE\label{Eqn:eig_geo}
P_\Lambda(\lambda;\alpha,\beta)=
\begin{cases}
\frac{1}{\beta {\lambda}}, & \text{if }{\lambda}\in\left(\alpha A(\beta)e^{-\beta},\alpha A(\beta)\right],\\
0,&\text{otherwise},
\end{cases}
\EE
where $\alpha>0$ is the mean, $\beta\ge0$ controls the spikeness of the distribution (with $\beta=0$ corresponding to a flat spectrum), and $A(\beta)=\frac{\beta}{1-e^{-\beta}}$. In all of our numerical experiments, the empirical eigenvalues are independently sampled from this distribution.

\subsection{Accuracy of state evolution}

Figure~\ref{Fig:approach_limit} demonstrates the mean-square error (MSE) performances of GLM-EP for $f(z)=|z|$ and the function defined in \eqref{Eqn:f_example3}. Clearly, $d(Y)=1$ for both functions. As Theorem \ref{The:1}, shows, the {\Alg} algorithm could achieve perfect reconstruction as soon as $\delta>1$ with a very spiky sensing matrix. Here, we considered the geometric eigenvalue setup with $\beta=20$. From Fig.~\ref{Fig:approach_limit}, we see that {\Alg}  recovers the signal accurately when $\delta$ is only slightly larger than the lower bound ($\delta=1.01$). Note that both $f$ considered in Fig.~\ref{Fig:approach_limit} are even functions, and for such functions the state evolution has a fixed point at $(V_r,V_l)=(1,\infty)$ (see Lemma \ref{Lem:FP_uninformative}), commonly referred to as the uninformative fixed point. This implies that the {\Alg} algorithm does not work for these $f$ if $\bm{z}_r^{-1}$ is uncorrelated with the signal. In our experiments, to get rid of the uninformative fixed point issue, we set $\bm{z}_r^{-1}=(1+V)^{-1}(\bm{z}+\sqrt{V}\bm{n})$ where $\bm{n}$ is standard Gaussian and $V$ is a large constant (here we set $V=20$).

\begin{figure}[htbp]
\begin{center}
\subfloat{\includegraphics[width=.45\textwidth]
{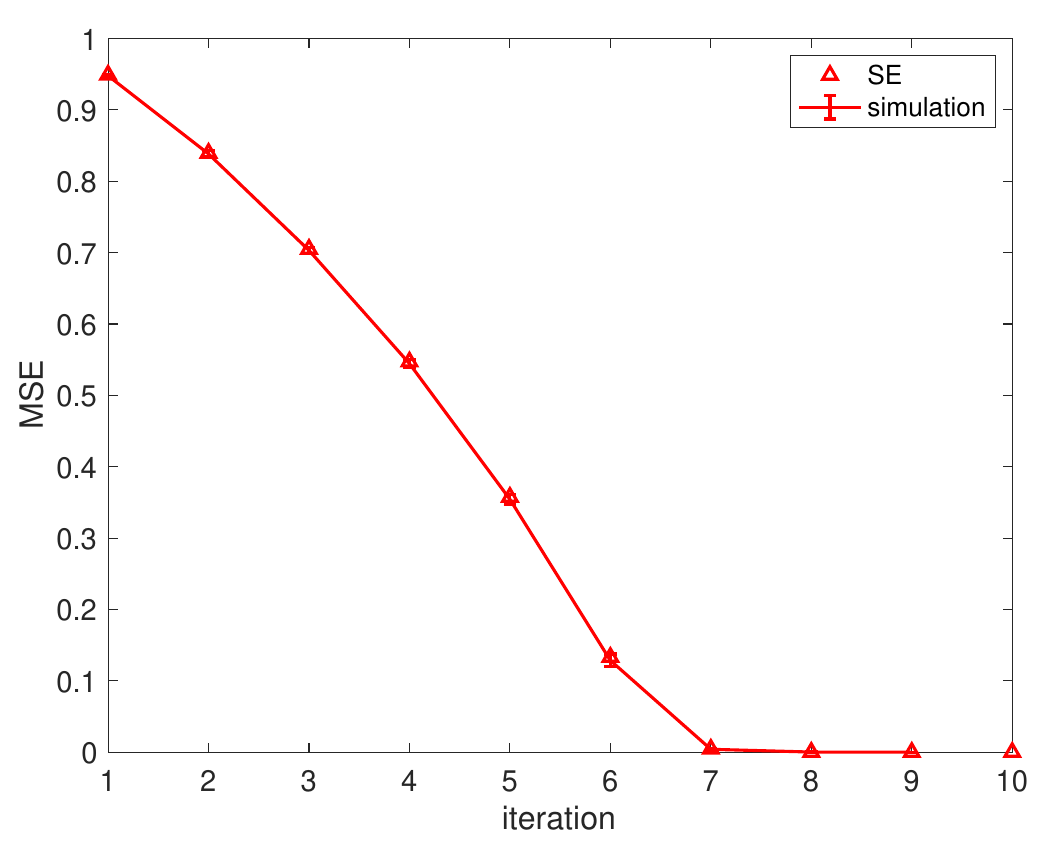}}\hfill
\subfloat{\includegraphics[width=.45\textwidth]
{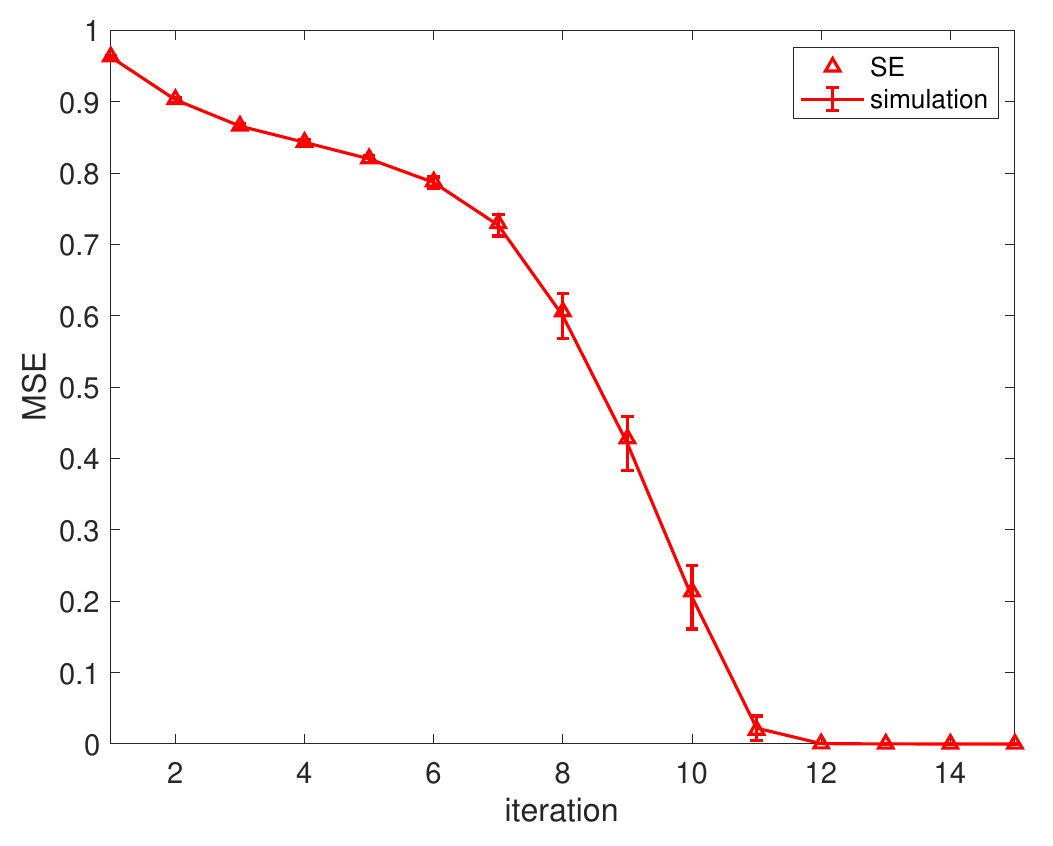}}
\caption{MSE performance of GLM-EP in the noiseless setting. \textbf{Left:} $f(z)=|z|$. \textbf{Right:} $f(z)$ defined in \eqref{Eqn:f_example3}. $n=2\times10^{5}$. $m=\lceil 1.01\cdot n\rceil$. $\beta=20$. $1000$ independent runs. The markers labeled `SE' are predictions obtained from state evolution. }\label{Fig:approach_limit}
\end{center}
\end{figure}

\subsection{Performance for medium-sized systems}
Fig.~\ref{Fig:medium_size} shows the performance of GLM-EP for  medium-sized sensing matrices ($n=5000$). Other settings are the same as Fig.~\ref{Fig:approach_limit}. In this case, we can observe a mismatch between the performance of GLM-EP and its theoretical predictions. Nevertheless, GLM-EP still achieve very good reconstruction result considering the fact that $\delta\approx1.01$ is very close to the information theoretical lower bound.


\begin{figure}[h]
\begin{center}
\subfloat{\includegraphics[width=.45\textwidth]
{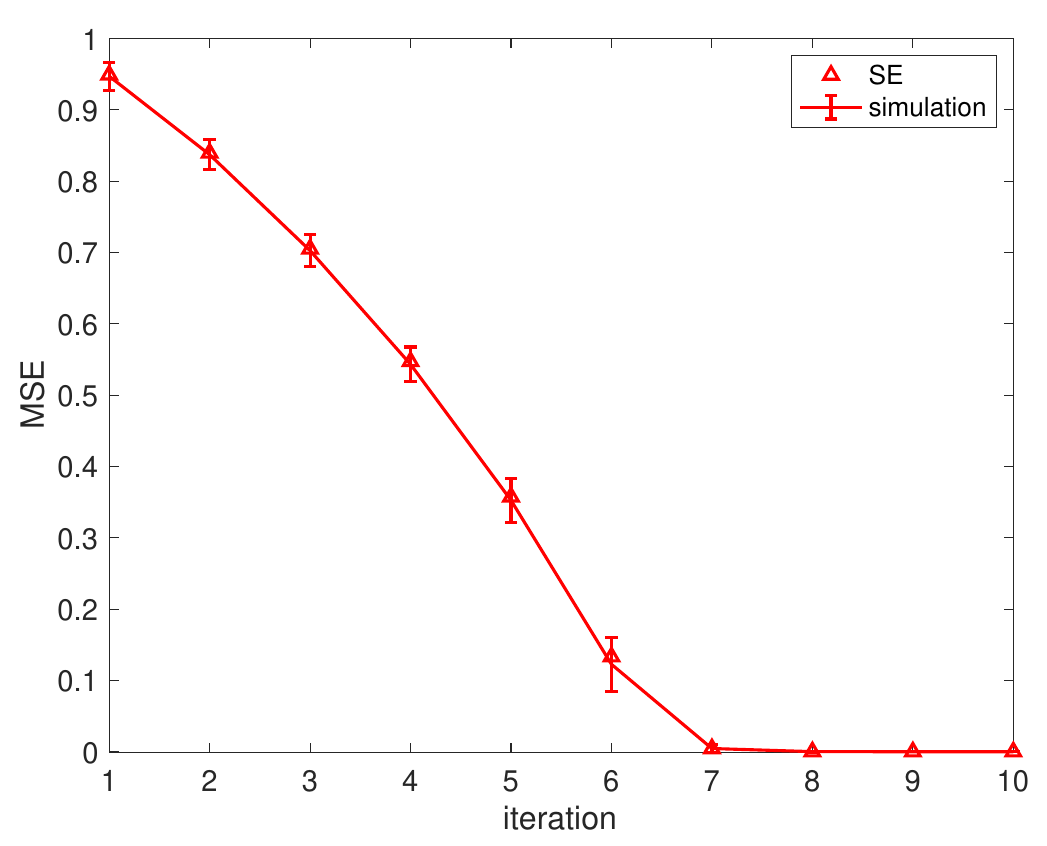}}\hfill
\subfloat{\includegraphics[width=.45\textwidth]
{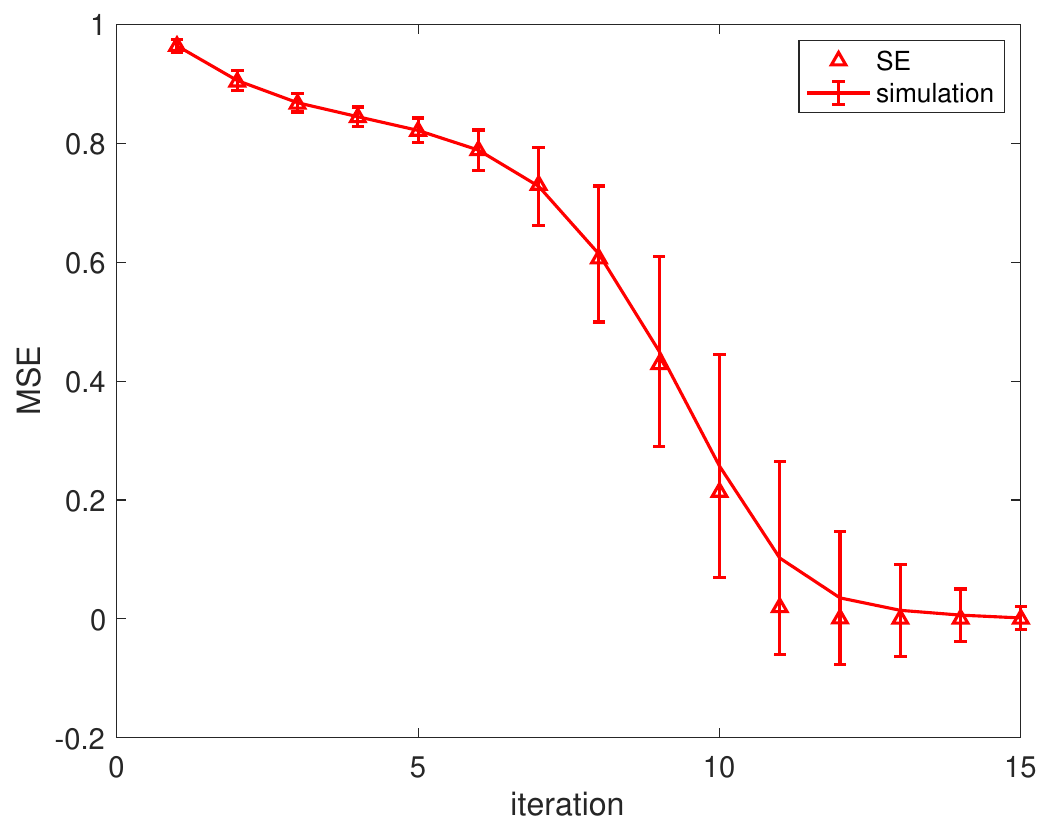}}
\caption{MSE performance of GLM-EP for medium-size systems. \textbf{Left:} $f(z)=|z|$. \textbf{Right:} $f(z)$ defined in \eqref{Eqn:f_example3}. $n=5000$. MSE are averaged over $1000$ independent runs. Other settings are the same as those of Fig.~\ref{Fig:approach_limit}. }\label{Fig:medium_size}
\end{center}
\end{figure}

\subsection{1-bit CS performance}
For the 1-bit compressed sensing (CS) problem, it is impossible to recover the signal accurately (namely, achieve zero MSE) at finite $\delta$. Tab.~\ref{Tab:1-bit} lists the MSE of GLM-EP for 1-bit CS under different values of $\delta$ and $\beta$. As expected, its performance improves as $\delta$ increases. Also, for each $\delta$, the MSE performances worsen as $\beta$ increases, which is consistent with our theoretical result about the impact of the spikeness.

\begin{table}[htbp]
\centering
 \begin{tabular}{| c|  c c c c c c c c |}
 \hline
$\delta$ & $1.5$ & $2$ & $2.5$ & $3$ & $3.5$ & $4$& $4.5$& $5$  \\[0.5ex]
 \hline
$\beta=0$ & 0.27 &  0.20 &  0.16 &  0.12  & 0.10   &  0.08    &0.07  &  0.06    \\
 \hline
 $\beta=5$ & 0.45   & 0.38  &  0.33 &   0.29 &   0.26  & 0.23 &   0.20  &  0.18  \\
 \hline
  $\beta=10$ & 0.62  &  0.58  &  0.55 &   0.52 &   0.50  &  0.48 &   0.45  &  0.44  \\
 \hline
 \end{tabular}  \\[0.5ex]
 \caption{MSE of GLM-EP for the 1-bit CS problem. $n=10^5$. The MSE is averaged over 100 independent runs. The number of iterations is 20.}\label{Tab:1-bit}
\end{table}

%
%
%

\subsection{Noisy measurements}

Lemma \ref{Lem:noise_sensitivity} analyzes the stability of the GLM-EP reconstruction for the noisy model $\bm{y}=f(\bm{Ax}+\bm{w})$. Tab.~\ref{Tab:PR_noisy} shows that the performance of GLM-EP for noisy phase retrieval. Here, the signal-to-noise ratio (SNR) is defined by
\[
\text{SNR}\Mydef \frac{\mathbb{E}[\|\bm{Ax}\|^2]}{\mathbb{E}[\|\bm{w}\|^2]}.
\]
Results in Tab.~\ref{Tab:PR_noisy} suggests that the performance of GLM-EP degrades gracefully as the noise variance increases.

\begin{table}[htbp]
\centering
 \begin{tabular}{| c | c c c c c |}
 \hline
 SNR &  30dB& 35dB  & 40dB &  45dB &  50dB   \\[0.2ex]
 \hline
MSE &1.28e-01 & 5.92e-02 & 2.18e-02 & 6.94e-03 & 2.14e-03\\
 \hline
 \end{tabular}  \\[0.5ex]
  \caption{MSE of GLM-EP for noisy phase retrieval. $\delta=1.1$. $n=10^5$. $\beta=10$. The MSE is averaged over 100 independent runs. The number of iterations is 10.} \label{Tab:PR_noisy}
\end{table}

\subsection{Phase transition}

To test the impact of the sensing spectrum on the performance of GLM-EP, we carry out phase transition study in Fig.~\ref{Fig:PT} under various values of $\beta$. We consider two instances of $f$, the absolute value function and that defined in \eqref{Eqn:f_example3}. We see that for both functions, the empirical perfect recovery threshold of $\delta$ improves as $\beta$ increases (corresponding to spikier spectrum), which is consistent with the claim of Theorem \ref{The:1}.

\begin{figure}[h]
\begin{center}
\subfloat{\includegraphics[width=.45\textwidth]
{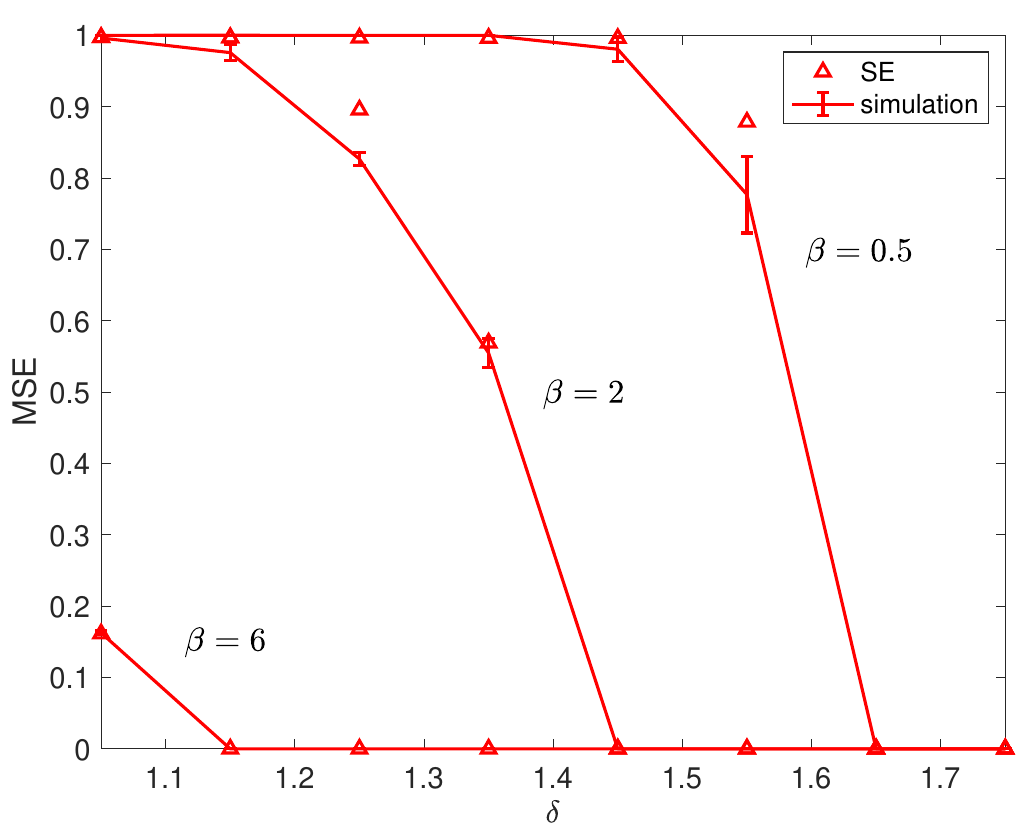}}\hfill
\subfloat{\includegraphics[width=.45\textwidth]
{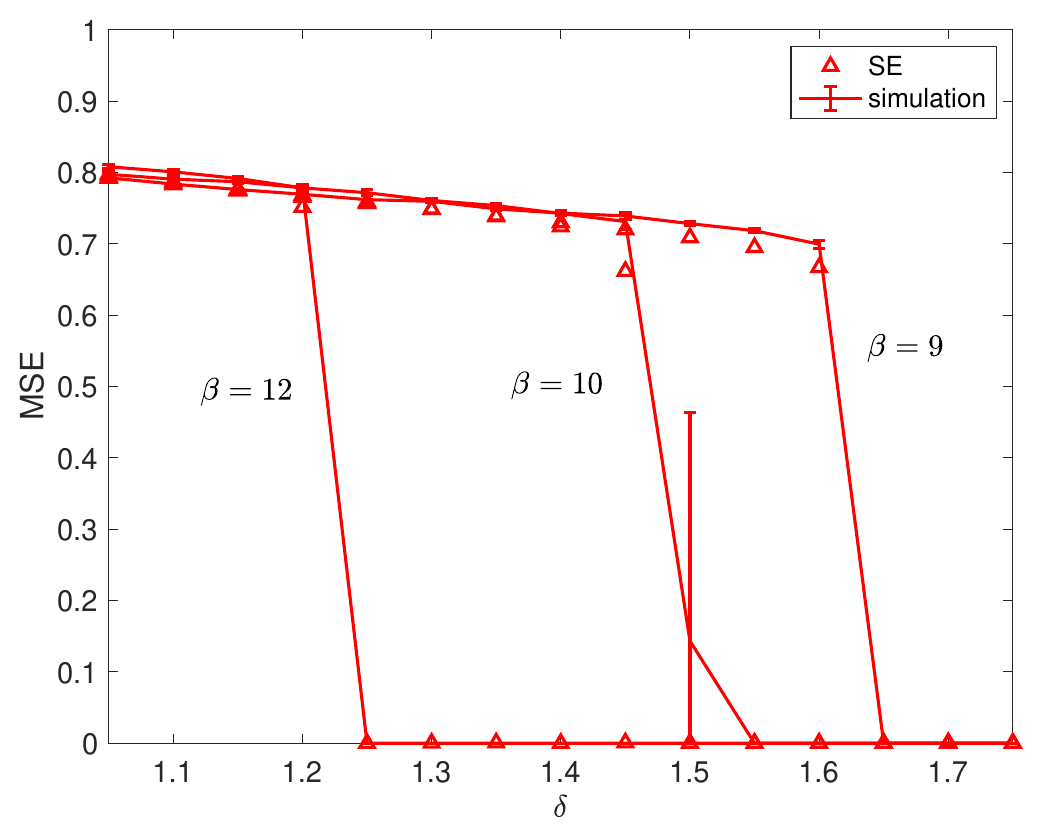}}
\caption{Phase transition of GLM-EP under various sensing matrix spectral. \textbf{Left:} $f(z)=|z|$. \textbf{Right:} $f(z)$ defined in \eqref{Eqn:f_example3}. $n=2\times10^5$. Error bars are calculated based on 100 independent runs.}\label{Fig:PT}
\end{center}
\end{figure}

\section{Conclusion and future work}
We studied the impact of the spectrum of the sensing matrix on the performance of the expectation propagation (EP) algorithm in recovering signals from the nonlinear model $\bm{y}=f(\bm{Ax})$. We defined a notion of spikiness of the distributions and showed that depending on $f (\cdot)$, the spikiness of the distribution can help or hurt the performance of EP. We also showed that spiky sensing matrices can always reduce the number of observations required for the exact recovery of $\bm{x}$ from $\bm{y}$. 

The results in this paper can serve as the first step towards the optimal design of sensing matrices. However, there are several directions that require further investigation before one can apply our results to real-world applications: (i) Since the structure of the signal is often used in recovery algorithms, the role of the structure should be studied more carefully when we deal with spiky sensing matrices. (ii) While we discussed the high-signal-to-noise ratio regime in the paper, some applications have low signal-to-noise ratios. The impact of the spectrum of the sensing matrix in such cases requires more careful considerations. 



\appendices

\section{Auxiliary results about MMSE dimension and the state evolution maps}\label{App:SE_maps}

In this section, after introducing the conditional MMSE dimension $\mathscr{D}(Z|Y)$, we present a few properties of $\mmse_z(\cdot)$ and the SE maps $\phi(\cdot)$, $\Phi(\cdot)$.

\subsection{MMSE Dimension and information dimension}\label{Sec:MMSE_dim_A}
The MMSE dimension $\mathscr{D}(Z)$ defined below characterizes the high SNR behavior of the MMSE $\mmse(Z,\snr)$. Similarly, $\mathscr{D}(Z|U)$ characterizes the high SNR behavior of $\mmse(Z,\snr|U)$.

\begin{definition}[MMSE dimension \cite{Wu11}]\label{Def:MMSE_dim}
The following limits, if exist, is called the MMSE dimension (resp. conditional MMSE dimension):
\BE
\begin{split}
\mathscr{D}(Z)&= \lim_{\snr\to\infty}\snr\cdot \mmse(Z,\snr),\\
\mathscr{D}(Z|U)&= \lim_{\snr\to\infty}\snr\cdot \mmse(Z,\snr|U).
\end{split}
\EE
\end{definition}\vspace{5pt}

The following lemma establishes the connection between the conditional MMSE dimension $\mathscr{D}(Z|Y)$ and the information dimension $d(Y)$ (see Section \ref{Sec:Renyi}).

\begin{lemma} \label{Lem:MMSE_D_ID}
Suppose Assumption (A.3) holds. Let $Z\sim\mathcal{N}(0,1)$ and $Y=f(Z)$. We have
\[
d(Y)=1-\mathscr{D}(Z|Y).
\]
\end{lemma}
\begin{proof}
The conditional MMSE dimension can be calculated as follows:
\BE
\begin{split}
\mathscr{D}(Z|Y) &= \lim_{\snr\to\infty}\snr\cdot \mmse(Z,\snr|Y)\\
&=\lim_{\snr\to\infty} \snr\cdot \mathbb{E}\left[\left(Z-\mathbb{E}[Z|\sqrt{\snr}Z+N,Y\right)^2\right]\\
&\overset{a}{=}\lim_{\snr\to\infty} \snr\cdot \mathbb{E}\left[\left(Z-\mathbb{E}[Z_y|\sqrt{\snr}Z_y+N]\right)^2\right]\\
&\Mydef\lim_{\snr\to\infty} \snr\cdot \mathbb{E}_Y \left[\mmse(Z_y,\snr)\right]\\
\end{split}
\EE
where $Z_y\sim P_{Z|Z\in f^{-1}(y)}$ and $N$ is independent of $Z$. Note that $\mmse(Z_y,\snr) \le \snr^{-1}$ \cite{guo2005mutual}\footnote{This is true even when the moments of $Z_u$ do not exist. To see this, consider $\tilde{Y}=\sqrt{\snr}Z_u+N$ and the linear estimator $\tilde{Y}/\sqrt{\snr}$. The MSE of this linear estimator is $\snr^{-1}$ and hence $\mmse(Z_u,\snr)\le \snr^{-1}$.} and so $\snr\cdot\mmse(Z_u,\snr)\le1$. Hence, by Lebesgue's dominated convergence theorem we have
\BE\label{Eqn:cond_D_tmp}
\begin{split}
\mathscr{D}(Z|Y)  &= \mathbb{E}_Y \left[\lim_{\snr\to\infty} \snr\cdot\mmse(Z_y,\snr)\right]\\
&= \mathbb{E}_Y \left[\mathscr{D}(Z_y)\right],
\end{split}
\EE
provided that $\lim_{\snr\to\infty} \snr\cdot\mmse(Z_y,\snr)$ exists almost surely. From \cite[Theorem 10 and Theorem 11]{Wu11}, 
\[
\mathscr{D}(Z_y)=
\begin{cases}
0 & \text{if } P_{Z|Y=y} \text{ is discrete}\\
1 & \text{if } P_{Z|Y=y} \text{ is absolutely continuous}
\end{cases}
\]
This implies that
\[
\mathscr{D}(Z_y)=
\begin{cases}
0 & \text{if } y\in\mathbb{R}\backslash\mathcal{Q}_f\\
1 & \text{if } y\in\mathcal{Q}_f.
\end{cases}
\]
Hence,
\[
\mathscr{D}(Z|Y)=\mathbb{P}\{f(Z)\in\mathcal{Q}_f\}=1-d(Y),
\]
where the second identity follows from Lemma \ref{Lem:ID_property}.
\end{proof}

\subsection{A property of $\mmse_z(v_r)$}
Note that $\eta_z(z_r,y,v)$ in GLM-EP (see \eqref{Eqn:eta_z_def}) is an MMSE estimator:
\BE\label{Eqn:eta_z_def_repeat}
\begin{split}
\eta_z(z_r,y,v)&=\mathbb{E}[Z|Y=y,Z_r=z_r]\\
&= \frac{\int_{f^{-1}(y)}u\cdot \mathcal{N}(u;z_r,v)du}{\int_{f^{-1}(y)}\mathcal{N}(u;z_r,v)du},
\end{split}
\EE
where $(Z,Z_r)\sim\mathcal{N}(\mathbf{0},\bm{\Sigma})$ where
\BE\label{Eqn:App_cov}
\bm{\Sigma} \Mydef 
\begin{bmatrix}
1 & 1-v_r\\
1-v_r & 1-v_r
\end{bmatrix},
\EE
and $Y=f(Z)$. Recall that $\mmse_z(v_r)$ is defined as 
\BE\label{Eqn:MMSE_def_app}
\mmse_z(v_r) = \mathbb{E}\Big( Z-\mathbb{E}[Z|Z_r,Y] \Big)^2,
\EE


Lemma \ref{Lem:MMSE_z_conditional} below is a consequence of the covariance structure of $(Z,Z_r)$ defined in \eqref{Eqn:App_cov}.
\begin{lemma}\label{Lem:MMSE_z_conditional}
Let $\mmse_z(v_r)$ be the MMSE defined in \eqref{Eqn:MMSE_def_app}. Let $Z\sim\mathcal{N}(0,1)$, $Y=f(Z)$ and $v_r\in(0,1]$. We have
\[
\mmse_z(v_r) = \mmse(Z,v_r^{-1}-1|Y),
\]
where the right hand side is a conditional MMSE defined in \eqref{Eqn:cond_MMSE}.
\end{lemma}

\subsection{Properties of the SE maps}

In this appendix, we discuss a few properties of the maps $\phi$ and $\Phi$ in \eqref{Eqn:SE_Gaussian}:
\BS\label{Eqn:SE_fixed}
\begin{align}
\phi(v_r)&=\left({ \frac{1}{\mmse_z\left(v_r\right)}-\frac{1}{v_r}}\right)^{-1}.\\
\Phi(v_l)&=\left({\frac{1}{\frac{1}{\delta}\cdot\mathbb{E}\left[ \frac{ v_l\Lambda}{v_l + \Lambda} \right]}-\frac{1}{v_l}}\right)^{-1},
\end{align}
\ES
where the expectation in $\Phi$ is over $\Lambda$, which is distributed according to the asymptotic eigenvalue distribution of $\bm{A}^{\UT}\bm{A}$, and
\[
\mmse_z(v_r) \Mydef  \mmse\left(Z,v_r^{-1}-1|f(Z)\right).
\]
The following lemmas collect some useful properties of the MMSE \cite{Guo2011}, and the maps $\phi$, $\Phi$.

\begin{lemma}[Properties of $\mmse(Z,\snr|U)$]\label{Lem:MMSE_z}
The following hold:
\begin{enumerate}
\item[(i)] Assume $Z\sim\mathcal{N}(0,1)$. Then, $\mmse(Z,\snr|U)\le \frac{1}{1+\snr}$, $\forall\snr>0$. Further, the inequality is strict if $U$ is not independent of $Z$.
\item[(ii)] $\frac{\mr{d}}{\mr{d}\snr}\mmse(Z,\snr|U)=-\mathbb{E}\left(\mr{var}^2[Z|\sqrt{\snr}Z+N,U]\right)$, where $\mr{var}[Z|\sqrt{\snr}Z+N,U]\Mydef \mathbb{E}[Z^2|\sqrt{\snr}Z+N,U]-\mathbb{E}^2[Z|\sqrt{\snr}Z+N,U]$, and $N\sim\mathcal{N}(0,1)$ is independent of $(Z,U)$.
\end{enumerate}
\end{lemma}\vspace{5pt}

\begin{lemma}[Properties of $\phi$ and $\Phi$]\label{Lem:phi}
The functions $\phi$ and $\Phi$ defined in \eqref{Eqn:SE_fixed} have the following properties:
\begin{enumerate}
\item[(i)] $\phi(v_r)$ is continuous and non-decreasing in $v_r\in(0,1)$. If $f(z)$ is not an invertible function, $\phi(v_r)$ is strictly increasing. Suppose that $f(Z)$ is not independent of $Z\sim\mathcal{N}(0,1)$. Then, $0\le\phi(v_r)<\infty$ and $\phi(0)=0$ if $d(f(Z))\neq0$, and $\phi(1)<\infty$ if $\mathbb{E}[Z|f(Z)]\neq0$;
\item[(ii)] $\Phi(v_l)$ is continuous and strictly increasing in $v_l\in(0,\infty)$. Further, $\Phi(0)=0$ and $\Phi(\infty)=1$. 
\end{enumerate}
\end{lemma}
\begin{proof}

\textit{Proof of (i):} The continuity of $\phi(v_r)$ is due to the continuity of the function $\mmse_z(v_r)=\mmse(Z,\snr|Y)$, where $\snr=v_r^{-1}-1$ \cite{Guo2011}.

We next prove that $\phi$ is strictly increasing. Differentiation yields (see \eqref{Eqn:SE_fixed})
\BS
\begin{align}
\phi'(v_r) &= \frac{v^2_r \cdot\mmse_z'(v_r)-\mmse_z^2(v_r)}{\left( v_r-\mmse_z(v_r) \right)^2}.\label{Eqn:phi_prime}
\end{align}
\ES
Hence, we only need to prove
\BS
\begin{align}
\mmse_z'(v_r) &> \frac{1}{v_r^2}\cdot \mmse_z^2(v_r), \quad \forall v_r\in(0,1]. \label{Eqn:lem_phi_1}
\end{align}
\ES
Recall the definition
\[
\mmse_z(v_r)=\mmse(Z,\snr|Y),\quad \snr \Mydef v_r^{-1}-1,
\]
and the derivative formula of the conditional MMSE in Lemma \ref{Lem:MMSE_z}, we have
\BE\label{Eqn:mmse_z_dif}
\mmse_z'(v_r) = \frac{1}{v_r^2}\cdot \mathbb{E}\left(\mr{var}^2[Z|\sqrt{\snr}Z+N,Y]\right),\quad \forall v_r\in(0,1].
\EE
Further,
\BE\label{Eqn:mmse_var}
\mmse_z(v_r)=\mmse(Z,\snr|Y)=\mathbb{E}\left(\var[Z|\sqrt{\snr}Z+N,Y]\right)
\EE
Combining \eqref{Eqn:mmse_z_dif} and \eqref{Eqn:mmse_var}, and applying Jensen's inequality proves $\mmse_z'(v_r) \ge \frac{1}{v_r^2}\cdot \mmse_z^2(v_r)$, and equality holds only when $\var[Z|\sqrt{\snr}Z+N,Y]$ is constant with respect to realizations of $\sqrt{\snr}Z+N$ and $Y$. This is only possible when $Z_y\sim P_{Z|Y=y}$ is Gaussian with $\var[Z_y]$ invariant to $y$ (including the degenerate case where $\var[Z_y]=0$). Again, this is only possible when $f(z)$ is an invertible function for which $Z_y$ is a constant and $\var[Z_y]=0$). To summarize, when $f(z)$ is not an invertible function, \eqref{Eqn:lem_phi_1} holds and so $\phi$ is a strictly increasing function.

Finally, we verify $\phi(0)$ and $\phi(1)$. First, for any $v_r\in(0,1)$, we have
\BE\label{Eqn:MMSE_le_v}
\begin{split}
\mmse_z(v_r)&=\mmse(Z,\snr|Y)\quad (\snr=v_r^{-1}-1)\\
&\overset{(a)}{\le} \mmse(Z,\snr)\\
&= \frac{1}{1+\snr}\\
&=v_r
\end{split}
\EE
where step (a) is from the fact that conditioning reduces MMSE \cite[Proposition 11]{Guo2011}. Further, the inequality is strict for $v_r\neq1$ ($\snr>0$) whenever $f(Z)$ is not independent of $Z$.
It follows that
\[
\phi(v_r)=\left(\frac{1}{\mmse_z(v_r)}-\frac{1}{v_r}\right)^{-1}\in[0,\infty),\quad \forall v_r\in(0,1).
\]
Further, $\phi(v_r)$ is continuously increasing in $(0,1)$ and so the limit $\lim_{v_r\to 0_+}\phi(v_r)$ exists (which is defined to be $\phi(0)$). Hence, $\phi(0)\ge0$.

Lemma \ref{Lem:MMSE_D_ID} shows $d(Y)=1-\mathscr{D}(Z|Y)$. Hence, if $d(Y)\neq0$, we would have 
\[
\mathscr{D}(Z|Y)\Mydef\lim_{\snr\to\infty}\snr\cdot \mmse(Z,\snr|Y)<1.
\]
Then,
\[
\begin{split}
\phi(0)& \Mydef \lim_{v_r\to0}\phi(v_r)\\
&=\lim_{v_r\to0}\frac{\mmse_z(v_r)}{1-\frac{\mmse_z(v_r)}{v_r}}\\
&\overset{(a)}{=}\lim_{\snr\to\infty}\frac{\mmse(Z,\snr|Y)}{1-(\snr+1)\mmse(Z,\snr|Y)}\quad (\snr=v_r^{-1}-1)\\
&= 0
\end{split}
\]
where step (a) follows from the definition of $\mmse_z$ below \eqref{Eqn:SE_Gaussian}, and the fact that $\lim_{\snr\to\infty}\mmse(Z,\snr|Y)=0$ and $\lim_{\snr\to\infty}\snr\cdot\mmse(Z,\snr|Y)=\mathscr{D}(Z|Y)<1$.

Finally,
\BE\label{Eqn:phi1}
\begin{split}
\phi(1)&=\left(\frac{1}{\mmse_z(1)}-1\right)^{-1}\\
&= \left(\frac{1}{\mmse(Z,\snr=0|Y)}-1\right)^{-1}\\
&= \left(\frac{1}{ \mathbb{E}\left( \var[Z|Y]\right) }-1\right)^{-1}\\
&=  \left(\frac{1}{ \mathbb{E}\left( \mathbb{E}[Z^2|Y]-\mathbb{E}^2[Z|Y]\right) }-1\right)^{-1}\\
&=  \left(\frac{1}{1- \mathbb{E}\left(\mathbb{E}^2[Z|Y]\right) }-1\right)^{-1}\quad(\mathbb{E}[Z^2]=1),
\end{split}
\EE
where $\mathbb{E}[Z^2]=1$ since $Z\sim\mathcal{N}(0,1)$. Hence, $\phi(1)\ge0$ and $\phi(1)<\infty$ if $\mathbb{E}[Z|Y]\neq 0$. 

\textit{Proof of (ii): } Similar to the proof of part (i), to prove $\Phi(v_l)$ is increasing, we only need to verify
\BE\label{Eqn:Phi_mono_ineq}
\frac{1}{\delta}\mathbb{E}\left[\left( \frac{v_l\Lambda}{v_l+\Lambda}\right)^2 \right]> \left(\frac{1}{\delta}\mathbb{E}\left[\frac{v_l\Lambda}{v_l+\Lambda}\right]\right)^2, \quad \forall v_r\in(0,1]. 
\EE
When $\delta>1$, Jensen's inequality yields the result:
\[
\begin{split}
\frac{1}{\delta}\mathbb{E}\left[\left( \frac{v_l\Lambda}{v_l+\Lambda}\right)^2 \right]&> \frac{1}{\delta}\left(\mathbb{E}\left[\frac{v_l\Lambda}{v_l+\Lambda}\right]\right)^2\\
&>\left(\frac{1}{\delta}\mathbb{E}\left[\frac{v_l\Lambda}{v_l+\Lambda}\right]\right)^2, \quad \forall v_r\in(0,1]. 
\end{split}
\]
For $\delta\le1$, note that $P_\Lambda = (1-\delta)P_0+\delta P_{\tilde{\Lambda}}$ where $P_{\tilde{\Lambda}}$ denotes the asymptotic eigenvalue distribution of $\bm{AA}^\UT$ (we have $\mathbb{E}[\tilde{\Lambda}^2]=1$). Hence, \eqref{Eqn:Phi_mono_ineq} can be reformulated as
\[
\mathbb{E}\left[\left( \frac{v_l\tilde{\Lambda}}{v_l+\tilde{\Lambda}}\right)^2 \right]> \left(\mathbb{E}\left[\frac{v_l\tilde{\Lambda}}{v_l+\tilde{\Lambda}}\right]\right)^2, \quad \forall v_r\in(0,1],
\]
and holds due to Jensen's inequality.
\end{proof}

\begin{lemma}\label{Lem:FP_uninformative}
If $f(z)=f(-z),\forall z$, then $\mmse_z(1)=1$. Further, $(V_r,V_l)=(1,\infty)$ is a fixed point of the state evolution equations in \eqref{Eqn:SE_Gaussian}.
\end{lemma}
\begin{proof}
Recall that $\mmse_z(v_r)=\mmse(Z,v_r^{-1}-1|Y)$. Hence, $\mmse_z(1)=\mmse(Z,\snr=0|Y)$ and
\[
\begin{split}
\mmse(Z,\snr=0|Y)&=\mathbb{E}\Big(\mathbb{E}[|Z|^2|Y]-\mathbb{E}^2[Z|Y]\Big)\\
&=\mathbb{E}\left(\mathbb{E}[|Z|^2|Y]\right)\\
&=\mathbb{E}[|Z|^2]=1.
\end{split}
\]

A simple calculation shows that $(V_r,V_l)=(1,0)$ is a fixed point of \eqref{Eqn:SE_Gaussian}.
\end{proof}

\begin{lemma}\label{Lem:Z_perp_Y}
Consider two independent Gaussian RVs: $Z\sim\mathcal{N}(0,\tau)$ and $W\sim\mathcal{N}(0,1)$. Suppose $U=Z+\sigma_w W$ and $Y_\sigma\sim P_{Y_\sigma}$, where $P_{Y_\sigma}\propto P_U \cdot P_{Y_\sigma|U}$ and $P_{Y_\sigma|U}$ is an arbitrary distribution. Define ${Z}_u^{\perp}\Mydef Z-\frac{\tau}{\tau+\sigma^2_w} U$. Then, we have $Z_u^{\perp}\Perp (U,Y_\sigma)$, where $A \Perp B$ means that $A$ and $B$ are independent. 
\end{lemma}

\begin{proof}
It is straightforward to show $Z_u^{\perp}\Perp U$. Since $Y_\sigma$ is generated from $U$, we also have $Z_u^{\perp}\Perp Y_\sigma$. 
\end{proof}

\vspace{10pt}

The following lemma summarizes a few useful properties of $\phi(v_r,\sigma^2_w)$ (which is the noisy counterpart of $\phi(v_r)$).
\begin{lemma}\label{Lem:phi_noisy}\label{Lem:MMSE_noisy}
Define 
\BE\label{Eqn:MMSE_noisy}
\mmse_z(v_r,\sigma^2_w)\Mydef\mathbb{E}\left[\left(\mathbb{E}[Z|Y_\sigma,Z_r]-Z\right)^2\right],
\EE
where $Z_r=(1-v_r)Z+\sqrt{v_r(1-v_r)}N$, $Y_\sigma=f(Z+\sigma_w W)$, $Z,W,N$ are mutually independent standar Gaussian RVs. Define
\BE\label{Eqn:phi_noisy_app}
\phi(v_r,\sigma^2_w)=\left( \frac{1}{\mmse_z(v_r,\sigma^2_w)} -\frac{1}{v_r} \right)^{-1}.
\EE
For any $\sigma_w>0$ and $v_r\in(0,1)$, $\phi(v_r,\sigma^2_w)$ satisfies the following:
\begin{itemize}
\item[(i)] $\phi(v_r,\sigma^2_w)$ is continuous and increasing in $v_r\in[0,1)$. Further, $\phi(v_r,\sigma^2_w)\ge0$;
\item[(ii)] $\sigma^2_w \le \phi(v_r,\sigma^2_w)<\infty$, $\forall v_r\in[0,1)$.
\end{itemize}
\end{lemma}

\begin{proof}
\textit{Part (i):} Same as Lemma \ref{Lem:phi}-(i). 

\textit{Part (ii):} 
We will show that $\mmse(v_r,\sigma^2_w)$ can be rewritten  as
\BE\label{Eqn:MMSE_noisy_equiv}
\begin{split}
&\mmse_z(v_r,\sigma^2_w) \\
&= \left(\frac{v_r}{v_r+\sigma^2_w}\right)^2 \mathbb{E}\Big(U-\mathbb{E}[U|Z_r,Y_\sigma]\Big)^2 + \frac{v_r\sigma^2_w}{v_r+\sigma^2_w},
\end{split}
\EE
where $U=Z+\sigma_w W$, $Y_\sigma=f(U)$, and $(U,Z_r)\sim\mathcal{N}(\mathbf{0},\bm{\Sigma})$ where
\[
\bm{\Sigma}=
\begin{bmatrix}
1+\sigma^2_w & 1-v_r\\
1-v_r & 1-v_r
\end{bmatrix}.
\]
From \eqref{Eqn:MMSE_noisy_equiv} we have
\[
\mmse_z(v_r,\sigma^2_w)\ge\frac{v_r\sigma^2_w}{v_r+\sigma^2_w},
\]
which together with \eqref{Eqn:phi_noisy_app} yields $\phi(v_r,\sigma^2_w)\ge\sigma_w^2$. We next prove the boundedness of $\phi(v_r,\sigma^2_w)$. Substituting \eqref{Eqn:MMSE_noisy_equiv} into \eqref{Eqn:phi_noisy_app} and after straightforward calculations, we have
\[
\phi(v_r,\sigma^2_w)=\frac{v_r\cdot \mathbb{E}\Big(U-\mathbb{E}[U|Z_r,Y_\sigma]\Big)^2}{ v_r+\sigma^2_w-\mathbb{E}\Big(U-\mathbb{E}[U|Z_r,Y_\sigma]\Big)^2}.
\]
Since conditioning reduces MMSE \cite[Proposition 11]{guo2005mutual}, we have 
\[
\mathbb{E}\Big(U-\mathbb{E}[U|Z_r,Y_\sigma]\Big)^2\le \mathbb{E}\Big(U-\mathbb{E}[U|Z_r]\Big)^2=v_r+\sigma^2_w,
\]
where the inequality is strict whenever $Y_\sigma$ is not independent of $U$. All together, we have $\phi(v_r,\sigma^2_w)<\infty$.\vspace{5pt}

It only remains to prove \eqref{Eqn:MMSE_noisy_equiv}.
Let us write $Z={Z}_r+\tilde{Z}$, where $\tilde{Z}\sim\mathcal{N}(0,v_r)$ is independent of $Z_r$. We have
\[
\begin{split}
\tilde{U}\Mydef U -Z_r= \tilde{Z}+\sigma_wW.
\end{split}
\]
Define
\[
\tilde{Z}_{\tilde{u}}^{\perp}\Mydef \tilde{Z} - \frac{v_r}{v_r+\sigma^2_w}\tilde{U}.
\]
By construction, $\tilde{Z}_{\tilde{u}}^{\perp}\Perp \tilde{U}$. We also have $\tilde{Z}_{\tilde{u}}^{\perp}\Perp Z_r$\footnote{Throughout this paper, $A\Perp B$ denotes the random variables $A$ and $B$ are independent}, since $\tilde{Z}_{\tilde{u}}^{\perp}$ a linear combination of $\tilde{Z}$ and $W$ and the latter two RVs are independent of $Z_r$. Also, $\tilde{Z}_{\tilde{u}}^{\perp}\Perp Y_\sigma$ according to Lemma \ref{Lem:Z_perp_Y}. Hence,
\[
\begin{split}
\mathbb{E}[\tilde{Z}|Z_r,Y_\sigma]&=\frac{v_r}{v_r+\sigma^2_w}\cdot\mathbb{E}[\tilde{U}|Z_r,Y_\sigma]+\mathbb{E}[\tilde{Z}_{\tilde{u}}^{\perp}|Z_r,Y_\sigma]\\
&= \frac{v_r}{v_r+\sigma^2_w}\cdot\mathbb{E}[\tilde{U}|Z_r,Y_\sigma],
\end{split}
\]
where the last step is due to the independence of $\tilde{Z}_{\tilde{u}}^{\perp}$ and $(Z_r,Y_\sigma)$ and the fact that $\tilde{Z}_{\tilde{u}}^{\perp}$ is zero-mean Gaussian. Hence, we have
\[
\begin{split}
&\mathbb{E}\Big(Z-\mathbb{E}[Z|Z_r,Y_\sigma]\Big)^2\\
&=\mathbb{E}\Big(\tilde{Z}-\mathbb{E}[\tilde{Z}|Z_r,Y_\sigma]\Big)^2\\
&= \mathbb{E}\Big(\frac{v_r}{v_r+\sigma^2_w}\tilde{U}+\tilde{Z}_{\tilde{u}}^{\perp}-\mathbb{E}[\tilde{Z}|Z_r,Y_\sigma]\Big)^2\\
&= \mathbb{E}\Big(\frac{v_r}{v_r+\sigma^2_w}\tilde{U}+\tilde{Z}_{\tilde{u}}^{\perp}-\frac{v_r}{v_r+\sigma^2_w}\cdot\mathbb{E}[\tilde{U}|Z_r,Y_\sigma]\Big)^2\\
&\overset{(a)}{=} \left(\frac{v_r}{v_r+\sigma^2_w}\right)^2\cdot \mathbb{E}\Big(\tilde{U}-\mathbb{E}[\tilde{U}|Z_r,Y_\sigma]\Big)^2+\frac{v_r\sigma^2_w}{v_r+\sigma^2_w}\\
&=\left(\frac{v_r}{v_r+\sigma^2_w}\right)^2\cdot \mathbb{E}\Big(U-\mathbb{E}[U|Z_r,Y_\sigma]\Big)^2+\frac{v_r\sigma^2_w}{v_r+\sigma^2_w}
\end{split}
\]
where step (a) is due to the fact that $\tilde{Z}_{\tilde{u}}^\perp\Perp (\tilde{U},Z_r,Y_\sigma)$ and $\mathbb{E}[(\tilde{Z}_{\tilde{u}}^\perp)^2]=\frac{v_r\sigma^2_w}{v_r+\sigma^2_w}$.
\end{proof}
\vspace{5pt}

\section{Proof of Theorem \ref{Lem:converse2}}\label{App:Low_bound}

We first recall a few definitions and useful lemmas from \cite{Wu10,Wu12} in Section \ref{Sec:converse_def}. Then, we introduce our main technical lemma in Section \ref{Sec:converse_aux}, and finally prove Theorem \ref{Lem:converse2} in Section \ref{Sec:converse_main}.
\subsection{Minkowski dimension}\label{Sec:converse_def}
In the almost lossless analog signal compression framework developed in \cite{Wu10,Wu12}, the description complexity of bounded sets is gauged via their Minkowski dimension. Minkowski dimension is also called box-counting dimension \cite{pesin2008dimension} ({hence the subscript $B$ in the notation $\overline{\mr{dim}}_B$).

\begin{definition}[Minkowski Dimension]
Let $\mathcal{S}$ be a nonempty bounded subset of a metric space. The upper Minkowski dimension of $\mathcal{S}$ is defined as
\BE\label{Eqn:Mink}
\overline{\mr{dim}}_B(\mathcal{S})= \limsup_{\epsilon\to0}\frac{\log N_S(\epsilon)}{\log\frac{1}{\epsilon}},
\EE
where $N_S(\epsilon)$ is the $\epsilon$-covering number of $\mathcal{S}$, that is
\[
N_S(\epsilon)\Mydef \min\Big\{ k: \mathcal{S}\subset \bigcup_{i=1}^k B(x_i,\epsilon), X_i\in \mathcal{S} \Big\},
\]
where $B(x_i,\epsilon)$ denotes a ball centered at $x_i$ with radius $\epsilon$.
\end{definition}\vspace{5pt}


For a probability measure, we define its $\epsilon$-Minkowski dimension \cite{Wu12} as the smallest Minkowski dimension among all sets with measure at least $1-\epsilon$.

\begin{definition}[$\epsilon$-Minkowski Dimension]\label{Def:epsilon_Min}
Let $\mu$ be a probability measure on $\mathbb{R}^n$. Define the $\epsilon$-Minkowski dimension of $\mu$ as
\BE\label{Eqn:Mink_epsilon}
\overline{\mr{dim}}_B^\epsilon(\mu)=\inf\{\overline{\mr{dim}}_B(\mathcal{S}): \mu(\mathcal{S})\ge 1-\epsilon \}.
\EE
\end{definition}

An asymptotic version of the $\epsilon$-Minkowski dimension (called the Minkowski dimension compression rate) was introduced in \cite{Wu10}. Wu and Verd{\'u} \cite{Wu10} proved that the probability measure of an i.i.d. source concentrates on sets with Minkowski dimension approximately equal to the R\'enyi information dimension of the measure.

We will use the following lemma from \cite{Wu10} in the proof of the auxiliary lemma in Section \ref{Sec:converse_aux}.

\begin{lemma}[Minkowski dimension in Euclidean spaces]\label{Lem:aux_lem3}
Let $\mathcal{S}$ be a bounded subset in $(\mathbb{R}^n,\|\cdot\|_2)$. The Minkowski dimension satisfies
\BE\label{Eqn:Min_def_equiv}
\overline{\mr{dim}}_B(\mathcal{S}) =\limsup_{q\to\infty}\frac{\log\big|\langle \mathcal{S} \rangle_{2^q}\big|}{q}
\EE
where $\langle x \rangle_p\Mydef \lfloor px\rfloor/p$, and $\langle \mathcal{S} \rangle_p\Mydef\{\langle x \rangle_p:x\in\mathcal{S}\}$, and the logarithm uses base $2$.
\end{lemma}
Lemma \ref{Lem:aux_lem3} shows that in Euclidean spaces, we could replace $\epsilon$-balls by mesh cubes in defining covering number for Minkowski dimension. The similar forms of \eqref{Eqn:Min_def_equiv} and Definition \ref{Def:Renyi_ID} suggest the close relationship between Minkowski dimension and information dimension. Roughly speaking, Minkowski dimension counts the number of small pieces needed to cover the set while the information dimension also takes into account the probability of each piece and replaces the $\log N_S(\epsilon)$ term in \eqref{Eqn:Mink} by an entropy term.

\subsection{An Auxiliary Lemma}\label{Sec:converse_aux}
We introduce a few definitions. First, recall the definition
\[
\mathcal{Q}_f\Mydef \{y: f^{-1}(y) \text{ contains an interval}\},
\]
where $f^{-1}(y)\Mydef\{z: f(z)=y\}$. We assumed $\mathcal{Q}_f$ to be a finite set. For $\bm{y}\in\mathbb{R}^m$, let
\BE\label{Eqn:Spt_def}
\text{Spt}(\bm{y}) \Mydef  \{i=1,\ldots,m: \ y_i\in\mathbb{R}\backslash \mathcal{Q}_f\}
\EE
be a kind of generalized support of $\bm{y}$ \cite{Wu10} (i.e., locations of the components of $\bm{y}$ that do not fall into the ``flat'' sections of $f$).

For convenience, we introduce the following definitions:
\BE\label{Eqn:def_XY_ar}
\begin{split}
\mathcal{A}_{\alpha}& \Mydef \left\{ \bm{s}\in\mathbb{R}^n:\bm{y}=f(\bm{A}(\bm{s})),\frac{|\text{Spt}(\bm{y})|}{m} \le \alpha\right\},\\
 \mathcal{B}_{\alpha} &\Mydef \left\{ \bm{y}\in\mathbb{R}^m:\frac{|\text{Spt}(\bm{y})|}{m} \le \alpha\right\},\\
\mathcal{A}_{r} &\Mydef \left\{ \bm{s}\in\mathbb{R}^n:\bm{y}=f(\bm{A}(\bm{s})),\|\bm{y}\| \le r\right\},\\
\mathcal{B}_{r}& \Mydef \left\{ \bm{y}\in\mathbb{R}^m:\|\bm{y}\| \le r\right\}.
\end{split}
\EE
Further, let $\mathcal{A}$ and $\mathcal{B}$ be the set of signals and measurements that can be perfectly reconstructed under decoder $g$. More specifically,
\[
\mathcal{A}\Mydef \left\{ \bm{s}\in\mathbb{R}^n: g(f(\bm{As}))=\bm{s} \right\}\quad\text{and}\quad \mathcal{B}\Mydef \left\{ f(\bm{As}):\bm{s}\in\mathcal{A} \right\}.
\]
Clearly, the composite function $f\circ \bm{A}$ is invertible (with $g$ being the inverse function) if we restrict its domain and co-domain to $\mathcal{A}$ and $\mathcal{B}$ respectively. With the above definitions, we have
\[
\begin{split}
\mathcal{B}\cap\mathcal{B}_{\alpha}\cap \mathcal{B}_r=\left\{ \bm{y}:\bm{y}\in\mathcal{B},\ \frac{|\text{Spt}(\bm{y})|}{m} \le \alpha,\ \|\bm{y}\|\le r \right\},
\end{split}
\]
and
\[
\begin{split}
&g\left( \mathcal{B}\cap\mathcal{B}_{\alpha}\cap \mathcal{B}_r \right) \\
&=\left\{ g(\bm{y}):\bm{y}\in\mathcal{B},\ \frac{|\text{Spt}(\bm{y})|}{m} \le \alpha,\ \|\bm{y}\|\le r \right\}\\
&= \left\{ \bm{s}:\bm{s}\in\mathcal{A},\ \frac{|\text{Spt}(g^{-1}(\bm{s}))|}{m} \le \alpha,\ \|g^{-1}(\bm{s})\|\le r \right\}\\
&= \left\{ \bm{s}:\bm{s}\in\mathcal{A},\ \frac{|\text{Spt}(f(\bm{As}))|}{m} \le \alpha,\ \|f(\bm{As})\|\le r \right\}\\
&=\mathcal{A}\cap \mathcal{A}_\alpha\cap\mathcal{A}_r.
\end{split}
\]

Lemma \ref{Lem:converse}, which is a variation of \cite[Theorem 5]{Wu12}, is key to our proof of Theorem \ref{Lem:converse2}. Notice that Lemma \ref{Lem:converse} is a non-asymptotic result. Also, the radius $r$ of the boundedness constraint does not appear in \eqref{Eqn:converse_finite}.

\begin{lemma}\label{Lem:converse}
Let $P_X$ be an arbitrary absolutely continuous distribution with respect to the Lebesgue measure and $\bm{x}\sim \prod_{i=1}^n P_X(x_i)$ a random vector. Suppose that for some $\alpha\in(0,1]$, $r>0$ and $\epsilon\in(0,1)$, there exists a matrix $\bm{A}\in\mathbb{R}^{m\times n}$ and a Lipschitz continuous decoder $g:\mathcal{Y}^m\mapsto \mathbb{R}^n$ such that
\BE\label{Eqn:Error_Prob}
\mathbb{P}\big\{\bm{x}\in \mathcal{A}\cap\mathcal{A}_\alpha\cap\mathcal{A}_r \big\}\ge1-\epsilon,
\EE
where the probability is taken over $\bm{x}$. 
Then, necessarily we have
\BE\label{Eqn:converse_finite}
\frac{m}{n} \ge \frac{1-\epsilon}{\alpha}.
\EE
\end{lemma}

\begin{proof}
Our proof follows from the following chain of inequalities:
\[
\begin{split}
\alpha\cdot m &\ \overset{(a)}{\ge} \ \overline{\mr{dim}}_B(\mathcal{B}_{\alpha}\cap\mathcal{B}_r)\\
&\ \ge\  \overline{\mr{dim}}_B(\mathcal{B}_{\alpha}\cap\mathcal{B}_r\cap \mathcal{B})\\
&\ \overset{(b)}{\ge}\ \overline{\mr{dim}}_B(g(\mathcal{B}_{\alpha}\cap\mathcal{B}_r\cap \mathcal{B}))\\
&\ =\ \overline{\mr{dim}}_B(\mathcal{A}_{\alpha}\cap\mathcal{A}_r\cap \mathcal{A})\\
&\ \overset{(c)}{\ge}\  \overline{\mr{dim}}_B^{\epsilon}(P_{\bm{x}})\\
&\ \overset{(d)}{\ge}\ \bar{d}(\bm{x})-\epsilon n\\
&\ \overset{(e)}{=}\ (1-\epsilon) n
\end{split}
\]
where step (b) is from the fact that Minkowski dimension does not increase under Lipschitz mapping \cite[Proposition 2.5]{falconer2004fractal}, and step (c) is from the definition of $\epsilon$-Minkowski dimension (see Definition \ref{Def:epsilon_Min}) and $\mathbb{P}\left\{ \bm{x}\in \mathcal{A}\cap\mathcal{A}_{\alpha}\cap \mathcal{A}_r \right\}\ge 1-\epsilon$, step (d) is proved in \cite[Theorem 5]{Wu12}, and step (e) is from the fact that $d(\bm{x})=n\cdot d(X)=n$ when $\bm{x}\sim\prod_{i}^n P_X(x_i)$ and $P_X$ is absolutely continuous.

It remains to prove step (a).
Now, we use Lemma \ref{Lem:aux_lem3}:
\[
\overline{\mr{dim}}_B(\mathcal{B}_{\alpha}\cap\mathcal{B}_r)=\limsup_{M\to\infty} \frac{\log\big| \langle \mathcal{B}_{\alpha}\cap\mathcal{B}_r\rangle_M \big|}{\log M},
\]
where $\langle \mathcal{B}_{\alpha}\cap\mathcal{B}_r \rangle_M$ is a set obtained by applying the discretization operator $\langle y\rangle_M=\floor{My}/M$ (which has $2rM$ quantization levels in $y\in[-r,r]$) component-wisely to all the elements in $\mathcal{B}_{\alpha}\cap\mathcal{B}_r$. From our definition of $\mathcal{B}_{\alpha}\cap\mathcal{B}_r$, we have
\begin{equation}\label{Eqn:B_bound}
\begin{split}
&\big| \langle \mathcal{B}_{\alpha}\cap\mathcal{B}_r  \rangle_M \big| \\
&= \left| \left \{ \langle\bm{y}\rangle_M :  \bm{y}\in\mathbb{R}^m,\ \frac{|\text{Spt}(\bm{y})|}{m}\le \alpha,\  { \|\bm{y}\|\le r }\right \}\right|\\
&\le \left| \left \{ \langle\bm{y}\rangle_M :  \frac{|\text{Spt}(\bm{y})|}{m}\le \alpha,\  { \bm{y}\in[-r,r]^m }\right \}\right|\\
&\overset{(a)}{\le} \sum_{i=0}^{\lfloor\alpha m\rfloor} {m\choose i} (2rM)^{i}|\mathcal{Q}_f|^{m-i} \\
&\le \sum_{i=0}^{\lfloor\alpha m\rfloor} {m\choose i} (2rM)^{\lfloor\alpha m\rfloor}|\mathcal{Q}_f|^{m-\lfloor\alpha m\rfloor} \quad \text{for } M>|\mathcal{Q}_f|/(2r)\\
&\le 2^m (2rM)^{\lfloor\alpha m\rfloor}|\mathcal{Q}_f|^{m-\lfloor \alpha m\rfloor},
\end{split}
\end{equation}
where $\mathcal{Q}_f\Mydef \{y: f^{-1}(y) \text{ contains an interval}\}$, and we assumed $|\mathcal{Q}_f|<\infty$. Here are the detailed derivations for step (a). Denote
\[
\begin{split}
\mathcal{C}&:=\left \{ \langle\bm{y}\rangle_M :  \frac{|\text{Spt}(\bm{y})|}{m}\le \alpha,\  { \bm{y}\in[-r,r]^m }\right \}\\
&=\bigcup_{i=0}^{\lfloor \alpha m\rfloor}\bigcup_{S\subseteq \{1,\ldots,m\},|S|=i}\mathcal{C}_{i,S}\\
\end{split}
\]
where 
\[
\mathcal{C}_{i,S}:=\left \{ \langle\bm{y}\rangle_M :  \bm{y}_{S^c}\in \mathcal{Q}_f^{|S^c|},\  { \bm{y}\in[-r,r]^m }\right \},\quad |S|=i.
\]
In the above display, $\bm{y}_{S^c}$ denotes the vector formed by the entries of $\bm{y}$ in the index set $S^c$ (complement of $S$). Now, consider an arbitrary element $z\in \mathcal{C}_{i,S}$. Recall that $\langle\cdot\rangle_M$ denotes a quantization operation and the total number of possible quantized values in the interval $[-r,r]$ is $2rM$. Hence, for any $j\in[m]$, $z_j$ can take at most $2rM$ different values, due to the constraint $\bm{y}\in[-r,r]^m$. If $j\in S$, we know additionally that $z_j\in \langle \mathcal{Q}_f\rangle_M$, which can take at most $|\mathcal{Q}_f|$ different values. Hence, the cardinality of $\mathcal{C}_{i,S}$ can be upper bounded as
\begin{equation}\label{Eqn:B_iS12}
|\mathcal{C}_{i,S}|\le (2r M)^{|S|}|\mathcal{Q}_f|^{m-|S|}=(2r M)^i|\mathcal{Q}_f|^{m-i}.
\end{equation}
Hence,
\[
\begin{split}
|\mathcal{C}| &\le \sum_{i=0}^{\lfloor \alpha m\rfloor}\sum_{S\subseteq \{1,\ldots,m\},|S|=i}|\mathcal{C}_{i,C}| \\
&\le \sum_{i=0}^{\lfloor\alpha m\rfloor} {m \choose i}(2r M)^i|\mathcal{Q}_f|^{m-i},
\end{split}
\]
which leads to Step (a) of \eqref{Eqn:B_bound}.

As a consequence of \eqref{Eqn:B_bound}, we have
\[
\limsup_{M\to\infty} \frac{\log\big| \langle \mathcal{B}_{\alpha}\cap\mathcal{B}_r \rangle_M \big|}{\log M}\le \lfloor\alpha\cdot m\rfloor\le \alpha m.
\]
Combining all the above arguments yields
\[
\alpha\cdot m \ge (1-\epsilon) n,
\]
and hence the claimed lower bound on $m/n$.
\end{proof}\vspace{5pt}

In view of Lemma \ref{Lem:converse}, we can now prove the converse result in Theorem \ref{Lem:converse2}.

\subsection{Proof of Theorem \ref{Lem:converse2}}\label{Sec:converse_main}

Our proof relies on the following lemma, whose proof is postponed in Section \ref{Sec:proof_Lemma10}.

\begin{lemma}\label{Lem:converse_converge}
Let $\bm{z}\Mydef \bm{Ax}$. Under Assumptions (A.1)-(A.3), the following holds almost surely as $m,n\to\infty$ and $m/n\to\delta \in(1,\infty)$:
\BE\label{Eqn:App_edf_def}
\frac{1}{m}\sum_{i=1}^m \mathbb{I}\left(z_i\le t\right)\overset{a.s.}{\longrightarrow} \Phi(t), \quad \forall t\in\mathbb{R}.
\EE
\end{lemma}

From Lemma \ref{Lem:converse_converge}, as $m,n\to\infty$ with $m,n\to\delta\in(1,\infty)$, the empirical distribution of $\bm{z}=\bm{Ax}$ converges to standard Gaussian in probability
\[
\frac{1}{m}\sum_{i=1}^m \mathbb{I}\left(z_i\le t\right)\overset{P}{\to} \Phi(t), \quad \forall t\in\mathbb{R}.
\]
Consequently, 
\[
\begin{split}
\frac{| \{i=1,\ldots,m:\,z_i\in\mathbb{R}\backslash f^{-1}(\mathcal{Q}_f) | }{m} &\overset{P}{\to} 1-\mathbb{P}(Z\in f^{-1}(\mathcal{Q}_f))\\
&=d(Y).
\end{split}
\]
where $Z\sim\mathcal{N}(0,1)$, $\mathcal{Q}_f\Mydef\{y:f^{-1}(y) \text{ contains an interval}\}$ and the identity $1-\mathbb{P}(Z\in f^{-1}(\mathcal{Q}_f))=d(Y)$ is due to Lemma \ref{Lem:ID_property}. This is equivalent to (see \ref{Eqn:Spt_def})
\BE\label{Eqn:converse_converge_1}
\begin{split}
\frac{|\mr{Spt}(f(\bm{Ax}))|}{m}&=\frac{|\mr{Spt}(f(\bm{z}))|}{m}\\
&=\frac{| \{i=1,\ldots,m:\,z_i\in\mathbb{R}\backslash f^{-1}(\mathcal{Q}_f) | }{m}\\
&\overset{P}{\to} d(Y).
\end{split}
\EE
Hence, for any $\kappa>0$,
\[
\lim_{n\to\infty}\mathbb{P}\left\{ \left| \frac{|\mr{Spt}(f(\bm{Ax}))|}{m} - d(Y) \right|<\kappa \right\}=1.
\]
It is understood that in the above limit $m$ and $n$ tend to infinity with $m/n\to\delta$.
In view of the definition of $\mathcal{A}_\alpha$ in \eqref{Eqn:def_XY_ar}, we have
\[
\lim_{n\to\infty}\mathbb{P}\left\{\bm{x}\in \mathcal{A}_\alpha\right\} =1,\quad \text{for }\alpha=d(Y)+\kappa.
\]
Hence, for any $\epsilon>0$ and $\alpha=d(Y)+\kappa$, there exists sufficiently large $n,m$ such that
\[
\mathbb{P}\left\{\bm{x}\in \mathcal{A}_\alpha\right\} \ge 1-\frac{\epsilon}{3}.
\]
Further, since $\lim_{r\to\infty}\mathbb{P}\left\{\bm{x}\in \mathcal{A}_\alpha \cap\mathcal{A}_r\right\} =\mathbb{P}\left\{\bm{x}\in \mathcal{A}_\alpha \right\} $, there exists sufficiently large $r$ such that
\BE\label{Eqn:converse_alpha_r_prob}
\mathbb{P}\left\{\bm{x}\in \mathcal{A}_\alpha \cap\mathcal{A}_r\right\} \ge \mathbb{P}\left\{\bm{x}\in \mathcal{A}_\alpha \right\} -\frac{\epsilon}{3}\ge 1-\frac{2\epsilon}{3}.
\EE

Suppose that the decoding error probability does not exceed $\epsilon/3$, namely,
\BE
\mathbb{P}\{\bm{x}\in\mathcal{A}\}\ge 1-\frac{\epsilon}{3}.
\EE
For $\alpha=d(Y)+\kappa$, and sufficiently large $r$ and $m,n$, we have
\[
\begin{split}
\mathbb{P}\left\{\bm{x}\in\mathcal{A}\cap\mathcal{A}_\alpha\cap\mathcal{A}_r  \right\} & \ge\mathbb{P}\left\{\bm{x}\in\mathcal{A}  \right\}+\mathbb{P}\left\{\bm{x}\in \mathcal{A}_\alpha \cap\mathcal{A}_r\right\}-1\\
&\ge 1-\epsilon,
\end{split}
\]
where the second step is form \eqref{Eqn:converse_alpha_r_prob}. Now, using Lemma \ref{Lem:converse}, we must have
\[
\frac{m}{n}\ge\frac{1-\epsilon/3}{\alpha}=\frac{1-\epsilon/3}{d(Y)+\kappa}.
\]
Since $\kappa>0$ is arbitrary, $m/n\ge\frac{1-\epsilon/3}{d(Y)}$. Hence, a necessary condition for achieving vanishing decoding error as $m,n\to\infty$ with $m/n\to\delta$ is
\[
\delta\ge\frac{1}{d(Y)}.
\]
This finishes the proof of Theorem \ref{Lem:converse2}.
\subsection{Proof of Lemma \ref{Lem:converse_converge}}\label{Sec:proof_Lemma10}

Let the SVD of $\bm{A}$ be $\bm{A}=\bm{U\Sigma V}^\UT$, where $\bm{U}\in\mathbb{R}^{m\times m}$, $\bm{\Sigma}\in\mathbb{R}^{m\times n}$ and $\bm{V}\in\mathbb{R}^{n\times n}$. By rotational invariance of $\bm{U}$ and independence between $\bm{U}$ and $\bm{\Sigma V}^\UT\bm{x}$, 
\[
\begin{split}
\bm{Ax}&=\bm{U\Sigma V}^\UT\bm{x}\\
&\overset{d}{=}\|\bm{\Sigma V}^\UT\bm{x}\|\cdot \bm{Ue}_1\\
&=\|\bm{\Sigma V}^\UT\bm{x}\|\cdot \bm{u}_1\\
&\overset{d}{=}{\frac{\|\bm{\Sigma V}^\UT\bm{x}\|}{\|\bm{g}_1\|}}\cdot \bm{g}_1,\quad \bm{g}_1\sim\mathcal{N}(\bm{0},\bm{I}_m),
\end{split}
\]
where $\bm{u}_1$ denotes the first column of $\bm{U}$ and $\overset{d}{=}$ means that the random vectors on the left and right hand sides have the same distribution. To show the desired weak convergence, it suffices to prove
\[
\alpha\Mydef \frac{\|\bm{\Sigma V}^\UT\bm{x}\|}{\|\bm{g}_1\|}\overset{a.s.}{\longrightarrow}1.
\]
To see this, we note that weak convergence is equivalent to convergence under bounded Lipschitz continuous test function $\phi$ \cite[Lemma 2.2]{vaart_1998}:
\[
\frac{1}{m}\sum_{j=1}^m \phi(\alpha\cdot g_{1j})\overset{a.s.}{\longrightarrow}\mathbb{E}[\phi(G_1)],\quad G_1\sim\mathcal{N}(0,1).
\]
On the other hand,
\[
\begin{split}
&\Bigg|\frac{1}{m}\sum_{j=1}^m \phi(\alpha g_{1j})-\mathbb{E}[\phi(G_1)]\Bigg|\\
 &=\Bigg|\frac{1}{m}\sum_{j=1}^m \phi(\alpha g_{1j})-\frac{1}{m}\sum_{j=1}^m \phi( g_{1j})+\frac{1}{m}\sum_{j=1}^m \phi( g_{1j})-\mathbb{E}[\phi(G_1)]\Bigg|\\
&\le \frac{1}{m}\sum_{j=1}^m |\phi(\alpha g_{1j})-\phi(g_{1j})|+\left|\frac{1}{m}\sum_{j=1}^m |\phi( g_{1j})-\mathbb{E}[\phi(G_1)]\right|\\
&\le L_{\text{lip}}|\alpha-1| \frac{1}{m}\sum_{j=1}^m |g_{1j}|+\left|\frac{1}{m}\sum_{j=1}^m |\phi( g_{1j})-\mathbb{E}[\phi(G_1)]\right|
\end{split}
\]
where the last step follows from the Lipschitz continuity of $\phi$ and $L_{\text{lip}}$ denotes the Lipschitz constant. Clearly, the desired convergence holds if $\alpha\overset{a.s.}{\longrightarrow}1$.

From the above discussions, we just need to prove
\[
 \alpha:=\frac{\|\bm{\Sigma V}^\UT\bm{x}\|}{\|\bm{g}_1\|}\overset{a.s.}{\longrightarrow}1.
\]
Since $\bm{V}$ is Haar distributed, we have $\bm{V}^\UT\bm{x}\overset{d}{=}\|\bm{x}\|/\|\bm{g}_2\|\cdot \bm{g}_2$ where $\bm{g}_2\sim\mathcal{N}(\mathbf{0},\bm{I}_n)$. Hence,
\[
\begin{split}
 \frac{\|\bm{\Sigma V}^\UT\bm{x}\|}{\|\bm{g}_1\|}\overset{d}{=}\frac{\|\bm{x}\|}{\|\bm{g}_1\|\|\bm{g}_2\|} \cdot\|\bm{\Sigma g}_2\|.
\end{split}
\]
As $\bm{x}\in\mathbb{R}^n$, $\bm{g}_1\in\mathbb{R}^m$ and $\bm{g}_2\in\mathbb{R}^n$ all have i.i.d. entries with unit variance, $\sqrt{m}\|\bm{x}\|/(\|\bm{g}_1\|\|\bm{g}_2\|) \overset{a.s.}{\longrightarrow}1$. Hence, it remains to prove $m^{-1/2}\|\bm{\Sigma g}_2\|\overset{a.s.}{\longrightarrow}1$. To this end, we shall prove
\BE\label{Eqn:Wasserstein_final}
\frac{1}{n} \sum_{i=1}^n \sigma_i^2\cdot g_{1i}^2=\frac{1}{n} \sum_{i=1}^n \lambda_i\cdot g_{1i}^2\overset{a.s.}{\longrightarrow}\delta,
\EE
which, together with the continuous mapping theorem, implies the desired result. 

Similar to \cite[Corollary 1]{KeigoV2}, we use Lyons' strong law of large numbers \cite{Lyons1988StrongLO} to prove \eqref{Eqn:Wasserstein_final}. We first show that the following holds conditional on $\{\lambda_i\}$:
\BE\label{Eqn:SLLN_conditional}
\lim_{n\to\infty}\frac{1}{n} \sum_{i=1}^n \lambda_i\cdot g_{1i}^2-\frac{1}{n} \sum_{i=1}^n \lambda_i\overset{a.s.}{\longrightarrow}0.
\EE
From \cite[Theorem 6]{Lyons1988StrongLO}, it suffices to verify 
\BE\label{Eqn:summaribility}
\sum_{n=1}^\infty\frac{1}{n^2}\sqrt{\text{Var}\Big(\sum_{i=1}^n \lambda_i g_{1i}^2\Big)} <\infty.
\EE
Since $\{g_{1i}\}$ are i.i.d. standard Gaussian,
\[
\begin{split}
\text{Var}\Big(\sum_{i=1}^n \lambda_i g_{1i}^2\Big) &= \sum_{i=1}^n \text{Var}(\lambda_i g_{1i}^2) = 2\sum_{i=1}^n \lambda_i^2.
\end{split}
\]
We have assumed $\frac{1}{n}\sum_{i=1}^n \lambda_i^2\overset{a.s.}{\longrightarrow}\mathbb{E}[\Lambda^2]<\infty$. Hence, for any $C>\mathbb{E}[\Lambda^2]$, the following holds for all sufficiently large $n$
\[
\text{Var}\Big(\sum_{i=1}^n \lambda_i g_{1i}^2\Big) < 2nC.
\]
Hence, \eqref{Eqn:summaribility} is satisfied and so \eqref{Eqn:SLLN_conditional} holds. On the other hand, from Assumption (A.2), weak convergence together with $\frac{1}{n}\sum_{i=1}^n \lambda_i^2\overset{a.s.}{\longrightarrow}\mathbb{E}[\Lambda^2]$ implies convergence in Wasserstein distance of order two \cite{feng2021unifying}. This further implies convergence in Wasserstein distance of order one \cite{feng2021unifying}, and so $\frac{1}{n}\sum_{i=1}^n \lambda_i\overset{a.s.}{\longrightarrow}\mathbb{E}[\Lambda]=\delta$. Putting things together proves \eqref{Eqn:Wasserstein_final}.

\section{Proof of Theorem \ref{The:noise_opt}}\label{App:noise_opt}

We begin with an auxiliary lemma and then provide the main proof in Section \ref{Sec:main_lemma2}.

\subsection{An auxiliary lemma}
\begin{lemma}\label{Lem:noise_sen_aux1}
Let $\bm{a}$ be the first row of $\bm{A}$, and $z=(\bm{Ax})_1$, $y=f(z+w)$, where $w\sim\mathcal{N}(0,\sigma_w^2)$ and $\bm{w}\Perp (\bm{A,\bm{x}})$. We have
\[
\limsup_{\sigma_w\to0}\frac{I(z;y|\bm{a})}{\log(\sigma_w^{-2})}\le1-\mathbb{E}\left[ \mathbb{P}\{z_a\in f^{-1}(\mathcal{Q}_f)\} \right],
\]
where $z_a\sim\mathcal{N}(0,\|\bm{a}\|^2)$.
\end{lemma}
\begin{IEEEproof}
We use $z_{a}$ to denote a random variable with distribution $P_{z|\bm{a}}$,  $y_a=f(z_a+w)$ and $z_a\Perp w$. Note that $z=\bm{a}^\UT\bm{x}$, and so $z_{{a}}\sim \mathcal{N}(\mathbf{0},\|\bm{a}\|^2)$. Then, by the reverse Fatou lemma, we have
\[
\begin{split}
\limsup_{\sigma_w\to0}\frac{I(z;y|\bm{a})}{\frac{1}{2}\log\sigma_w^{-2}}&=\limsup_{\sigma_w\to0}\mathbb{E}\left[\frac{I(z_a;y_a)}{\frac{1}{2}\log\sigma_w^{-2}}\right]\\
&\le \mathbb{E}\left[\limsup_{\sigma_w\to0}\frac{I(z_a;y_a)}{\frac{1}{2}\log\sigma_w^{-2}}\right].
\end{split}
\]
In what follows, we prove
\BE\label{Eqn:STEP_2}
\limsup_{\sigma_w\to0}\frac{I(z_a;y_a)}{\frac{1}{2}\log\sigma_w^{-2}} \le \mathbb{P}\{ z_{a} \in\mathbb{R}\backslash \mathcal{Q}_f \},
\EE
where $z_a\in\mathcal{N}(0,\|\bm{a}\|^2)$. For convenience, define
\BE\label{Eqn:Z_sigma_def}
\begin{split}
p_{a} \Mydef z_{a}+w.
\end{split}
\EE
Using this notation, $y_{a}=f(p_{a})$. We introduce an auxiliary random variable
\[
Q\Mydef \mathbb{I}(p_{a} \in f^{-1}(\mathcal{Q}_f)).
\]
Then,
\BE\label{Eqn:STEP_2_main}
\begin{split}
&I(z_a;f(z_a+w)) \\
&\le I(z_a;f(z_a+w),Q)\\
&= I(z_a;Q) + I(z_a;f(z_a+w)|Q)\\
&= I(z_a;Q) +\mathbb{P}\{Q=1\}\cdot I(z_a;f(z_a+w)|Q=1)\\
&+\mathbb{P}\{Q=0\}\cdot I(z_a;f(z_a+w)|Q=0)\\
&\le 1 + \log(|\mathcal{Q}_f|)+\mathbb{P}\{Q=0\}\cdot I(z_a;f(z_a+w)|Q=0)
\end{split}
\EE
where the last inequality follows from the fact that $I(z_a;Q)\le 1$ for the binary random variable $Q$, and $I(z_a;f(z_a+w)|Q=1)\le \log(|\mathcal{Q}_f|)$ since $f(z_a+w)$ takes at most $|\mathcal{Q}_f|$ different values conditional on $Q=1$.

Since $z_a\in\mathcal{N}(0,\|\bm{a}\|^2)$ and $w\sim \mathcal{N}(0,\sigma_w^2)$, we can represent $p_a=z_a+w$ as
\[
z_a \overset{d}{=}\underbrace{\frac{\|\bm{a}\|^2}{\|\bm{a}\|^2+\sigma_w^2}}_{\alpha}\cdot p_a + \underbrace{\sqrt{\frac{\|\bm{a}\|^2\sigma_w^2}{\|\bm{a}\|^2+\sigma_w^2}}}_{\beta}\cdot N
\]
where $N$ is standard Gaussian and independent of $p_a$. Hence, $N$ is still independent of $p_a$ conditioned on $Q=0$. Hence, the conditional distribution of $(z_a,p_a)$ is characterized by
\[
\tilde{z}_a=\alpha\cdot \tilde{p}_a + \beta N,
\]
where $\tilde{p}_a \sim P_{p_a|Q=0}=P_{p_a|p_a\in\mathbb{R}\backslash f^{-1}(\mathcal{Q}_f)}$. Since $f(\tilde{p}_a)\to \tilde{p}_a\to z_a$ forms a Markov chain, by data processing inequality, we have
\[
\begin{split}
I(z_a;f(z_a+w)|Q=0)&=I(\tilde{z}_a;f(\tilde{p}_a))\le I(\tilde{z}_a;\tilde{p}_a)\\
&=\frac{1}{2}\log\left( 1 + \frac{\alpha^2}{\beta^2}\mathbb{E}[\tilde{p}_a^2] \right)\\
&=\frac{1}{2}\log\left(1+\frac{\|\bm{a}\|^2}{\|\bm{a}\|^2+\sigma_w^2} \frac{\mathbb{E}[\tilde{p}_a^2]}{\sigma_w^2}\right) .
\end{split}
\]
It is easy to show that $\mathbb{E}[\tilde{p}_a^2]$ converges to a positive constant as $\sigma_w\to0$, and
\BS \label{Eqn:STEP_2_main2}
\BE
\limsup_{\sigma_w\to0}\frac{I(z_a;f(z_a+w)|Q=0)}{\frac{1}{2}\log(\sigma_w^{-2})}\le 1,
\EE
and
\BE
\lim_{\sigma_w\to0}\mathbb{P}\{Q=0\}=\mathbb{P}\{z_a\in\mathbb{R}\backslash \mathcal{Q}_f\}.
\EE
\ES
Combining \eqref{Eqn:STEP_2_main} and \eqref{Eqn:STEP_2_main2} proves Lemma \ref{Lem:noise_sen_aux1}.
\end{IEEEproof}

\subsection{Main proof for Theorem \ref{The:noise_opt}}\label{Sec:main_lemma2}

The proof is analogous to \cite[Theorem 9]{Wu12}. Notice that we assumed $\bm{x}\sim\mathcal{N}(\mathbf{0},\bm{I})$ in Theorem \ref{The:noise_opt}.

Let $R_X(D)$ be the rate distortion functions of $P_X$ with mean square error distortion:
\[
\begin{split}
R_X(D)=\underset{\mathbb{E}[d(X,\hat{S})]\le D,X\sim P_X}{\inf}\ I(X;\hat{S}),
\end{split}
\]
where $d(X,\hat{S}):=(X-\hat{S})^2$, $I(X;\hat{S})$ denotes the mutual information between $X$ and $\hat{S}$, and the infimum in the above definition is over the transition probability $P_{\hat{S}|X}$ subject to average distortion constraint. Notice that $R_X(D)$ can be equivalently defined as \cite[Theorem 9.6.1]{gallager1968information}
\[
R_X(D)=\underset{\mathbb{E}[d_n(\bm{x},\hat{\bm{s}})]\le D,\{x_i\}\overset{i.i.d.}{\sim} P_X}{\inf}\ \frac{1}{n}I(\bm{x};\hat{\bm{s}}),
\]
where $d_n(\bm{x},\hat{\bm{s}}):=\frac{1}{n}\sum_{i=1}^n (x_i-\hat{s}_i)^2$.

Consider the MMSE estimator $\hat{\bm{x}}=\mathbb{E}[\bm{x}|\bm{y},\bm{A}]$ with mean square distortion
\[
D_n(\sigma_w)\Mydef \frac{1}{n}\mmse(\bm{x}|\bm{y},\bm{A})=\frac{1}{n}\sum_{i=1}^n\mathbb{E}(x_i-\hat{x}_i)^2,
\]
where the expectation is over the joint distribution of $\bm{x}$ and $\hat{\bm{x}}$.
In what follows, we will sometimes write $D_n(\sigma_w)$ as $D_n$ for notational convenience. By the definition of rate distortion functions,
\BE\label{Eqn:rate_distort_0}
\begin{split}
n\cdot R_X(D_n) & \le I(\bm{x},\hat{\bm{x}}).
\end{split}
\EE
Denote by $I(\hat{\bm{x}};\bm{A},\bm{x})$ the mutual information between $\hat{\bm{x}}$ and $(\bm{A},\bm{x})$. We have
\[
\begin{split}
I(\bm{x};\hat{\bm{x}}) &\le I(\bm{x};\bm{A},\hat{\bm{x}}) \\
&= \underbrace{I(\bm{x};\bm{A})}_{0}+I(\bm{x};\hat{\bm{x}}|\bm{A}).
\end{split}
\]
Hence,
\BE\label{Eqn:rate_distort_02}
n\cdot R_X(D_n) \le I(\bm{x};\hat{\bm{x}}|\bm{A}).
\EE

Denote $\bm{z}\Mydef\bm{Ax}$ and $\bm{y}=f(\bm{z}+\bm{w})$. For every realization of $\bm{A}$, we have the following Markov chain:
\[
\bm{x}\to\bm{z}_A\to\bm{y}_A\to\hat{\bm{x}}_A,
\]
where the subscript ``$_A$'' is added to emphasize the fact that $\bm{A}$ is fixed.
From data processing inequality \cite[Theorem 4.3.3]{gallager1968information}, we have $I(\bm{x};\hat{\bm{x}}_A)  \le I\left(\bm{z}_A;\bm{y}_A\right)$. Further averaging over $\bm{A}$ yields
\BE\label{Eqn:rate_distort_1}
\begin{split}
	I(\bm{x};\hat{\bm{x}}|\bm{A}) & \le I\left(\bm{z};\bm{y}|\bm{A}\right)\\
	&\le \sum_{i=1}^m I(z_i;y_i|\bm{A})\\
	&=\sum_{i=1}^m I(z_i;y_i|\bm{a}_i)\\
	&=m\cdot I(z;y|\bm{a}),
	\end{split}
\EE
where $\bm{a}_i$ denotes the $i$-th row of $\bm{A}$, the second inequality follows from \cite[Eq. (7.2.19)]{gallager1968information} (note that $\{y_{i}\}$ are conditionally independent given $\bm{z}$), and in the last inequality we dropped the subscripts as the joint distributions of $\{(\bm{z}_i,\bm{y}_i,\bm{a}_i)\}$ are identical due to the rotationally-invariance of $\bm{A}$. Combining \eqref{Eqn:rate_distort_02} and \eqref{Eqn:rate_distort_1} gives us
the following lower bound on $m/n$:
\BE\label{Eqn:rate_distort_15}
\frac{m}{n}\ge \frac{R_X(D_n)}{ I(z;y|\bm{a})}.
\EE

Now, suppose that
\[
M^\ast(X,f,\Lambda,\delta)=  \sup_{\sigma_w}\, \limsup_{n\to\infty}\frac{\frac{1}{n}\mmse(\bm{x}|\bm{y},\bm{A})}{\sigma^2_w}<\infty.
\]
Then, there exits $C>0$ such that the following holds for sufficiently large $n$
\[
D_n=\frac{1}{n}\mmse(\bm{x}|\bm{y},\bm{A})\le C\cdot\sigma_w^2,\quad\forall \sigma_w>0.
\]
The following arguments are similar to the proof of \cite[Theorem 9]{Wu12}. Let $R_X^{-1}$ be the inverse function of $R_X$. (Since $R_X$ is a monotonically decreasing function, its inverse exists.) We have
\[
R_X(D_n) \ge R_X(C\cdot \sigma_w^2),\quad\forall \sigma_w>0.
\]
Hence, the following holds for any $\sigma_w>0$,
\BE\label{Eqn:rate_distort_bound1}
\begin{split}
\frac{n}{m}&\le \frac{ I(z;y|\bm{a})}{R_X(C\cdot \sigma_w^2)}\\
&=  \frac{I(z;y|\bm{a})}{\frac{1}{2}\log \frac{1}{C\cdot \sigma_w^2}}\cdot \frac{\frac{1}{2}\log\frac{1}{C\cdot\sigma_w^2}}{ R_X(C\cdot \sigma_w^2)}.
\end{split}
\EE
When $X\sim\mathcal{N}(0,1)$, we have \cite{Wu10}
\BE\label{Eqn:rate_distort_bound2}
\lim_{\sigma_w\to0} \frac{R_X(C\cdot \sigma_w^2)}{\frac{1}{2}\log \frac{1}{C\cdot \sigma_w^2}} = 1.
\EE

Further, from Lemma \ref{Lem:noise_sen_aux1}, we have
\BE\label{Eqn:rate_distort_bound3}
\limsup_{\sigma_w\to0} \frac{ I(z;y|\bm{a})}{\frac{1}{2}\log\frac{1}{C\cdot\sigma_w^2}} \le 1-\mathbb{E}_{\bm{a}}[ \mathbb{P}\{ z_{{a}} \in f^{-1}(\mathcal{Q}_f )\} ],
\EE
where $z_a\sim\mathcal{N}(0,\|\bm{a}\|^2)$, and $\bm{a}$ has the same distribution as the first row of $\bm{A}$. Note that the proof of Lemma \ref{Lem:converse_converge} shows that $\|\bm{As}\|\overset{a.s.}{\to}1$ as $m,n\to\infty$ with $m/n\to\delta$, whenever $\|\bm{s}\|\to 1$. Hence, $\|\bm{a}\|\overset{a.s.}{\to}1$ where $\bm{a}$ is an arbitrary row of $\bm{A}$. It is easy to show that $\mathbb{P}\{z_a\in f^{-1}(\mathcal{Q}_f)\}$ is {a continuous function of $\|\bm{a}\|^2$}, and by continuous mapping theorem we have $\mathbb{P}\{z_a\in f^{-1}(\mathcal{Q}_f)\}\overset{a.s.}{\to} \mathbb{P}\{Z\in \mathcal{Q}_f\}$ where $Z\sim\mathcal{N}(0,1)$. Then by dominated convergence theorem, the following holds as $m,n\to\infty$ with $m/n\to\delta$,
\BE\label{Eqn:rate_distort_bound4}
\mathbb{E}_{\bm{a}}[ \mathbb{P}\{ z_{{a}} \in f^{-1}(\mathcal{Q}_f )\} ]\to \mathbb{P}\{Z\in \mathcal{Q}_f\}.
\EE
Combining \eqref{Eqn:rate_distort_bound1}-\eqref{Eqn:rate_distort_bound4} yields our desired result and concludes the proof of Theorem \ref{The:noise_opt}.

\section{Proof of Lemma \ref{Eqn:LFD2}} \label{App:Smoothing}

We use the smoothing argument of \cite[Theorem 1]{Zheng17}. Roughly speaking, we construct a sequence of smoothed functions $\tilde{\eta}_z$ (indexed by $\xi,M,\sigma$; see \eqref{Eqn:eta_smooth_final}), and show that the performance of the corresponding GLM-EP-app algorithm tends to the predicted performance of GLM-EP as $\xi,M,\sigma$ approaches a certain limit. This implies the performance of GLM-EP-app could be made arbitrarily close to the predicted one with proper choice of $\xi,M,\sigma$. 

\begin{remark}
We emphasize that the GLM-EP-app algorithm is introduced mainly for performance analysis purposes. In practice, GLM-EP is preferable. Our simulations suggest that the asymptotic prediction is accurate even for the original GLM-EP under wide choices of $f$ (including the quantization function).
\end{remark}

As many steps of the proof are the same as \cite[Theorem 1]{Zheng17}, we only sketch the main idea here.

\subsection{Constructions of $\tilde{\eta}_z$ and $C_t$}
For brevity, we omitted the argument $v$ in the notation $\eta_z(z_r,y,v)$ throughout this section. Let $\mathcal{Q}_f\Mydef \{y_q: 1\le q\le Q\}$ (where $Q<\infty$) be the set for which $f^{-1}$ contains an interval. Let $\xi<\frac{1}{2} \min\ \{|y_p-y_q|,p\neq q\}$ and define
\BS
\BE\label{Eqn:eta_smooth_1}
\eta_z^{\xi}(z_r,y)\Mydef 
\begin{cases}
\eta_{q}(z_r) & z_r\in\mathbb{R},y\in (y_q-\xi,y_q+\xi),\forall q\\
\eta_z(z_r,y) &z_r\in\mathbb{R},y\in \mathcal{Y}\backslash\bigcup_{q=1}^Q (y_q-\xi,y_q+\xi)\\
0&z_r\in\mathbb{R},y\in\mathbb{R}\backslash \mathcal{Y}
\end{cases}
\EE
where we denoted
\BE
\eta_q(z_r)\Mydef \eta_z(z_r,y_q).
\EE
\ES
Here, we extended the definition of $\eta_z$ at the isolated points $\{y_1,\ldots,y_Q\}$ to their neighborhoods. This treatment ensures $\eta_z^{\xi}(z_r,y)$ to be continuous at $(z_r,y_q)$, which is a useful property for our analysis. We apply an additional truncation to $\eta_z^{\xi}(z_r,y)$ (where $M>\max\{|y_1|,\ldots,|y_Q|\}$):
\BE\label{Eqn:eta_smooth_2}
\eta_z^{\xi,M}(z_r,y)\Mydef \eta_z^{\xi}(z_r,y)\cdot \mathbb{I}_{[-M,M]^2}(z_r,y),
\EE
where $\mathbb{I}_{[-M,M]^2}(z_r,y)$ is an indicator function that equals one when $(z_r,y)\in[-M,M]^2$ and zero elsewhere.
Finally, we smooth $\eta_z^{\xi,M}(z_r,y)$ by convolving it with a Gaussian kernel\footnote{The smoothing parameter $\sigma$ should not be confused with $\sigma_w$, which denotes the noise variance in Section \ref{Sec:noisy_converse_1}.}:
\BE\label{Eqn:eta_smooth_final}
\begin{split}
&\eta_z^{\sigma,\xi,M}(z_r,y)\\
&\Mydef \eta_z^{\xi,M}(z_r,y)\star \phi_\sigma(z_r,y) \\
&=\iint_{\mathbb{R}^2} \eta^{\xi,M}(s,t)\cdot \frac{1}{2\pi \sigma^2} \exp\left( -\frac{(s-z_r)^2+(t-y)^2}{2\sigma^2}\right)dsdt\\
\end{split}
\EE

Some useful properties of $\eta_z^{\sigma,\xi,M}(z_r,y)$ and $\eta_z$ are given in Section \ref{Sec:app_smooth_aux} (see Lemma \ref{Lem:eta_smooth_pro} and Lemma \ref{Lem:poly}).\vspace{5pt}

In the GLM-EP-app algorithm (see \eqref{Eqn:GLM_EP_app}), the function $\tilde{\eta}_z$ and the constant $C_t$ are given by
\BS\label{Eqn:choices_eta_C}
\begin{align}
&\tilde{\eta}_z(z_r,y,v_r) = \eta_z^{\sigma,\xi,M}(z_r,y,v_r),
\intertext{and}
&C_t = \frac{1}{1-\mathbb{E}\left[\eta'_z(Z_r^{t-1},Y,V_r^{-1})\right]}\label{Eqn:choices_eta_C2}
\end{align}
\ES
where the expectation in \eqref{Eqn:choices_eta_C2} are taken with respect to $Z_r^{t-1}\sim\mathcal{N}(0,1-V_r^{t-1})$, $Z=Z_r+\mathcal{N}(0,V_r^{t-1})$ and $Y=f(Z)$. Notice that $C_t$ depends on the the original function $\eta_z$, not the smoothed and truncated function $\eta_z^{\sigma,\xi,M}$. This choice is for the purpose of simplifying our analysis.

Before we move to the proof sketch in Section \ref{Sec:smooth_main}, we present some auxiliary results in the next section.
\subsection{Auxiliary results} \label{Sec:app_smooth_aux}

\begin{lemma}\label{Lem:eta_smooth_pro}
Let $v_r,\xi,M,\sigma>0$. The following hold
\begin{itemize}
\item[(P.1)] $\eta_z^\xi(z_r,y)$ is continuous a.e. with respect to the Lebesgue measure. Further, $\eta_z^{\xi,M}(z_r,y)$ is a.e. bounded;
\item[(P.2)] $\eta_z^{\sigma,\xi,M}(z_r,y)$ is Lipschitz continuous and bounded on $\mathbb{R}^2$;
\item[(P.3)] $\lim_{\sigma\to0} \eta_z^{\sigma,\xi,M}(z_r,y)=\eta_z^{\xi,M}(z_r,y)$ whenever $\eta_z^{\xi,M}$ is continuous at $(z_r,y)$.
\end{itemize}
\end{lemma}\vspace{5pt}

\begin{IEEEproof}
\textit{Proof of (P.1):} We note that $\eta_q(z_r)$ ($q=1,\ldots,Q$) is a continuous function of $z_r\in\mathbb{R}$:
\[
\eta_q(z_r)\Mydef\eta_z(z_r,y_q)=\frac{\int_{f^{-1}(y_q)}u\cdot \mathcal{N}(u;z_r,v_r)du}{\int_{f^{-1}(y_q)}\mathcal{N}(u;z_r,v_r)du}.
\]
By definition of $y_q$, $f^{-1}(y_q)$ contains an interval (could be union of intervals), and it is straightforward to show that $\eta_q(z_r)$ is continuous on $\mathbb{R}$. 

When $y\in\mathcal{Y}\backslash\mathcal{Q}_f$, $f^{-1}(y)$ is a finite set and we have (see \eqref{Eqn:eta_z_def})
\[
\begin{split}
\eta_z(z_r,y)= \frac{\sum_{u_i\in f^{-1}(y)} u_i\cdot \mathcal{N}(u_i;z_r,v)}{\sum_{u_i\in f^{-1}(y)}\mathcal{N}(u_i;z_r,v)}.
\end{split}
\]
By the piecewise smooth assumption of $f$, it can be shown that $\mathcal{Y}\backslash\mathcal{Q}_f$ can be further decomposed into several non-overlapping {intervals}, denoted as $\mathcal{Y}\backslash\mathcal{Q}_f=\bigcup_{j=1}^J\mathcal{Y}_j$ (where $J<\infty$), such that $\eta_z(z_r,y)$ is continuous on $\mathbb{R}\times \mathcal{Y}_j$, $\forall j$. This is due to the fact that each point of $f^{-1}(y)$ is a continuous function of $y$ for $y\in\mathcal{Y}_j$. Specifically, it is possible to write $f^{-1}(y)$ as
\[
f^{-1}(y) = \{ F^1_{j}(y),\ldots,F_{j}^{K_j}(y) \},\quad \forall y\in\mathcal{Y}_j,
\]
where $K_j<\infty$, and each $F_j^{k}(y)$ is a continuous function of $y$ (by piecewise continuity of $f$). Hence,
\[
\begin{split}
\eta_z(z_r,y)= \frac{\sum_{k=1}^{K_j} F^k_{j}(y)\cdot \mathcal{N}(F^k_{j}(y);z_r,v)}{\sum_{k=1}^{K_j} \mathcal{N}(F^k_{j}(y);z_r,v)},\, \forall (z_r,y)\in\mathbb{R}\times\mathcal{Y}_j,
\end{split}
\]
and it is continuous on \textit{the interior} of $\mathbb{R}\times\mathcal{Y}_j$. As an example, consider $f$ given in \eqref{Eqn:f_example3} (see illustration on the left panel of Figure \ref{Fig:finite_f}). In this case, 
\[
f^{-1}(y)=
\begin{cases}
\{-y-1,y+1\}, & \text{for } y>1\\
\{-y-1,y+1,-y,y\}, & \text{for } 0\le y\le1
\end{cases}
\]
It can be shown that $\eta_z(z_r,y)$ is continuous on $\mathbb{R}\times (1,\infty)$ and $\mathbb{R}\times (0,1)$. 

The claimed a.e. continuity of $\eta_z^\xi$ (see definition in \eqref{Eqn:eta_smooth_1}) follows from the above properties of $\eta_z(z_r,y)$.

Since $\eta_z^{\xi}(z_r,y)$ is continuous almost everywhere (with respect to the Lebesgue measure), $\eta_z^{\xi,M}(z_r,y)= \eta^{\xi}(z_r,y)\cdot \mathbb{I}_{[-M,M]^2}(z_r,y)$ is bounded almost everywhere. Let $M'<\infty$ denote this a.e. bound of $|\eta_z^{\xi,M}|$. Then, the smoothed function $|\eta_z^{\sigma,\xi,M}(z_r,y)|$ is upper bounded by $M'$ on $\mathbb{R}^2$.\vspace{3pt}

\textit{Proof of (P.2):}
With slight abuse of notations, let $\phi_\sigma(s,t)$ denote the bivariate and univariate Gaussian pdf functions respectively. Namely, $\phi_\sigma(s,t)=\phi_\sigma(s)\phi_h(t)$, where $\phi_\sigma(s):=\frac{1}{\sqrt{2\pi\sigma^2}}\exp(-s^2/(2\sigma^2))$. To prove Lipschitz continuity, note that
\[
\begin{split}
&|\eta_z^{\sigma,\xi,M}(z_1,y_1)-\eta_z^{\sigma,\xi,M}(z_2,y_2)| \\
&\le \iint \eta_z^{\xi,M}(s,t)\big|\phi_\sigma(z_1-s)\phi_\sigma(y_1-t)\\
&\quad -\phi_\sigma(z_2-s)\phi_\sigma(y_2-t)\big| dsdt\\
&\le 8M'M^2\|\phi_\sigma\|_\infty \|\phi_\sigma'\|_\infty\cdot \|(y_1,z_1)-(y_2,z_2)\|,
\end{split}
\]
where we used
\[
\begin{split}
&\left|\phi_\sigma(z_1-s)\phi_\sigma(y_1-t)-\phi_\sigma(z_2-s)\phi_\sigma(y_2-t)\right| \\
&\le\|\phi_\sigma\|_\infty \big(|\phi_\sigma(y_1-t)-\phi_\sigma(y_2-t)|\\
&\quad +|\phi_\sigma(z_1-t)-\phi_\sigma(z_2-t)|\big)\\
&\le \|\phi_\sigma\|_\infty \|\phi_\sigma'\|_\infty\cdot(|y_1-y_2|+|z_1-z_2|)\\
&\le2\|\phi_\sigma\|_\infty \|\phi_\sigma'\|_\infty\cdot \|(y_1,z_1)-(y_2,z_2)\|.
\end{split}
\]
Hence, the function $\eta^{\sigma,\xi,M}$ is Lipschitz continuous. The boundedness of $\eta_z^{\sigma,\xi,M}$ follows from the fact that $\eta_z^{\xi,M}$ is a.e. bounded (P.1), and the Gaussian convolution kernel is absolutely continuous w.r.t the Lebesgue measure.
\vspace{3pt}

\textit{Proof of (P.3):} Let $(z_0,y_0)$ be a point at which $\eta^{\xi,M}$ is continuous. Since $\eta^{\xi,M}$ is bounded almost everywhere, we could apply DCT to get
\[
\begin{split}
&\lim_{\sigma\to0} \eta^{\sigma,\xi,M}(z_0,y_0)\\
&=\lim_{\sigma\to0}\iint \eta^{\xi,M}(z_0+s,y_0+t) \frac{1}{2\pi \sigma^2} \exp\left( -\frac{s^2+t^2}{2\sigma^2}\right)dsdt\\
&=\lim_{\sigma\to0}\iint \eta^{\xi,M}(z_0+\sigma s,y_0+\sigma t) \frac{1}{2\pi} \exp\left( -\frac{s^2+t^2}{2}\right)dsdt\\
&=\iint\lim_{\sigma\to0}\eta^{\xi,M}(z_0+\sigma s,y_0+\sigma t) \frac{1}{2\pi} \exp\left( -\frac{s^2+t^2}{2}\right)dsdt\\
\end{split}
\]
Since $\eta^{\xi,M}$ is continuous at $(z_0,y_0)$, we have
\[
\lim_{\sigma\to0}\eta^{\xi,M}(z_0+\sigma s,y_0+\sigma t) = \eta^{\xi,M}(z_0,y_0),\, \forall (s,t)\in\mathbb{R}^2.
\]
Combining the two steps completes the proof.
\end{IEEEproof}
\vspace{10pt}
\begin{lemma}\label{Lem:poly}
Let $v_r>0$. There exists a constant $C>0$ such that
\BE\label{Eqn:}
\eta_z(z_r,f(z),v_r) < C\cdot \left( 1+ \|(z_r,z)\|\right),\quad \forall (z_r,z)\in\mathbb{R}^2.
\EE
\end{lemma}

\begin{IEEEproof}
We differentiate between two cases: $f(z)\in \mathcal{Q}_f$ and $f(z)\in\mathcal{Y}\backslash\mathcal{Q}_f$, where $\mathcal{Q}_f=\{y_1,\ldots,y_Q\}$ correspond to the flat sections of $f$.\vspace{5pt}

\textit{Case 1: $f(z)\in \mathcal{Q}_f$.} Assume $f(z)=y_i$ and denote $\mathcal{I}_i\Mydef f^{-1}(y_i)$. In what follows, we will prove
\[
\eta_z(z_r,y_i,v_r) =\frac{\int_{\mathcal{I}_i}u\cdot \mathcal{N}(u;z_r,v_r)du}{\int_{\mathcal{I}_i}\mathcal{N}(u;z_r,v_r)du}< C_i \left( 1+ |z_r|\right).
\]
Suppose $\mathcal{I}_i$ can be written as $\mathcal{I}_i=\bigcup_{k=1}^K (a_k,b_k)$, where $a_k$ could be $-\infty$ and $b_k$ could be $\infty$ (we do not index $a_k,b_k$ by $i$ to simplify notation.) Then, 
\[
\begin{split}
\eta_z(z_r,y_i,v_r) &=\frac{\int_{\mathcal{I}_i}u\cdot \mathcal{N}(u;z_r,v_r)du}{\int_{\mathcal{I}_i}\mathcal{N}(u;z_r,v_r)du}\\
&=\sum_{k=1}^K\left(\frac{\int_{(a_k,b_k)}u\cdot \mathcal{N}(u;z_r,v_r)du}{\sum_{j=1}^K\int_{(a_j,b_j)}\mathcal{N}(u;z_r,v_r)du}\right)
\end{split}
\]
We have
\BE
\begin{split}
|\eta_z(z_r,y_i,v_r)  | &\le \sum_{k=1}^K\left(\frac{\left|\int_{(a_k,b_k)}u\cdot \mathcal{N}(u;z_r,v_r)du\right|}{\sum_{j=1}^K\int_{(a_j,b_j)}\mathcal{N}(u;z_r,v_r)du}\right)\\
&\le \sum_{k=1}^K\left(\frac{\left|\int_{(a_k,b_k)}u\cdot \mathcal{N}(u;z_r,v_r)du\right|}{\int_{(a_k,b_k)}\mathcal{N}(u;z_r,v_r)du}\right)
\end{split}
\EE
We bound the terms inside the summation seperately. First, assume both $a_k$ and $b_k$ are finite. Then,
\[
\begin{split}
&\frac{\left|\int_{(a_k,b_k)}u\cdot \mathcal{N}(u;z_r,v_r)du\right|}{\int_{(a_k,b_k)}\mathcal{N}(u;z_r,v_r)du} \\
&=\frac{\left|\int_{(a_k,b_k)}u\cdot \mathcal{N}(u;z_r,v_r)du\right|}{\int_{(a_k,b_k)}\mathcal{N}(u;z_r,v_r)du} \\
&=\left|z_r+\frac{\int_{(a_k,b_k)-z_r}t\cdot \mathcal{N}(t;0,v_r)dt}{\int_{(a_k,b_k)-z_r}\mathcal{N}(t;0,v_r)dt}\right|\\
&\le |z_r| + \max\{ |a_k-z_r|,|b_k-z_r| \}\\
&\le C'(1+|z_r|)
\end{split}
\]
Now, suppose $a_k=-\infty$. (The argument is similar for the case $b_k=\infty$.) We have
\[
\begin{split}
\left|\frac{\int_{(-\infty,b_k-z_r)}t\cdot \mathcal{N}(t;0,v_r)dt}{\int_{(-\infty,b_k-z_r)}\mathcal{N}(t;0,v_r)dt}\right| &= \left|\sqrt{v_r}\frac{\phi_1\left( \frac{b_k-z_r}{\sqrt{v_r}} \right)}{\Phi_1\left(  \frac{b_k-z_r}{\sqrt{v_r}} \right) } \right| \\
&\le C''(1+|z_r|)
\end{split}
\]
where $\phi_1$ and $\Phi_1$ denote the pdf and cdf functions of standard Gaussian distribution, respectively, and the last step is from mean value theorem together with the following elementary result
\[
\left|\left(\frac{\phi_1(x)}{\Phi_1(x)}\right)'\right|\le 1,\quad \forall x\in\mathbb{R}.
\]

\textit{Case 2: $f(z)\in\mathcal{Y} \backslash\mathcal{Q}_f$.}  In this case, $f^{-1}(f(z))$ is a finite set, and
\BE
\begin{split}
\eta_z(z_r,f(z),v_r) &=\frac{\sum_{u_i\in f^{-1}(f(z))} u_i\cdot \exp\left(-E_i\right)}{ \sum_{u_i\in f^{-1}(f(z))}\exp\left(-E_i\right) },
\end{split}
\EE
where
\[
E_i\Mydef \frac{(u_i-z_r)^2}{2v_r}.
\]
Hence,
\BE
\begin{split}
|\eta_z(z_r,f(z),v_r)| &\le\frac{\sum_{u_i\in f^{-1}(f(z))} |u_i|\cdot \exp\left(-E_i\right)}{ \sum_{u_i\in f^{-1}(f(z))}\exp\left(-E_i\right) }\\
&\le \sum_{u_i\in f^{-1}(f(z))} |u_i|\cdot \exp\left(-E_i+E_{\min}\right),
\end{split}
\EE
where
\BE
E_{\min}=\min\ \{E_j\}
\EE
From the piecewise assumption of $f$, we have that $|f^{-1}(f(z))|<K$ for all $f(z)\in \mathcal{Y} \backslash\mathcal{Q}_f$. It suffices to prove the following for $1\le i\le|f^{-1}(f(z))|$:
\[
|u_i|\cdot \exp\left(-E_i+E_{\min}\right) <C\left( 1+ \|(z,z_r)\|\right)\quad \forall (z,z_r)\in\mathbb{R}^2.
\]
Denote
\[
t_i\Mydef \exp\left(-E_i+E_{\min}\right)=\exp\left(- \frac{(u_i-z_r)^2}{2v_r}+E_{\min}\right).
\]
(As $E_i\ge E_{\min}$, we have $0<t_i\le 1$.) From this definition,
\[
|u_i-z_r|=\sqrt{2v_r\cdot\left(E_{\min}+\log  \frac{1}{t_i}\right) }.
\]
Hence,
\[
|u_i|\le |z_r|+\sqrt{2v_r\cdot\left(E_{\min}+\log  \frac{1}{t_i}\right) }.
\]
Then,
\[
\begin{split}
&|u_i|\cdot \exp\left(-E_i+E_{\min}\right) \\
&=|u_i|\cdot t_i\\
&\le |z_r|\cdot t_i+ \sqrt{2v_r\cdot\left(t_i^2\cdot E_{\min}+t_i^2\cdot\log  \frac{1}{t_i}\right) }\\
&\overset{(a)}{\le} |z_r|\cdot t_i+ \sqrt{2v_r\cdot\left(t_i^2\cdot \frac{(z-z_r)^2}{2v_r}+t_i^2\cdot\log  \frac{1}{t_i}\right) }\\
&\overset{(b)}{\le}|z_r| + \sqrt{(z-z_r)^2+0.4v_r }\\
&< C\cdot \left(1+\|(z,z_r)\|\right),
\end{split}
\]
where step (a) is from the definition of $E_{\min}$ and the fact that $z\in f^{-1}(f(z))$), and step (b) is due to $0< t_i\le 1$ and $t_i^2\log(1/t_i)< 0.2$.
\end{IEEEproof}

\subsection{Proof sketch for Lemma \ref{Eqn:LFD2}}\label{Sec:smooth_main}

Our proof for Lemma \ref{Eqn:LFD2} follows the approach proposed in \cite[Theorem 1]{Zheng17}. \textit{As many steps are similar to Lemma \ref{Eqn:LFD2}, we will not provide the full details of the proof, and only sketch the main idea.} The proof has two main steps:
\begin{itemize}
\item[(1)] The smoothed function $\eta_z^{\sigma,\xi,M}$ is Lipschitz continuous, so the asymptotic MSE of GLM-EP-app could be characterized by a state evolution (SE) recursion;
\item[(2)] Using the SE platform, we show that the asymptotic MSE of GLM-EP-app converges to the (expected) MSE of GLM-EP, as $\sigma\to0$, and $\xi\to0,M\to\infty$ sequentially. This implies that, with proper choice of $\sigma,\xi,M$, the asymptotic performance of GLM-EP-app is arbitrarily close to that of GLM-EP.
\end{itemize}

Step 1 is a consequence of \cite[Theorem 1]{fletcher2017inference}. Note that the model considered in this paper is a special case of that adopted in \cite[Theorem 1]{fletcher2017inference}. Also, here we assumed $f$ to be Lipschitz continuous, as required by \cite[Theorem 1]{fletcher2017inference}. The crucial assumption of \cite[Theorem 1]{fletcher2017inference} is the Lipschitz continuity of $\eta_z^{\sigma,\xi,M}$, which we prove in Lemma \ref{Lem:eta_smooth_pro} (see Section \ref{Sec:app_smooth_aux}).

A caveat is that \cite[Theorem 1]{fletcher2017inference} assumes $\eta_z^{\sigma,\xi,M}(z_r,y,v_r)$ to be uniform Lipschitz (see definition in \cite{fletcher2017inference}) w.r.t. to $(z_r,y)$ and $v_r$. However, since GLM-EP-app uses the deterministic sequences $\{V_r^t,V_l^t\}_{t\ge0}$ instead of their empirical counterparts $\{v_r^t,v_l^t\}$, this additional uniform continuity assumption is not required here.

Step 2 follows the same argument as in \cite[Theorem 1]{Zheng17}. First, the state evolution of GLM-EP-app is slightly more complicated than that of GLM-EP, and involve four sequences $\{\alpha_l^t,\tau_l^t,\alpha_r^t,\tau_r^t\}_{t\ge0}$. (The SE of GLM-EP can be viewed as a special case of this more general SE.) \textit{Note that these sequences all depend on the parameters $\sigma,\xi,M$}, but to keep notation light we do not make such dependency explicit. Intuitively speaking, $(\alpha_l^t,\tau_l^t)$ describes the correlation matrix of the components of $(\bm{z},\bm{z}_l^t)$ (where $\bm{z}\Mydef\bm{Ax}$):
\[
\mr{Cov}(Z,Z_l^t)\Mydef
\begin{bmatrix}
\mathbb{E}[Z^2] & \mathbb{E}[ZZ_l^t]\\
\mathbb{E}[ZZ_l^t] & \mathbb{E}[(Z_l^t)^2]
\end{bmatrix}
=
\begin{bmatrix}
1 & \alpha_l^t\\
\alpha_l^t & \tau_l^t
\end{bmatrix}.
\]
Similarly, $(\alpha_r^t,\tau_r^t)$ describes the correlation of the components of $(\bm{z},\bm{z}_r^t)$

The SE describing the recursive relationship of $\{\alpha_l^t,\tau_l^t,\alpha_r^t,\tau_r^t\}_{t\ge0}$ is given by
\[
\begin{split}
\alpha_l^t = \phi_1^{\sigma,\xi,M}(\alpha_r^t,\sigma_r^t),\quad &\text{and}\quad \tau_l^t=\phi_2^{\sigma,\xi,M}(\alpha_r^t,\sigma_r^t),\\
\alpha_r^t = \Phi_1(\alpha_r^t,\sigma_r^t),\quad &\text{and}\quad \tau_r^t=\Phi_2(\alpha_r^t,\sigma_r^t),\\
\end{split}
\]
where GLM-EP and GLM-EP-app start from the same initializations, i.e., $\alpha_r^{-1}=\tau_r^{-1}=V_r^{-1}$. A formal definition of these functions may be found in, e.g., \cite{fletcher2017inference}. 

Our goal is to show that the limit of the covariance $\mr{Cov}(Z,Z_l^t)$ for GLM-EP and GLM-EP-app for all $t\ge0$. Note that if $\mr{Cov}(Z,Z_l^t)$ for GLM-EP and GLM-EP-app are the same, then $\mr{Cov}(Z,Z_r^t)$ would also be the same, as the second steps of the two algorithms are identical (cf. \eqref{Eqn:EP_Gaussian} and \eqref{Eqn:GLM_EP_app_b}).

As in \cite[Theorem 1]{Zheng17}, the argument proceeds inductively on $t$. Because the steps are straightforward, we do not provide the full details and only consider the first iteration. Basically, we need to prove the following:
\BS\label{Eqn:app_smooth_goal2}
\begin{align}
\lim_{\xi\to0,M\to\infty}\lim_{\sigma\to0}\mathbb{E}[Z \eta_z^{\sigma,\xi,M}(Z_r,Y)] &= \mathbb{E}[Z \eta_z(Z_r,Y)]\label{Eqn:app_smooth_goal2_a} \\
\lim_{\xi\to0,M\to\infty}\lim_{\sigma\to0}\mathbb{E}[Z_r \eta_z^{\sigma,\xi,M}(Z_r,Y)] &= \mathbb{E}[Z_r \eta_z(Z_r,Y)] \label{Eqn:app_smooth_goal2_b}\\
\lim_{\xi\to0,M\to\infty}\lim_{\sigma\to0}\mathbb{E}[\left( \eta_z^{\sigma,\xi,M}(Z_r,Y)\right)^2] &= \mathbb{E}[ \left(\eta_z(Z_r,Y)\right)^2] \label{Eqn:app_smooth_goal2_c}
\end{align}
\ES
where $Z\sim\mathcal{N}(0,1)$, $Y=f(Z)$. Note that these results hold for any $Z_r$ as long as it is joint Gaussian with $Z$ (non-degenerate). 

We next prove \eqref{Eqn:app_smooth_goal2_a}. Other results can be proved in the same way. We first calculate its limit of $\mathbb{E}[Z\cdot \eta_z^{\sigma,\xi,M}(Z_r,Y)]$ as $\sigma\to0$. From Lemma \ref{Lem:eta_smooth_pro}, the function $\eta_z^{\sigma,\xi,M}$ is bounded, and using dominated convergence theorem we get
\[
\begin{split}
\lim_{\sigma\to0}\mathbb{E}[Z\cdot \eta_z^{\sigma,\xi,M}(Z_r,Y)]&=\mathbb{E}\left[ Z\cdot \lim_{\sigma\to0}\eta_z^{\sigma,\xi,M}(Z_r,Y) \right]\\
&=\mathbb{E}\left[ Z\cdot \eta_z^{\xi,M}(Z_r,Y) \right],
\end{split}
\]
where the last step follows from
\BE\label{Eqn:app_smooth_as_con}
\lim_{\sigma\to0}\eta_z^{\sigma,\xi,M}(Z_r,Y)  = \eta_z^{\xi,M}(Z_r,Y)\quad \text{a.s.}
\EE
To see \eqref{Eqn:app_smooth_as_con}, note that Lemma \ref{Lem:eta_smooth_pro} shows $\lim_{\sigma\to0}\eta_z^{\sigma,\xi,M}(z_r,y)  = \eta_z^{\xi,M}(z_r,y)$ whenever $\eta_z^{\xi,M}$ is continuous at $(z_r,y)$. In particular, our construction of $\eta_z^{\xi,M}$ (see \eqref{Eqn:eta_smooth_1}) guarantees that $\eta_z^{\xi,M}$ is continuous at $(z_r,y)\in\mathbb{R}\times\{y_1,\ldots,y_Q\}$. Similar to the proof of Lemma \ref{Lem:eta_smooth_pro}-(P.1), it can be shown that the set of points at which $\eta_z^{\xi,M}$ is discontinuous has zero probability (with respect to the distribution of $(Z_r,Y)$). 

It remains to prove
\[
\begin{split}
\lim_{\xi\to0,M\to\infty}\mathbb{E}\left[ Z\cdot (\eta_z^{\xi,M}(Z_r,Y)-\eta_z(Z_r,Y))\right]=0.
\end{split}
\]
Similar to \ref{Lem:eta_smooth_pro}-(P.1), it can be shown that $\eta_z^{\xi,M}(Z_r,Y)$ is almost surely bounded w.r.t, the distribution of $Z_r,Y$. Also, by Lemma \ref{Lem:poly}, $\eta_z(Z_r,Y)=\eta_z(Z_r,f(Z))\le C\cdot (1+\|(Z_r,Z)\|)$. Hence, we could apply DCT to show
\[
\begin{split}
&\lim_{\xi\to0,M\to\infty}\mathbb{E}\left[ Z\cdot (\eta_z^{\xi,M}(Z_r,Y)-\eta_z(Z_r,Y))\right]\\
&=\mathbb{E}\left[ \lim_{\xi\to0,M\to\infty}Z\cdot (\eta_z^{\xi,M}(Z_r,Y)-\eta_z(Z_r,Y))\right]\\
&=0.
\end{split}
\]

\section{Proofs of Lemma \ref{Lem:MSE_SE} and Lemma \ref{Lem:Gaussian_condition}}\label{App:SE_conditions}

\subsection{Proof of Lemma \ref{Lem:MSE_SE}}\label{App:MSE_SE}
From \eqref{Eqn:SE_Gaussian}, the state evolution recursion for $V_r$ reads
\[
V_r^{t+1}=\Phi\left(\phi(V_r^t)\right),
\]
where $V_r^0=1$.
Lemma \ref{Lem:phi} in Appendix \ref{App:SE_maps} implies that the composite function $\Phi(\phi(V_r)))$ is continuously increasing in $[0,1]$. Further, $\Phi(\phi(V_r))\ge0$ for any $V_r\in[0,1]$. An induction argument shows that $\{V_r^t\}$ monotonically converges if and only if
\BE\label{Eqn:Gaussian_condition_1}
\Phi(\phi(V_r^1)) \le V_r^1,
\EE
which holds since $V_r^1=1$, $\phi(1)\ge0$ (see Lemma \ref{Lem:phi}) and $\Phi(v)\le1$ for all $v\ge0$. Further, if $\phi(1)\neq\infty$, then the sequence $\{V_r^t\}$ converges to $V_r^\star$, where $V_r^\star$ is the smallest $v$ so that the following holds for all $V_r\in[v,1]$, i.e.,
\[
V_r^\star = \inf\Big\{v\in[0,1]\,:\,\Phi\left(\phi(v_r)\right)< v_r, \forall v_r\in[v,1]\Big\}.
\]
Substituting in the definitions of $\phi$ and $\Phi$ in \eqref{Eqn:SE_Gaussian}, it is straightforward to show that the above definition of $V_r^\star$ is equivalent to that in \eqref{Eqn:gamma_leftFP}.

For the degenerate case where $\phi(1)=\infty$ (which corresponds to $\mmse_z(1)=1$ and happens when $f$ is an even function), $P(1)=0$ and so $V_r^\star$ in \eqref{Eqn:gamma_leftFP}) is not defined. Lemma \ref{Lem:MSE_SE} holds by defining $V_r^\star=1$ for this degenerate case.

\subsection{Proof of Lemma \ref{Lem:Gaussian_condition}}\label{App:Proof_Condition}

Throughout this paper, we denote $\phi(0)\Mydef\lim_{v_r\to0}\phi(v_r)$. In our discussions below, we shall exclude two cases for which the Lemma holds trivially: (1) $f$ is invertible, it is easy to show $P(v_r)=\delta-1>0$ and $\mathsf{MSE}^\star_\Lambda=0$; (2) $f(Z)$ is independent of $Z$ (e.g., $f(Z)$ is a constant). Clealry, $\mathsf{MSE}^\star_\Lambda=1$. At the same time, $\mmse_z(v_r)=v_r$, $\phi(v_r)=+\infty$, and $P(v_r)=0$ for all $v_r$.

\subsubsection{Proof of (i)}

Consider two following cases.
\begin{itemize}
\item \textit{Case 1:} $d(Y)=0$;
\item \textit{Case 2:} $d(Y)> 0$.
\end{itemize}

For Case 1, we next show that the condition \eqref{Eqn:Gaussian_perfect2} does not hold, and further perfect recovery is impossible, i.e., $\mathsf{MSE}^\star_\Lambda\neq0$. To see this, we consider $P(0)$:
\begin{equation}\label{Eqn:P_0_first}
\begin{split}
&P(0) = \mathbb{E}\left[\frac{\phi(0)}{\phi(0)+\Lambda}\right]-1+\delta\left[1-\lim_{v_r\to0}\frac{\mmse_z(v_r)}{v_r}\right]\\
&\overset{(a)}{=}\mathbb{E}\left[\frac{\phi(0)}{\phi(0)+\Lambda}\right]-1+\delta\left[1-\lim_{v_r\to0}\frac{\mmse(Z,v_r^{-1}-1|Y)}{v_r}\right]\\
&=\mathbb{E}\left[\frac{\phi(0)}{\phi(0)+\Lambda}\right]-1\\
&\quad +\delta\left[1-\lim_{\snr_{\text{eff}}\to\infty}(\snr_{\text{eff}}+1)\mmse(Z,\snr_{\text{eff}}|Y)\right]\\
&\qquad (\snr_{\text{eff}}:=\frac{1}{v}-1)\\
&\overset{(b)}{=}\mathbb{E}\left[\frac{\phi(0)}{\phi(0)+\Lambda}\right]-1+\delta\left[1-\mathscr{D}(Z|Y)\right]\\
&\overset{(c)}{=}\mathbb{E}\left[\frac{\phi(0)}{\phi(0)+\Lambda}\right]-1+\delta d(Y)\\
&\overset{(d)}{=}\mathbb{E}\left[\frac{\phi(0)}{\phi(0)+\Lambda}\right]-1\\
&<0,
\end{split}
\end{equation}
where step (a) is from the definition of $\mmse_z$ below \eqref{Eqn:SE_Gaussian}, and step (b) is from definition of $\mathscr{D}(Z|Y)$ (see Definition \eqref{Def:MMSE_dim}) and the fact that $\mmse_z(0)=0$, step (c) is from  Lemma \ref{Lem:MMSE_D_ID}, and the last step is from $0\le\phi(0)<+\infty$ (see Lemma \ref{Lem:phi}).

By continuity of $P$, there exists $\hat{v}_r\in(0,1]$ such that $P(\hat{v}_r)\le0$, implying that \eqref{Eqn:Gaussian_perfect2} does not hold. Further, from Lemma \ref{Lem:MSE_SE}, $\lim_{t\to\infty} V_r^{t}=V_r^{\star}\ge \hat{v}_r>0$. By the strict monotonicity of $\phi$, $\phi(V_r^{\star})>\phi(\hat{v}_r)>0$. Then, from \eqref{Eqn:MSE_lambda_def}, $\mathsf{MSE}^\star_\Lambda\Mydef \mathsf{MSE}_\Lambda(\phi(V_r^\star))>0$.

\vspace{5pt}

The proof for Case 2 is also straightforward. From Lemma \ref{Lem:MSE_SE}, we have the following equivalence:
\[
V_r^{\star}\Mydef\lim_{t\to\infty} V_r^{t}=0 \quad \Longleftrightarrow \quad \eqref{Eqn:Gaussian_perfect2} \text{ holds}.
\]
Furthermore, for Case 2, \eqref{Eqn:Gaussian_perfect2} guarantees $\phi(0)=0$. On the other hand,
\[
V_r^{\star}=0\text{ and }\phi(0)=0\  \Longleftrightarrow \  \mathsf{MSE}^\star_\Lambda\Mydef \mathsf{MSE}_\Lambda(\phi(V_r^\star))=0.
\]
This proves the equivalence between \eqref{Eqn:Gaussian_perfect2} and $\mathsf{MSE}^\star_\Lambda=0$.

\subsubsection{Proof of (ii)} 

From \eqref{Eqn:P_0_first},
\[
\begin{split}
P(0) &=\mathbb{E}\left[\frac{\phi(0)}{\phi(0)+\Lambda}\right]-1+\delta\cdot d(Y)\\
&=
\begin{cases}
\mathbb{E}\left[\frac{\phi(0)}{\phi(0)+\Lambda}\right] -1<0 & \text{if }d(Y)=0\\
-1+\delta\cdot d(Y) & \text{if }d(Y)\neq 0
\end{cases}
\end{split}
\]
where we used the fact that $\phi(0)=0$ when $d(Y)\neq 0$. Overall, if $\delta < 1/d(Y)$, we have $P(0)<0$. By continuity of $P$, \eqref{Eqn:Gaussian_perfect2} does not hold, which together with part (i) shows that $\mathsf{MSE}^\star_\Lambda\neq0$. Hence, $\delta\ge1/d(Y)$ is necessary for achieving $\mathsf{MSE}^\star_\Lambda=0$.

Now suppose $\delta > 1/d(Y)$. We next prove there exists a spectrum $P_\Lambda$ such that $\mathsf{MSE}^\star_\Lambda=0$. From part (i), this is equivalent to checking there exists $P_\Lambda$ such that $P(v_r)>0$ for all $v_r\in(0,1]$, which can be rewritten as (from \eqref{Eqn:Pv_def}):
\begin{align}\label{Eqn:exist_obj}
\mathbb{E}\left[\frac{\phi(v_r)}{\phi(v_r) +\Lambda}\right]  &>g(v_r), \quad\forall v_r\in(0,1].
\end{align}
We first prove
\BE\label{Eqn:g_strict_less_1}
\sup_{v_r\in(0,1]}g(v_r)<1.
\EE
where $g(v_r)$ is defined by (see \eqref{Eqn:Pv_def})
\[
g(v_r)\Mydef 1-\delta\left(1-\frac{\mmse_z(v_r)}{v_r}\right).
\]
As shown in \eqref{Eqn:MMSE_le_v} (Appendix \ref{App:SE_maps}), $\mmse_z(v_r)\le v_r$ for all $v_r\in(0,1)$. Further, the inequality is strict when $Y:=f(Z)$ and $Z$ are not independent, which was assumed to hold (see discussions at the start of this appendix). Therefore,
\[
g(v_r)< 1,\quad \forall v_r\in(0,1).
\]
Further, when $v_r=1$, $\mmse_z(1)<1$ (by assumption) and so $g(1)<1$. When $v_r\to0$,
\BE\label{Eqn:g_0}
\begin{split}
g(0)&\Mydef\lim_{v_r\to0_+}g(v_r)\\
&=\lim_{v_r\to0}1-\delta\left(1-\frac{\mmse_z(v_r)}{v_r}\right)\\
&=1-\delta(1-\mathscr{D}(Z|f(Z)))\\
&=1-\delta \cdot d(Y)\\
&<0,
\end{split}
\EE
where the last equality is due to Lemma \ref{Lem:MMSE_D_ID} and the last inequality is from the assumption $\delta > 1/d(Y)$.  Combining the above facts (together with the continuity of $g$) proves \eqref{Eqn:g_strict_less_1}.

We are now in the position to prove \eqref{Eqn:exist_obj}. Consider the following two-point distribution parameterized by $P\in(0,1)$ and $a\in(0,\delta)$:
\BE\label{Eqn:two_point}
P_\Lambda=
\begin{cases}
a & \text{with prob. } P\\
\frac{\delta - a P}{1-P} & \text{with prob. } 1-P.
\end{cases}
\EE
This distribution satisfies the normalization assumption $\mathbb{E}[\Lambda]=\delta$. Under this distribution, the left-hand side of \eqref{Eqn:exist_obj} becomes
\BE
\begin{split}
\mathbb{E}\left[\frac{\phi(v_r)}{\phi(v_r) +\Lambda}\right] &= P\cdot \frac{\phi(v_r)}{\phi(v_r)+a}+(1-P)\cdot \frac{\phi(v_r)}{\phi(v_r)+b}\\
&> P\cdot \frac{\phi(v_r)}{\phi(v_r)+a}\quad(b>0),
\end{split}
\EE
where
\[
 b \Mydef \frac{\delta - a P}{1-P}.
\]
We next show that there exists $a\in(0,\delta)$ and $P\in(0,1)$ for which the following holds,
\[
P\cdot \frac{\phi(v_r)}{\phi(v_r)+a} > g(v_r).\quad \forall v_r\in(0,1].
\]
Since $\phi(v_r)$ is non-negative (see Lemma \ref{Lem:phi}), It suffices to prove
\BE\label{Eqn:D_def}
a < \phi(v_r)\cdot\left(\frac{P}{g(v_r)}-1\right),
\EE
for all $v_r\in\mathbb{D}\Mydef \left\{ v_r\in(0,1]\,:\,  g(v_r)\ge0\right\}$. Consider an arbitrary $P\in\left(\sup_{v\in\mathbb{D}}\ g(v) ,1\right)$. Due to \eqref{Eqn:g_strict_less_1}, this choice of $P$ is valid.

Let
\[
a_{\min}(P)\Mydef\inf_{v_r\in\mathbb{D}}\ \phi(v_r)\cdot\left(\frac{P}{g(v_r)}-1\right).
\]
We conclude our proof by showing $a_{\min}(P)>0$ for $P\in\left(\sup_{v\in\mathbb{D}}\ g(v) ,1\right)$, and setting $a\in\left(0,\min(a_{\min}(P),\delta)\right)$. To this end, we note
\BS
\BE\label{Eqn:impact_spec_final1}
\inf_{v_r\in\mathbb{D}}\ \phi(v_r)>0,
\EE
and
\BE\label{Eqn:impact_spec_final2}
\inf_{v_r\in\mathbb{D}}\ \left(\frac{P}{g(v_r)}-1\right)>0.
\EE
\ES
Eq.~\eqref{Eqn:impact_spec_final1} is due to the following facts: (i) $\phi(v_r)>0$ for all $v_r\neq0$ when $f$ is not invertible (see Lemma \ref{Lem:phi}); and (ii) $\mathbb{D}\Mydef \left\{ v_r\in(0,1]\,:\,  g(v_r)\ge0\right\}\subset (\hat{v},1]$ for some $\hat{v}>0$. (Since $g(0)<0$ and $g$ is continuous.) Eq.~\eqref{Eqn:impact_spec_final2} is due to the definition $P\in\left(\sup_{v\in\mathbb{D}}\ g(v) ,1\right)$.

This completes the proof.

\section{Proof of Lemma \ref{Lem:impact}}\label{App:Proof_Impact_MSE}

Lemma \ref{Lem:MSE_SE} shows that the MSE of the {\Alg} algorithm is given by
\BE\label{Eqn:MSE_lambda}
\mathsf{MSE}^\star_\Lambda\Mydef\mathbb{E}\left[ \frac{\phi(V_\Lambda^\star)}{\phi(V_\Lambda^\star) + \Lambda} \right],
\EE
where
\BE\label{Eqn:funs_to_plot}
V_\Lambda^{\star}=\inf\left\{v\in[0,1]\,:\,\mathbb{E}\left[ \frac{\phi(v_r)}{\phi(v_r) + \Lambda} \right] > g(v_r), \forall v_r\in[v,1]\right\}.
\EE
Here, the subscript $(\cdot)_{\Lambda}$ is added to emphasize the dependency on the spectrum $P_\Lambda$. Since $\mathbb{E}\left[ \frac{\phi(v_r)}{\phi(v_r) + \Lambda} \right]\ge0$, $V_\Lambda^{\star}$ can be equivalently defined as
\BE\label{Eqn:v_hat_def}
{V}_\Lambda^{\star}=\inf\left\{v\in[0,1]\,:\,\mathbb{E}\left[ \frac{\phi(v_r)}{\phi(v_r) + \Lambda} \right] > G(v_r;\delta), \forall v_r\in[v,1]\right\}.
\EE
where $G(v_r;\delta)$ is defined in \eqref{Eqn:g_def}. We next prove that $v_\Lambda^{\star}$ must satisfy
\BE\label{Eqn:g_conditions}
\mathsf{MSE}^\star_\Lambda\Mydef\mathbb{E}\left[ \frac{\phi(V_\Lambda^{\star})}{\phi(V_\Lambda^{\star}) + \Lambda} \right]=G(V_\Lambda^\star;\delta).
\EE
Eq.~\eqref{Eqn:Gaussian_condition_1} implies $\mathbb{E}\left[\frac{\phi(1)}{\phi(1)+\Lambda}\right]\ge g(1)$. Further, $\phi(1)\ge0$, and thus $\mathbb{E}\left[\frac{\phi(1)}{\phi(1)+\Lambda}\right]\ge 0$. Together, we have $\mathbb{E}\left[\frac{\phi(1)}{\phi(1)+\Lambda}\right]\ge G(1;\delta)$. The only possibility \eqref{Eqn:g_conditions} does not hold is when
\BE\label{Eqn:g_conditions2}
\mathbb{E}\left[ \frac{\phi(v_r)}{\phi(v_r) + \Lambda} \right] > G(v_r;\delta)\quad \forall v_r\in[0,1].
\EE
We next show that \eqref{Eqn:g_conditions2} cannot hold. We only need to prove \eqref{Eqn:g_conditions2} cannot hold for $v_r=0$. We consider two case $d(Y)>0$ and $d(Y)=0$ separately. When $d(Y)>0$, Lemma \ref{Lem:phi} guarantees $\phi(0)=0$, and thus $
\mathbb{E}\left[ \frac{\phi(0)}{\phi(0) + \Lambda} \right]=0\le G(0;\delta),$
where the inequality is from the definition of $G(\cdot)$. When $d(Y)=0$, as shown in \eqref{Eqn:P_0_first}, we have
\[
\begin{split}
\lim_{v_r\to0}1-\delta\cdot\left[1-\frac{\mmse_z(v_r)}{v_r}\right]&=1-\delta\left[1-\mathscr{D}(Z|Y)\right]\\
&=1-\delta\cdot d(Y)\\
&=1.
\end{split}
\]
Hence, from \eqref{Eqn:g_def}, we have $G(0;\delta)=1$. On the other hand, $\mathbb{E}\left[ \frac{\phi(0)}{\phi(0) + \Lambda} \right]\le1$ since $\phi(0)\ge0$. Hence, \eqref{Eqn:g_conditions2} cannot hold at $v_r=0$. Combining the previous arguments proves \eqref{Eqn:g_conditions}.

\vspace{5pt}

At this point, we can compare $\mathsf{MSE}_{\Lambda_1}^\star$ and $\mathsf{MSE}_{\Lambda_2}^\star$. Note that $C/( C +  \Lambda )^{-1}$ is a convex function of $\Lambda$ for every $C>0$. Hence, Lemma \ref{Lem:convex_Lorenz} implies that the following holds for all $\gamma_l>0$:
\[
\Lambda_1 \succeq_L \Lambda_2\quad \Longrightarrow\quad \mathbb{E}\left[ \frac{\phi(v_r)}{ \phi(v_r) + \Lambda_1 } \right] \ge \mathbb{E}\left[ \frac{\phi(v_r)}{ \phi(v_r) + \Lambda_2 }  \right],
\]
where $\succeq_L$ means spikier in the Lorenz sense (see Definition \ref{Def:Lorenz}). From the definition of $V_\Lambda^\star$, we have
\BE\label{Eqn:impact_v}
\Lambda_1 \succeq_L \Lambda_2\quad \Longrightarrow\quad V_{\Lambda_1}^{\star}\le V_{\Lambda_2}^{\star}.
\EE
To compare $\mathsf{MSE}_{\Lambda_1}^\star$ and $\mathsf{MSE}_{\Lambda_2}^\star$, it is not very convenient to directly use \eqref{Eqn:MSE_lambda} since the expectation in \eqref{Eqn:MSE_lambda} itself depends on the distribution of $\Lambda$. Instead, due to \eqref{Eqn:g_conditions}, we only need to compare $G(V_{\Lambda_1}^\star;\delta)$ and $G(V_{\Lambda_2}^\star;\delta)$. Since $V_{\Lambda_1}^{\star}\le V_{\Lambda_2}^{\star}$, the claims in the lemma follow directly.

\section{Proof of Theorem \ref{The:1}}\label{App_Proof_Achie}


From Lemma \ref{Lem:Gaussian_condition}, the {\Alg} algorithm cannot achieve perfect recovery at finite $\delta$ if $d(Y)=0$. Therefore, we will only consider the case $d(Y)>0$. In this case, we have $\phi(0)=0$ (see Lemma \ref{Lem:phi}). \vspace{5pt}

When $\phi(0)=0$, we have $\mathsf{MSE}_\Lambda^\star=0$ if and only if $V_\Lambda^\star=0$, where $V_\Lambda^\star$ is defined in \eqref{Eqn:v_hat_def}. Therefore, we can equivalently define $\delta_\Lambda^{\mathsf{alg}}$ as
\BE\label{Eqn:delta_lambda2}
\delta_\Lambda^{\mathsf{alg}} = \inf\left\{\delta\,:\, V_\Lambda^\star=0\right\}.
\EE
We have proved in \eqref{Eqn:impact_v} that if $\Lambda_1 \succeq_L \Lambda_2$, then $ V_{\Lambda_1}^{\star}\le V_{\Lambda_2}^{\star}$ and hence $\delta_{\Lambda_1}^{\mathsf{alg}}\le\delta_{\Lambda_2}^{\mathsf{alg}}$ (from \eqref{Eqn:delta_lambda2}).

\section{Proof of Lemma \ref{Lem:noise_sensitivity}}
\label{App:proof_noise}
Throughout this appendix, we assume $\delta>\delta_\Lambda^\mathsf{alg}\ge1/d(Y)$, where the second inequality is a consequence of the necessary condition for perfect reconstruction given in Lemma \ref{Lem:Gaussian_condition}. 

We first collect some auxiliary lemmas in Section \ref{Sec:Aux} before we present our main proof in Section \ref{Sec:Noise_main}.

\subsection{Auxiliary Results}\label{Sec:Aux}

We denote
\BE\label{Eqn:App_RVs}
\begin{split}
Y &= f(Z),\\
Y_\sigma&=f(Z+\sigma_w W),\\
U_\sigma&= Z+\sigma_w W,\\
Z_r&=(1-v_r)Z+\sqrt{v_r(1-v_r)} N,\\
R&=\sqrt{v}_r Z - \sqrt{1-v_r} N,
\end{split}
\EE
where $Z,N,W,R$ are standard Gaussian RVs, and $(Z,N,W)$ are mutually independent and $R\Perp (Z_r,W)$. (Here, $A \Perp B$ denotes $A,B$ are independent RVs.) Notice that
\[
Z=Z_r + \sqrt{v_r}R.
\]

\begin{lemma}\label{Lem:aux_noisy2}
Let $\mmse_z(v_r,\sigma_w^2)$ be the noisy MMSE defined in \eqref{Eqn:MMSE_noisy}. Define
\BE\label{Eqn:MMSE_app_def}
\begin{split}
\mmse_{\mr{app}}(v_r,\sigma^2_w)  \Mydef  & v_r \mathbb{E}\left( \left( \frac{\sigma^2_wR}{v_r+\sigma^2_w} -\frac{\sqrt{v}\sigma_w W}{v_r+\sigma^2_w} \right)^2 \mathbb{I}(\mathcal{E}_1) \right)\\
&+v_r\mathbb{E}\left( R^2 \mathbb{I}(\mathcal{E}_1^c)\right).
\end{split}
\EE
where
\BE\label{Eqn:E1_def}
\mathcal{E}_1\Mydef\{U_\sigma\in \mathbb{R}\backslash\mathcal{Q}_f \},
\EE
$U_\sigma=Z+\sigma_wN$ and $\mathcal{E}_1^c$ is the complement of $\mathcal{E}_1$. Then, the following holds
\BE
\lim_{v_r+\sigma^2_w\to0} \frac{1}{v_r}\cdot\Big( \mmse_z(v_r,\sigma^2_w)-\mmse_{\mr{app}}(v_r,\sigma^2_w)  \Big)=0.
\EE
\end{lemma}

\begin{IEEEproof}
By the definitions of $\mmse_z$ and $\mmse_{\mr{app}}$, we have
\[
\begin{split}
&\frac{1}{v_r}\left(\mmse_z(v_r,\sigma_w^2)-\mmse_{\mr{app}}(v_r,\sigma_w^2)\right)\\
&= \mathbb{E}\bigg[ \frac{1}{v_r}\left( Z-\mathbb{E}[Z|Y_\sigma,Z_r] \right)^2-\\
&\ \left( \frac{\sigma^2_wR}{v_r+\sigma^2_w} -\frac{\sqrt{v_r}\sigma_w W}{v_r+\sigma^2_w}  \right)^2 \mathbb{I}(\mathcal{E}_1) -  R^2 \mathbb{I}(\mathcal{E}_1^c)\bigg].
\end{split}
\]
We bound the term inside the expectation by
\BE\label{Eqn:noisy_DCT}
\begin{split}
&\Bigg|\frac{1}{v_r}\left( Z-\mathbb{E}[Z|Y_\sigma,Z_r] \right)^2\\
&\ -\left( \frac{\sigma^2_wR}{v_r+\sigma^2_w} -\frac{\sqrt{v_r}\sigma_w W}{v_r+\sigma^2_w}  \right)^2 \mathbb{I}(\mathcal{E}_1) -  R^2 \mathbb{I}(\mathcal{E}_1^c)\Bigg|\\
&\overset{(a)}{=}\Bigg|(R-\mathbb{E}[R|Y_\sigma,Z_r])^2\\
&\  - \left( \frac{\sigma^2_wR}{v_r+\sigma^2_w} -\frac{\sqrt{v_r}\sigma_w W}{v_r+\sigma^2_w} \right)^2 \mathbb{I}(\mathcal{E}_1) -  R^2 \mathbb{I}(\mathcal{E}_1^c) \Bigg|\\
&\le 2(R^2 + \mathbb{E}^2[R|Y_\sigma,Z_r])\\
&\ +\left( \frac{\sigma^2_wR}{v_r+\sigma^2_w} -\frac{\sqrt{v_r}\sigma_w W}{v_r+\sigma^2_w} \right)^2 \mathbb{I}(\mathcal{E}_1) + R^2 \mathbb{I}(\mathcal{E}_1^c)\\
&\le 2(R^2 + \mathbb{E}^2[R|Y_\sigma,Z_r])+2\left(R^2+\frac{1}{4}W^2\right)  + R^2\\
\end{split}
\EE
where step (a) follows from the definition $Z = Z_r + \sqrt{v_r}R$. Since 
\[\mathbb{E}\left[2(R^2 + \mathbb{E}^2[R|Y_\sigma,Z_r])+\left(R^2+\frac{1}{4}W^2\right)  + R^2\right]<\infty,
\]
by dominated convergence theorem we have
\[
\begin{split}
&\lim_{v_r+\sigma^2_w\to0} \frac{1}{v_r}\mmse_z(v_r,\sigma^2_w)-\mmse_{\mr{app}}(v_r,\sigma^2_w)\\
&=\lim_{v_r+\sigma^2_w\to0} \frac{1}{v_r}\mathbb{E}\Big(Z-\mathbb{E}[Z|Y_\sigma,Z_r]\Big)^2-\mmse_{\mr{app}}(v_r,\sigma^2_w) \\
&= \mathbb{E}\Bigg[ \lim_{v_r+\sigma^2_w\to0} \frac{1}{v_r}\left( Z-\mathbb{E}[Z|Y_\sigma,Z_r] \right)^2\\
&\ -\left( \frac{\sigma^2_wR}{v_r+\sigma^2_w} -\frac{\sqrt{v}\sigma_w W}{v_r+\sigma^2_w} \right)^2 \mathbb{I}(\mathcal{E}_1) -  R^2 \mathbb{I}(\mathcal{E}_1^c)\Bigg]\\
&=\mathbb{E}[T_1] + \mathbb{E}[T_2],
\end{split}
\]
where
\BE\label{Eqn:T12_def}
\begin{split}
T_1 \Mydef &\lim_{v_r+\sigma^2_w\to0} \frac{1}{v_r}\left( Z-\mathbb{E}[Z|Y_\sigma,Z_r] \right)^2\mathbb{I}(\mathcal{E}_1)\\
&\ -\left( \frac{\sigma^2_wR}{v_r+\sigma^2_w} -\frac{\sqrt{v}\sigma_w W}{v_r+\sigma^2_w}  \right)^2 \mathbb{I}(\mathcal{E}_1)\\
T_2\Mydef & \lim_{v_r+\sigma^2_w\to0} \frac{1}{v_r}\left( Z-\mathbb{E}[Z|Y_\sigma,Z_r] \right)^2\mathbb{I}(\mathcal{E}_1^c)-R^2 \mathbb{I}(\mathcal{E}_1^c).
\end{split}
\EE
We next prove $\mathbb{E}[T_1]=0$ and $\mathbb{E}[T_2]=0$ separately. \vspace{5pt}

\textit{Analysis of $T_1$:} Direct calculations yield
\BS
\begin{align}
&\mathbb{E}[Z|Y_\sigma=y,Z_r=z_r] \\
&=\frac{\int_{f^{-1}(y)} \mathcal{N}(u;z_r,v_r+\sigma^2_w)\frac{v_r u +\sigma^2_w z_r}{v_r+\sigma^2_w} \mr{d} u}{\int_{f^{-1}(y)} \mathcal{N}(u;z_r,v_r+\sigma^2_w) \mr{d} u}\\
&=z_r + \frac{v_r}{v_r+\sigma^2_w}\frac{\int_{\mathcal{I}} u \mathcal{N}(u;0,v_r+\sigma^2_w)\mr{d} u}{\int_{\mathcal{I}} \mathcal{N}(u;0,v_r+\sigma^2_w)\mr{d} u},
\end{align}
\ES
where $\mathcal{N}(x;m,v)\Mydef \frac{1}{\sqrt{2\pi v}}\exp\left(-\frac{(x-m)^2}{2v}\right)$, $\mathcal{I}\Mydef f^{-1}(y)-z_r$, and the second step is due to a change of variable. We emphasize that $\mathcal{I}$ is indexed by $y$ and $z_r$, but to make notation light we did not make such dependency explicit. When $f^{-1}(y)$ is a discrete set, the integration is simply replaced by a summation. 

With slight abuse of notations, let $(z,w,n,y,z_r)$ be an instance of $(Z, W,N,Y_\sigma,Z_r)$. From \eqref{Eqn:App_RVs}, we have $z_r =(1-v_r)z+\sqrt{v_r(1-v_r)}n $ and $y=f(z+\sigma_w w)$. Then, 
\BE\label{Eqn:App_CondMean}
\begin{split}
&\frac{1}{v_r}\left( z-\mathbb{E}[Z|Y_\sigma=y,Z_r=z_r] \right)^2\\
&=\frac{1}{v_r}\left( z-z_r - \frac{v_r}{v_r+\sigma^2_w}\frac{\int_{\mathcal{I}} u \mathcal{N}(u;0,v_r+\sigma^2_w)\mr{d} u}{\int_{\mathcal{I}} \mathcal{N}(u;0,v_r+\sigma^2_w)\mr{d} u}\right)^2\\
&=\frac{1}{v_r}\left( v_r z - \sqrt{v_r(1-v_r)}n- \frac{v_r}{v_r+\sigma^2_w}\frac{\int_{\mathcal{I}} u \mathcal{N}(u;0,v_r+\sigma^2_w)\mr{d} u}{\int_{\mathcal{I}} \mathcal{N}(u;0,v_r+\sigma^2_w)\mr{d} u}\right)^2\\
&=\left( \sqrt{v_r}z - \sqrt{1-v_r}n- \frac{\sqrt{v_r}}{{v_r+\sigma^2_w}}\frac{\int_{\mathcal{I}} u \mathcal{N}(u;0,v_r+\sigma^2_w)\mr{d} u}{\int_{\mathcal{I}} \mathcal{N}(u;0,v_r+\sigma^2_w)\mr{d} u}\right)^2\\
&=\left( r- \frac{\sqrt{v_r}}{{v_r+\sigma^2_w}}\frac{\int_{\mathcal{I}} u \mathcal{N}(u;0,v_r+\sigma^2_w)\mr{d} u}{\int_{\mathcal{I}} \mathcal{N}(u;0,v_r+\sigma^2_w)\mr{d} u}\right)^2,
\end{split}
\EE
where the last step is due to the definition of the r.v. $R$ in \eqref{Eqn:App_RVs}.
Recall that $\mathcal{E}_1=\{Z+\sigma_\sigma W\in\mathbb{R}\backslash \mathcal{Q}_f\}$, where $\mathcal{Q}_f=\{z:f^{-1}(f(z)) \text{ contains an interval}\}$. Conditioned on $\mathcal{E}_1$, $f^{-1}(y)$ is a discrete set, and so is $\mathcal{I}\Mydef f^{-1}(y)-z_r$. Hence, conditioned on $\mathcal{E}_1$, the integration in the above formula is replaced by summation over the elements in $\mathcal{I}$. Since $y=f(z+\sigma_w w)$, we have $z+\sigma_w w\in f^{-1}(y)$. Further, $z=z_r+\sqrt{v_r}r$, and thus
\[
z+\sigma_w w - z_r = \sqrt{v_r} r+\sigma_w w \in f^{-1}(y)-z_r=\mathcal{I}.
\]
Let $\mathcal{E}_2$ be the event that there does not exist $x\in f^{-1}(y)$ and $x\neq z+\sigma_w w$ such that $|z+\sigma_w w|=|x-z_r|$. Then, on the event $\mathcal{E}_1\cap \mathcal{E}_2$,  
\[
\lim_{v_r+\sigma^2_w\to0} \frac{\int_{\mathcal{I}} u \mathcal{N}(u;0,v_r+\sigma^2_w)\mr{d} u}{\int_{\mathcal{I}} \mathcal{N}(u;0,v_r+\sigma^2_w)\mr{d} u}-(\sqrt{v_r}r+\sigma_ww)=0.
\]
{This is due to the fact that ${\mathcal{I}}$ is a discrete set and the term with minimum exponent dominates.} Hence,
\[
\begin{split}
&\lim_{v_r+\sigma^2_w\to0}\frac{\sqrt{v_r}}{{v_r+\sigma^2_w}}\frac{\int_{\mathcal{I}} u \mathcal{N}(u;0,v_r+\sigma^2_w)\mr{d} u}{\int_{\mathcal{I}} \mathcal{N}(u;0,v_r+\sigma^2_w)\mr{d} u}\\
&\ -\frac{\sqrt{v_r}(\sqrt{v_r}r+\sigma_w w) }{v_r+\sigma^2_w}\\
&=0,
\end{split}
\]
Hence, conditioned $\mathcal{E}_1\cap\mathcal{E}_2$, we have (see \eqref{Eqn:App_CondMean})
\[
\begin{split}
&\frac{1}{v_r}\left( z-\mathbb{E}[Z|Y_\sigma=y,Z_r=z_r] \right)^2\\
&=\left( r- \frac{\sqrt{v_r}}{{v_r+\sigma^2_w}}\frac{\int_{\mathcal{I}} u \mathcal{N}(u;0,v_r+\sigma^2_w)\mr{d} u}{\int_{\mathcal{I}} \mathcal{N}(u;0,v_r+\sigma^2_w)\mr{d} u}\right)^2+o(v_r+\sigma_w^2)\\
&=\left(r-\frac{\sqrt{v_r}(\sqrt{v_r}r+\sigma_w w) }{v_r+\sigma^2_w}\right)^2+o(v_r+\sigma_w^2)\\
&=\left(\frac{\sigma_w^2 r-\sqrt{v_r}\sigma_w w}{v_r+\sigma^2_w}\right)^2+o(v_r+\sigma_w^2)
\end{split}
\]

Since $\mathbb{P}(\mathcal{E}_2^c)=0$, overall we have
\[
\begin{split}
\mathbb{P}\left(T_1=0\right) &= \mathbb{P}\Bigg\{\lim_{v_r+\sigma^2_w\to0} \mathbb{I}(\mathcal{E}_1)\cdot\bigg[\frac{1}{v_r}\left( Z-\mathbb{E}[Z|Y_\sigma,Z_r] \right)^2\\
&\ -\Big( \frac{\sigma^2_wR}{v_r+\sigma^2_w} -\frac{\sqrt{v}\sigma_w W}{v_r+\sigma^2_w}  \Big)^2 \bigg]=0\Bigg\}\\
&=1.
\end{split}
\]
Hence, $\mathbb{E}[T_1]=0$.\vspace{5pt}

\textit{Analysis of $T_2$:} Let $(z,n,w,r,y,z_r)$ be an instance of $(Z,N,W,R,Y_\sigma,Z_r)$. From \eqref{Eqn:App_CondMean}, we have
\BS
\begin{align}
&\frac{1}{v_r}\left( z-\mathbb{E}[Z|y,z_r] \right)^2\\
&=\left( \sqrt{v_r}z - \sqrt{1-v_r}n- \frac{\sqrt{v_r}}{{v_r+\sigma^2_w}}\frac{\int_{\mathcal{I}} u \mathcal{N}(u;0,v_r+\sigma^2_w)\mr{d} u}{\int_{\mathcal{I}} \mathcal{N}(u;0,v_r+\sigma^2_w)\mr{d} u}\right)^2\\
&=\Big(\sqrt{v_r}z - \sqrt{1-v_r}n- \frac{\sqrt{v_r}}{{v_r+\sigma^2_w}}\frac{\int_{\hat{\mathcal{I}}} u \mathcal{N}(u;0,1)\mr{d} u}{\int_{\hat{\mathcal{I}}} \mathcal{N}(u;0,1)\mr{d} u}\Big)^2 \label{Eqn:noisy_MMSE_continu}
\end{align}
\ES
where $\hat{\mathcal{I}}\Mydef \frac{\mathcal{I}}{\sqrt{v_r+\sigma^2_w}}=\frac{f^{-1}(y)-z_r}{\sqrt{v_r+\sigma^2_w}}$.
Let $\mathcal{E}_3$ be the event that $z_r$ is not on the boundary of $f^{-1}(y)$. Consider the third term in \eqref{Eqn:noisy_MMSE_continu} under  $\mathcal{E}_1^c\cap \mathcal{E}_3$. From the definition of $\mathcal{E}_1^c$, $\hat{\mathcal{I}}$ only consists of intervals. If $0$ is an interior point of $\hat{\mathcal{I}}$, we have 
\[
\left| \int_{\hat{\mathcal{I}}} \mathcal{N}(u;0,1)\mr{d} u-1 \right|\le  \int_{\hat{\mathcal{I}}^c} \mathcal{N}(u;0,1)\mr{d} u = O\left(e^{-c/(v_r+\sigma^2_w))}\right).
\]
where $\hat{\mathcal{I}}^c=\mathbb{R}\backslash \hat{\mathcal{I}}$ and $c>0$ is some constant. Similarly, for the numerator,
\[
\left| \int_{\hat{\mathcal{I}}} u\mathcal{N}(u;0,1)\mr{d} u \right|\le  \int_{\hat{\mathcal{I}}^c} |u|\mathcal{N}(u;0,1)\mr{d} u = O\left(e^{-c/(v_r+\sigma^2_w))}\right).
\]
Hence, when $0$ is an interior point of $\hat{\mathcal{I}}$, we have
\[
\lim_{v_r+\sigma^2_w\to0}\frac{\sqrt{v_r}}{{v_r+\sigma^2_w}}\frac{\int_{\hat{\mathcal{I}}} u \mathcal{N}(u;0,1)\mr{d} u}{\int_{\hat{\mathcal{I}}} \mathcal{N}(u;0,1)\mr{d} u}=0.
\]
Next, we decompose $S\Mydef \left( z-\mathbb{E}[Z|y,z_r] \right)^2/v_r$ as
\[
S = S\cdot \mathbb{I}\left( 0\in (f^{-1}(y) - z_r) \right)+S\cdot\mathbb{I}\left( 0\notin (f^{-1}(y) - z_r) \right).
\]
We note that as $v_r+\sigma^2_w\to0$, we have $z_r\to z$. Further, $z\in f^{-1}(y)$. Therefore,
\[
\lim_{v_r+\sigma^2_w\to0} \mathbb{I}\left( 0\in (f^{-1}(y) - z_r) \right) = 1.
\]
We have shown in \eqref{Eqn:noisy_DCT} that $S<\infty$. Hence,
\[
\begin{split}
&\lim_{v_r+\sigma^2_w\to0} \frac{1}{v_r}\left( z-\mathbb{E}[Z|y,z_r] \right)^2\\
&=\lim_{v_r+\sigma^2_w\to0}S\cdot \mathbb{I}\left( 0\in (f^{-1}(y) - z_r) \right)\\
&\quad +\lim_{v_r+\sigma^2_w\to0}S\cdot \mathbb{I}\left( 0\notin (f^{-1}(y) - z_r) \right)\\
&=\lim_{v_r+\sigma^2_w\to0}S\cdot \mathbb{I}\left( 0\in (f^{-1}(y) - z_r) \right)\\
&= \left( \sqrt{v_r}z - \sqrt{1-v_r}n \right)^2.
\end{split}
\]
Since $\mathbb{P}(\mathcal{E}_3^c)=0$, we have
\[
\begin{split}
&\mathbb{P}\left(T_2=0\right) \\
&= \mathbb{P}\Bigg(\lim_{v_r+\sigma^2_w\to0} \frac{1}{v_r}\left( z-\mathbb{E}[Z|y,z_r] \right)^2\mathbb{I}(\mathcal{E}_2)\\
&\quad =\frac{1}{v_r}\left( v_r z - \sqrt{v_r(1-v_r)}n\right)^2\mathbb{I}(\mathcal{E}_2)\Bigg)\\
&=1.
\end{split}
\]
\end{IEEEproof}

\vspace{5pt}

\begin{lemma}\label{Lem:Aux1}
Suppose $\sigma_w^2\neq 0$ and $\delta>\delta_\Lambda^\mathsf{alg}\ge 1/d(Y)$. Define $v^\diamond\Mydef \inf\{v\in[0,1]:\ g(v_r)=0\}$. For arbitrary $\epsilon\in(0,v^\diamond)$, define \
\BE\label{Eqn:v_dia_def_formal}
v_\epsilon^\diamond(\sigma^2_w)\Mydef \sup\bigg\{ v\in(0,\epsilon):\, \mmse_z(v_r,\sigma^2_w)=\left(1-\frac{1}{\delta}\right)v_r \bigg\},
\EE
where $\mmse_z(v_r,\sigma_w^2)$ is defined in \eqref{Eqn:MMSE_noisy}.
Then, the following holds as $\sigma^2_w\to0$
\BE
v_\epsilon^\diamond(\sigma^2_w)\le C(\delta,f)\cdot \sigma^2_w,
\EE
where $0< C(\delta,f)<\infty$ is a constant depending on $\delta$ and $f$.
\end{lemma}
\begin{IEEEproof}
Our proof is mainly concerned with proving the following upper bound of $\mmse_z(v_r,\sigma^2_w)$ as $v_r+\sigma^2_w\to0$:
\BE\label{Eqn:App_mmse_bound}
\mmse_z(v_r,\sigma^2_w)\le v_r\cdot \mathscr{D}(Z|Y)+o(v_r)+C\cdot \sigma^2_w,
\EE
where $C$ is some constant depending on $\delta$ and $f$. Using this, we can upper bound $v^\diamond_\epsilon(\sigma^2_w)$ by the solution to the solution to the following equation:
\BE\label{Eqn:v_diag_upp}
v_r\cdot \mathscr{D}(Z|Y)+o(v_r)+C\cdot \sigma^2_w = \left( 1-\frac{1}{\delta} \right) v_r.
\EE
Namely,
\[
v_\epsilon^\diamond(\sigma^2_w) \le \frac{C\sigma^2_w}{ 1-\frac{1}{\delta}- \mathscr{D}(Z|Y) - o(1)},
\]
which yields the desired result.\vspace{10pt}

The rest of this section is devoted to the proof of \eqref{Eqn:App_mmse_bound}. Lemma \ref{Lem:aux_noisy2} shows that the following holds
\[
\lim_{v_r+\sigma^2_w\to0} \frac{1}{v_r}\left( \mmse(v_r,\sigma^2_w)-\mmse_{\mr{app}}(v_r,\sigma^2_w)  \right)=0.
\]
As a consequence,
\BE\label{Eqn:mmse_app_Err}
\mmse_z(v_r,\sigma^2_w) = \mmse_{\mr{app}}(v_r,\sigma^2_w)  + o(1)\cdot v_r.
\EE
In what follows, we prove that the following holds for all $v_r\in(0,1)$
\BE\label{Eqn:mmse_app_noiseErr}
\begin{split}
\mmse_{\mr{app}}(v_r,\sigma^2_w) &= \mmse_{\mr{app}}(v_r,0) +O(\sigma^2_w).
\end{split}
\EE
We first recall that $\mmse_{\mr{app}}$ is defined as
\BE\label{Eqn:mmse_app2}
\begin{split}
&\mmse_{\mr{app}}(v_r,\sigma^2_w) \\
 \Mydef  & \underbrace{v_r \mathbb{E}\left( \left( \frac{\sigma^2_wR}{v_r+\sigma^2_w} -\frac{\sqrt{v}\sigma_w W}{v_r+\sigma^2_w} \right)^2 \mathbb{I}(\mathcal{E}_1) \right)}_{\text{Part 1}}\\
 &\quad +\underbrace{v_r\mathbb{E}\left( R^2 \mathbb{I}(\mathcal{E}_1^c)\right)}_{\text{Part 2}}
\end{split}
\EE
Clearly, Part one is $O(\sigma^2_w)$. We next show that the difference between Part two and $\mmse_{\mr{app}}(v_r,0)$ is $O(\sigma^2_w)$. To this end, notice that $R$ is correlated with $U_\sigma$, and it is convenient to decompose it as
\[
R=\frac{\sqrt{v_r}}{1+\sigma^2_w} U_\sigma + \sqrt{\frac{1-v_r+\sigma_w^2}{  1 + \sigma^2_w } } S,
\] 
where $S\sim\mathcal{N}(0,1)$ and $S\Perp U_\sigma$. Then,
\[
\begin{split}
\text{Part 2}&=v_r\,\mathbb{E}\left( R^2 \mathbb{I}(\mathcal{E}_1^c)\right)\\
&=v_r\,\mathbb{E}\left( \left(\frac{\sqrt{v_r}}{1+\sigma^2_w} U_\sigma + \sqrt{\frac{1-v_r+\sigma_w^2}{  1 + \sigma^2_w } } S\right)^2 \mathbb{I}(\mathcal{E}_1^c)\right)\\
&=\frac{v_r^2}{1+\sigma_w^2}\mathbb{E}\left( U_\sigma^2\mathbb{I}(\mathcal{E}_1^c)\right)+{\frac{v_r(1-v_r+\sigma_w^2)}{  1 + \sigma^2_w } }\cdot\mathbb{P}(\mathcal{E}_1^c)
\end{split}
\]
We notice the following facts: (i) $U_\sigma = Z+\sigma_w W$; (ii) $\mathcal{E}_1^c=\mathbb{I}(U_\sigma \in \{x: f^{-1}(f(x)) \text{ is an interval}\})$.
It can be shown that there exists a constant $C<\infty$ such that the following hold for all $v_r\in(0,1)$
\[
\begin{split}
\mathbb{E}\left( U_\sigma^2\mathbb{I}(\mathcal{E}_1^c)\right) &\le \left.\mathbb{E}\left( U_\sigma^2\mathbb{I}(\mathcal{E}_1^c)\right)\right|_{\sigma_w=0} + C\cdot \sigma_w^2,\\
\mathbb{P}(\mathcal{E}_1^c) &\le \left.\mathbb{P}(\mathcal{E}_1^c)\right|_{\sigma_w=0} + C\cdot \sigma_w^2,
\end{split}
\]
as $\sigma^2_w\to0$.
We skip the details here. Combining the above arguments proves \eqref{Eqn:mmse_app_noiseErr}.

Finally, combining \eqref{Eqn:mmse_app_Err} and \eqref{Eqn:mmse_app_noiseErr}, we have 
\[
\mmse_z(v_r,\sigma^2_w)= \mmse_{\mr{app}}(v_r,\sigma_w=0)+o(1)v_r+O(\sigma^2_w),
\]
as $v_r+\sigma^2_w\to0$. Notice that
\[
\begin{split}
 \mmse_{\mr{app}}(v_r,\sigma_w=0) &= \mathbb{E}\left( v_r Z-\sqrt{v_r(1-v_r)}N \right)^2\mathbb{I}(\mathcal{E}_1^c)
 \end{split}
\]
where without slight abuse of notation $\mathcal{E}_1^c=\mathbb{I}(Z \in \{x: f^{-1}(f(x)) \text{ is an interval}\})$ (namely, it is the previous defined $\mathcal{E}_1^c$ at $\sigma_w=0$). 
This term has the same behavior as $\mmse_z(v_r)$ for small $v_r$. Here, the $O(v_r)$ term is
\[
v_r(1-v_r)\cdot \mathbb{E}[N^2 \mathbb{I}(\mathcal{E}_1^c)]=v_r(1-v_r)\cdot \mathscr{D}(Z|Y).
\]
Hence, overall we have
\[
\mmse_z(v_r,\sigma^2_w)\le v_r\cdot \mathscr{D}(Z|Y)+o(v_r)+C\cdot \sigma^2_w,
\]
as $v_r+\sigma^2_w\to0$.
\end{IEEEproof}

\vspace{5pt}

\subsection{Main Proof for Lemma \ref{Lem:noise_sensitivity}}\label{Sec:Noise_main}
Let $g(v_r,\sigma^2_w)$ and $P(v_r,\sigma^2_w)$ be the noisy counterparts of $g(v_r)$ and $P(v_r)$, respectively:
\BS
\begin{align}
g(v_r,\sigma_w^2)\Mydef 1-\delta\left(1-\frac{\mmse_z(v_r,\sigma^2_w)}{v_r}\right)\label{Eqn:App_g_def_noisy},\\
P(v_r,\sigma^2_w)\Mydef \mathbb{E}\left[ \frac{\phi(v_r,\sigma_w^2)}{\phi(v_r,\sigma_w^2) + \Lambda} \right] - g(v_r,\sigma^2_w),\label{Eqn:App_P_noisy}
\end{align}
\ES
where $\mmse_z(v_r,\sigma_w^2)$ and $\phi(v_r,\sigma_w^2)$ are defined in \eqref{Eqn:MMSE_noisy} and \eqref{Eqn:phi_noisy_app}, respectively.

The behaviors of $g$ and $P$ around $v_r=0$ are different under the noiseless and noisy settings. Specifically,
\[
\lim_{v_r\to0}g(v_r,\sigma_w^2)=
\begin{cases}
1-\delta\cdot d(Y)<0 & \text{ if $\sigma^2_w=0$},\\
1 &\text{ if $\sigma^2_w\neq0$},
\end{cases}
\]
which is from the definition of conditional MMSE dimension and the fact that the distribution $P_{Z|Y_\sigma}$ (where $Y_\sigma\Mydef f(Z+\sigma_w W)$) is absolutely continuous when $\sigma^2_w\neq0$. Further, from Lemma \ref{Lem:phi_noisy}, 
\[
0<\phi(0,\sigma^2_w)\Mydef \lim_{v_r\to0}\phi(v_r,\sigma^2_w)<\infty.
\]
Hence,
\[
\lim_{v_r\to0}P(v_r,\sigma_w^2)=
\begin{cases}
\delta \cdot d(Y)-1>0 & \text{ if $\sigma^2_w=0$},\\
\mathbb{E}\left[\frac{\phi(0,\sigma^2_w)}{\phi(0,\sigma^2_w)+\Lambda}\right]-1<0 &\text{ if $\sigma^2_w\neq0$},
\end{cases}
\]

Since $g(0)<0$ (where $g(v_r)$ is a shorthand for $g(v_r,0)$), there exists a neighbor of $v_r=0$ for which $g(v_r)<0$. Define 
\BE\label{Eqn:v_diam_noiseless_def}
v^\diamond\Mydef \inf\{v\in[0,1]:\ g(v_r)=0\}.
\EE
If $g(v_r)>0$ for all $v_r\in[0,1]$, we set $v^\diamond=1$.
Note that $P$ and $g$ are continuous functions of $\sigma^2_w\ge0$ whenever $v_r\neq0$. Let $\epsilon\in(0,v^\diamond)$ be an arbitrary constant. By continuity, for sufficiently small $\sigma^2_w$, we have
\BE\label{Eqn:noisy_P_g_tmp}
P(v_r,\sigma^2_w)>0,\quad \forall v_r\in(\epsilon,1),
\EE
and
\BE\label{Eqn:noisy_P_g_tmp2}
 g(\epsilon,\sigma^2_w) < 0.
\EE
Since $g(0,\sigma^2_w)=1$, $g(v_r,\sigma^2_w)=0$ has at least one solution in $v_r\in(0,\epsilon)$. Let $v_\epsilon^\diamond(\sigma^2_w)$ be the largest one, i.e., 
\BE\label{Eqn:App_v_diam}
v_\epsilon^\diamond(\sigma^2_w)\Mydef \sup\left\{ v\in(0,\epsilon):\, g(v_r,\sigma^2_w)=0 \right\}.
\EE
This definition ensures $g(v_r,\sigma^2_w)<0,\, \forall v_r\in(v_\epsilon^\diamond(\sigma^2_w),\epsilon)$ (see \eqref{Eqn:noisy_P_g_tmp2}). This further ensures $P(v_r,\sigma^2_w)>0$ for $v_r\in(v_\epsilon(\sigma^2_w),\epsilon)$, since the first term in \eqref{Eqn:App_P_noisy} is positive. Together with \eqref{Eqn:noisy_P_g_tmp}, we have
\BE\label{Eqn:noisy_p_tmp}
P(v_r,\sigma^2_w)>0\quad \forall v_r\in(v_\epsilon^\diamond(\sigma^2_w),1).
\EE

Now, let us define
\BE\label{Eqn:v_r_noisy_def}
v_r^\star(\sigma^2_w)=\sup\left\{v\in[0,1]:P(v_r,\sigma^2_w)=0\right\},
\EE
which is the fixed point reached by the state evolution.
As a consequence of \eqref{Eqn:noisy_p_tmp} and \eqref{Eqn:v_r_noisy_def}, we have (for small enough $\sigma_w^2$)
\[
v_r^\star(\sigma^2_w)\le v_\epsilon^\diamond(\sigma^2_w).
\]
By the monotonicity of $\phi(v_r,\sigma^2_w)$ with respect to $v_r$ (see Lemma \ref{Lem:phi_noisy}), we have the following for small $\sigma^2_w$
\BE\label{Eqn:App_phi_upp}
\begin{split}
\phi(v_r^\star(\sigma^2_w),\sigma^2_w)&\le \phi(v_\epsilon^\diamond(\sigma^2_w),\sigma^2_w)\\
&\overset{(a)}{=}(\delta-1)\cdot v_\epsilon^\diamond(\sigma^2_w)\\
&\overset{(b)}{\le }(\delta-1)\cdot C(\delta,f)\cdot\sigma^2_w,
\end{split}
\EE
where step (a) follows from \eqref{Eqn:App_g_def_noisy} and the fact that $v_\epsilon^\diamond(\sigma^2_w)$ is a solution to $g(v_r,\sigma^2_w)=0$, and step (b) is due to Lemma \ref{Lem:Aux1}. Together with Lemma \ref{Lem:phi_noisy}, we finally have
\BE\label{Eqn:v_star_bound}
\sigma^2_w\le\phi(v^\star(\sigma^2_w),\sigma^2_w)\le (\delta-1)\cdot C(\delta,f)\cdot\sigma^2_w.
\EE

Finally, for small $\sigma^2_w$, the MSE is given by
\[
\begin{split}
\mathsf{MSE}_\Lambda^{\star}(\sigma^2_w,\Lambda) &= \mathbb{E}\left[ \frac{\phi(v_r^{\star}(\sigma^2_w),\sigma_w^2)}{\phi(v_r^{\star}(\sigma^2_w),\sigma_w^2)+\Lambda}\right]\\
&=\phi\left(v^\star(\sigma^2_w),\sigma^2_w\right)\cdot \left(\mathbb{E}[\Lambda^{-1}]+o(1)\right).
\end{split}
\]
From \eqref{Eqn:v_star_bound}, we have
\[
\begin{split}
\sigma^2_w\cdot \left(\mathbb{E}[\Lambda^{-1}]+o(1)\right) &\le \mathsf{MSE}^{\star}_\Lambda(\sigma^2_w,\Lambda)\\
&\le (\delta-1) C(\delta,f)\sigma^2_w\left(\mathbb{E}[\Lambda^{-1}]+o(1)\right).
\end{split}
\]
This completes our proof of Lemma \ref{Lem:noise_sensitivity}.

{

\bibliographystyle{IEEEtran}    
\bibliography{IEEEabrv,Phase_retrieval3}         
}



\begin{IEEEbiographynophoto}{Junjie Ma}
received the B.E. degree from Xidian University, Xi’an, China, in 2010, and the Ph.D. degree from the City University of Hong Kong, Hong Kong, in 2015. He was a Postdoctoral Researcher with City University of Hong Kong, from 2015 to 2016, and with Columbia University from 2016 to 2019, and with Harvard University from 2019 to 2020. He has been an assistant professor at the Institute of Computational Mathematics and Scientific/Engineering Computing, AMSS, Chinese Academy of Sciences, since July 2020. His current research interests include message passing algorithms and their applications for high dimensional signal processing.
\end{IEEEbiographynophoto}

\begin{IEEEbiographynophoto}{Ji Xu}
holds a Ph.D. degree in Computer Science department
at Columbia University. Previously, he holds a M.A. in Statistics from
Columbia University, and a B.S. in Math with minor in Economics from
Peking University. He received Kathryn and Shelby Cullom Davis International Fellowship in 2015 and Inaugural Cheung-Kong Graduate School
of Business (CKGSB) Fellowship in 2018 and 2019. His research interests are in
algorithmic statistics, machine learning and high dimensional statistics. His
works have been accepted in top tier conferences including oral presentation
at NIPS 2016 and long talk presentation at ICML 2018.
\end{IEEEbiographynophoto}

\begin{IEEEbiographynophoto}{Arian Maleki}
is an associate professor in the Department of Statistics at
Columbia University. He received Ph.D. from Stanford University in 2010.
Before joining Columbia University, he was a postdoctoral scholar in the
department of Electrical and Computer Engineering at Rice University
\end{IEEEbiographynophoto}


\end{document}